\documentclass[12pt]{article}
\usepackage[margin=1in]{geometry} 
\usepackage{amsmath,amsthm,amssymb,hyperref}
\numberwithin{equation}{section}
\usepackage{mdframed,tcolorbox}
\usepackage{float}
\usepackage{graphicx,psfrag,epsf, subcaption, enumitem}
\usepackage{mathrsfs}
\usepackage[sort,round,authoryear]{natbib}
\usepackage{setspace}
\usepackage{algorithm}                  
\usepackage{algorithmicx}               
\usepackage{algpseudocode} \usepackage{authblk}      

\newtheorem{proposition}{Proposition}

\newtheorem{definition}{Definition}
\newtheorem{lemma}{Lemma}
\newtheorem{theorem}{Theorem}

\newtheorem{assumption}{Assumption}

\title{Integral-Operator-Based  Spectral Algorithms  for Goodness-of-Fit Tests}
\author[1]{Shiwei Sang}
\author[2]{Shao-Bo Lin\thanks{Corresponding author: sblin1983@gmail.com}}
\author[1]{Xuehu Zhu\thanks{Corresponding author: zhuxuehu@xjtu.edu.cn}}

\affil[1]{School of Mathematics and Statistics, Xi'an Jiaotong University, Xi'an, China}
\affil[2]{Center for Intelligent Decision-Making and Machine Learning, School of Management, Xi'an Jiaotong University, Xi'an, China}

\date{} 
\date{}
\begin{document}
\maketitle
\begin{abstract}
 The widespread adoption of the \emph{maximum mean discrepancy} (MMD) in goodness-of-fit testing has spurred extensive research on its statistical performance. However, recent studies indicate that the inherent structure of MMD may constrain its ability to distinguish between distributions, leaving room for improvement. Regularization techniques have the potential to overcome this limitation by refining the discrepancy measure.
In this paper, we introduce a family of regularized kernel-based discrepancy measures constructed via spectral filtering. Our framework can be regarded as a natural generalization of prior studies, removing restrictive assumptions on both kernel functions and filter functions, thereby broadening the methodological scope and the theoretical inclusiveness. We establish non-asymptotic guarantees showing that the resulting tests achieve valid Type~I error control and enhanced power performance. 
Numerical experiments are conducted to demonstrate the broader generality and competitive performance of the proposed tests compared with existing methods.
\end{abstract}


\section{Introduction}

Statistical hypothesis testing has long been a cornerstone of statistical inference, providing a rigorous framework for making data-driven decisions across various scientific and practical fields. A classical problem in hypothesis testing involves determining whether a set of observations originates from a fixed, given distribution. This type of testing, known as \emph{goodness-of-fit} (GOF) testing, serves as a fundamental tool in many contemporary applications. In healthcare, GOF tests are employed to assess the calibration of personal risk models, ensuring that predicted risks align with observed outcomes~\citep{gong2014assessing}. In finance, these tests are utilized to evaluate whether financial models accurately represent market behaviors, aiding in the development of robust investment strategies~\citep{frezza2014goodness, ritchey1986application}. In psychology and social sciences, GOF tests are applied to structural equation modeling, verifying that theoretical models adequately fit empirical data~\citep{schermelleh2003evaluating}. GOF tests enable researchers and practitioners to evaluate how well a statistical model aligns with observed data, which is essential for making informed decisions, developing effective interventions, and guiding policy-making across various disciplines.

In the classical literature, GOF tests have been broadly categorized into parametric and non-parametric frameworks. Parametric tests rely on explicit distributional assumptions, such as normality or membership in exponential families. Representative examples include the likelihood ratio test and Student’s $t$-test, which provide efficient inference under correctly specified models. However, parametric methods can be unreliable when the distributional assumption is misspecified, limiting their applicability in real-world applications. To address this limitation, non-parametric GOF tests have been developed that impose fewer assumptions and instead quantify discrepancies between distributions through distance-based statistics. Classical examples include the \emph{Kolmogorov–Smirnov} (KS) test~\citep{massey1951kolmogorov} and the \emph{Cram{\'e}r–von Mises} (CVM) test~\citep{cramer1928composition}, among many others, which have been widely adopted and are  supported by strong theoretical guarantees.

Despite their success, traditional GOF tests face fundamental challenges in contemporary data analysis. First, in terms of data types, modern real-world datasets increasingly feature high-dimensional, heterogeneous, and complex structures--such as networks, time series, and functional data among many others. Classical methods often struggle to capture distributional discrepancies for such complex data. For instance, Hotelling’s $t$-test fails to detect mean differences when both the sample size and the data dimension grow simultaneously~\citep{dempster1958high}.
Similarly, the \emph{Jarque-Bera} (JB) test for normality, which relies on skewness and kurtosis under the assumption of independent and identically distributed (i.i.d.) observations, becomes unreliable in the presence of unconditional heteroscedasticity~\citep{raissi2018testing}. Then, in terms of decision-making, most GOF tests rely heavily on asymptotic theory, which assumes access to infinitely many samples for validity. In practice, however, data acquisition is constrained by privacy regulations, transmission costs, and storage limitations, making such large-sample guarantees inapplicable and potentially misleading in finite-sample settings. Moreover, in terms of research paradigms, traditional studies emphasis on Type~I error control and asymptotic consistency, but devote relatively little attention to finite-sample performance. This 
mismatch between theoretical assurances and practical needs means that classical guarantees can fail to provide precise and sufficient statistical guidance in real-world applications. These limitations have motivated the development of new GOF testing methods, along with more precise non-asymptotic evaluation protocols, to facilitate distributional comparisons over general domains.
    
Kernel methods have recently emerged as effective alternatives to classical tests due to their well-established finite-sample guarantees and natural applicability to distributions comparison over general domains. A notable development in the literature is the \emph{maximum mean discrepancy} (MMD)~\citep{gretton2007kernel,gretton2012kernel,smola2007hilbert}, which quantifies differences between distributions through their mean embeddings in \emph{reproducing kernel Hilbert spaces} (RKHSs) and has become a widely used tool for nonparametric hypothesis testing.  With the empirical success and growing popularity of MMD-based tests, a large body of research~\citep{li2019optimality,balasubramanian2021optimality,hagrass-two-sample-test,hagrass-gof-test,schrab2023mmd,fromont2013two} has sought to provide a deeper theoretical understanding of their behavior, particularly in terms of statistical \emph{power}--the efficiency with which MMD distinguishes between two distributions. 
    
A recent work \citep{balasubramanian2021optimality} observed that test statistics constructed from MMD may suffer from low power due to its particular structure, and regularization strategies have the potential to improve the power of kernel-based tests. Building on this insight, \citet{balasubramanian2021optimality} introduced a Tikhonov-regularized variant of MMD to enhance testing performance. However, their theoretical guarantees impose strong restrictions on kernel functions, and the resulting procedure requires knowledge of the kernel and underlying distribution that is often unavailable in practice. Moreover, the well-known \emph{saturation} phenomenon of Tikhonov regularization also manifests in their testing framework: the power reaches a ceiling and cannot be further improved, even under stronger regularity conditions. Although spectral filtering techniques have long been employed to enhance numerical stability and mitigate saturation in regression and inverse problems~\citep{bauer2007regularization,engl1996regularization,gerfo2008spectral,lin2017distributed,lin2020distributed}, their potential in hypothesis testing remains underexplored~\citep{hagrass-two-sample-test,hagrass-gof-test}.
A deeper understanding of spectral regularization within this framework, particularly its influence on the power of kernel-based tests, is therefore of both theoretical and practical importance.

Motivated by advances in spectral algorithms~\citep{guo2017learning,engl1996regularization,bauer2007regularization,gerfo2008spectral} and their recent applications in nonparametric hypothesis testing~\citep{balasubramanian2021optimality, hagrass-two-sample-test,hagrass-gof-test}, this present work introduces a broad class of regularized kernel GOF tests based on spectral filtering. We demonstrate that with appropriate regularization on the spectrum of the kernel operator, discrepancies between distributions can be more effectively captured, leading to more powerful tests. This proposed framework improves upon MMD-based methods while addressing the limitations identified in~\citep{balasubramanian2021optimality,hagrass-two-sample-test,hagrass-gof-test}. Our main contributions are summarized below.

 $\bullet $ From a methodological standpoint, we recast the statistic construction proposed by \citet{balasubramanian2021optimality} within a spectral-filtering framework: we replace the Tikhonov filter with a broad class of admissible spectral filters, yielding a natural generalization that relaxes several restrictive assumptions and mitigates the saturation inherent in Tikhonov regularization. Unlike the regularization strategies in \citep{hagrass-two-sample-test,hagrass-gof-test}, which are built on the \emph{centered covariance operator}, our approach is grounded in the \emph{integral operator}.  This operator-level distinction leads to substantially different \emph{bias} and \emph{variance}
properties of the resulting statistics. As a byproduct, the framework accommodates a wider family of filters--including, in particular, the spectral cut-off--thereby providing a more general and flexible extension of existing methods.

$\bullet$ From the theoretical consideration, we establish rigorous finite-sample guarantees for the proposed test, ensuring both valid Type~I error control and enhanced power properties. Technically, leveraging a novel error decomposition and recently developed integral operator approach for spectral algorithms~\citep{guo2017learning}, we first derive the estimation error between the proposed statistic and its approximated probability metric, while removing the additional kernel restrictions required in  \citep{balasubramanian2021optimality} and filter-specific constraints imposed in~\citep{hagrass-two-sample-test,hagrass-gof-test}. By combining a classical bias--variance analysis with our general reduction that translates estimation-error bounds into \emph{detection boundaries} (i.e., the smallest signal strength that a test can reliably detect), we then derive the detection boundary of the proposed tests. The theoretical results of our proposed statistic over specific distribution classes match the state-of-the-art results established in prior work.

 $\bullet$ From an empirical perspective, we complement our theoretical analysis with numerical experiments that support the finite-sample guarantees and demonstrate the performance improvements enabled by spectral regularization. At the same time, our approach accommodates a broader and more flexible class of spectral filters. In practice, the power of our proposed tests is generally comparable to, and in some cases surpasses, existing methods in the literature.
    
The remainder of this paper is organized as follows. Section~\ref{sec: GOF testing framwork} introduces the background on goodness-of-fit testing and the general framework for evaluating test performance.  In Section~\ref{sec:  spectral regularized GOF test}, we review recent developments of MMD-based tests with regularization and propose a broad class of kernel GOF tests based on spectral filtering. The associated theoretical guarantees are established in Section~\ref{sec: theory}, while Section~\ref{sec: numerical experiment} presents numerical experiments illustrating the empirical performance of the proposed methods. Section~\ref{sec: error decomposition} provides an error decomposition for the proposed statistic. Proofs of the main results are given in Section~\ref{sec: proofs}, and the testing procedures, additional proofs, and technical lemmas are deferred to Appendices~\ref{section: appendix A}--\ref{section: appendix C}.

\section{Analysis Framework for Goodness-of-Fit Tests with Finite Samples}
\label{sec:  GOF testing framwork}
Let $x^n:=\{x_i\}_{i=1}^n$ be a set of i.i.d. samples drawn from an unknown probability distribution $P$ on a measurable space $(\mathcal{X},\mathcal{B})$ and $P_0$ be a known distribution on $(\mathcal{X},\mathcal{B})$. The goal of the \emph{goodness-of-fit} (GOF) test is to deduce a rule based on $x^n$ to decide between the following two hypotheses
 \begin{equation}
     \label{GOF test}
     H_0: P = P_0 \quad \textrm{versus}\quad H_1: P \neq P_0,
 \end{equation}
where $H_0$ and $H_1$ are referred to as the \emph{null hypothesis} and the \emph{alternative hypothesis}, respectively.
A preferable rule is an indicator function $\phi$ of $x^n$ with $\phi(x^n) = 0$ the acceptance of $H_0$ (or rejection of $H_1$) and $\phi(x^n) = 1$ the rejection of $H_0$.

The quality of $\phi$ is generally measured by  two types of errors defined by 
$$
\begin{aligned}
	&e_{n}^{(\mathrm{I})}(\phi ):= \mathbb{P} \left\{\phi\left(x^n\right)  = 1\right\},\quad \mathrm{under}\,\, H_0;\qquad e_{n}^{(\mathrm{II})}(\phi):= \mathbb{P} \left\{\phi\left(x^n\right)  = 0\right\},\quad \mathrm{under}\,\, H_1,\end{aligned}
$$
where $e_{n}^{\mathrm{I}}(\cdot)$ is the \emph{Type I error}, the probability of rejecting $H_0$ when it is true, and $e_{n}^{\mathrm{II}}(\cdot)$ is the \emph{Type II error}, the probability of failing to reject $H_0$ when $H_1$ holds.  In the Neyman--Pearson's framework~\citep{lehmann2008testing}, the primary objective is to control the Type~I error at a pre-specified level and then to minimize the Type~II error as much as possible, which naturally leads to the following definition of the significance level.

\begin{definition}[Significance level]
    \label{testing level}
     A test $\phi_{\alpha}$ is said to be of (significance) level $\alpha$ if the Type I error does not exceed $\alpha$, i.e., $e_n^{(\mathrm{I})}(\phi_{\alpha}) \leq \alpha$.
\end{definition}



\begin{figure}
			\centering
			\includegraphics[width=0.5\linewidth]{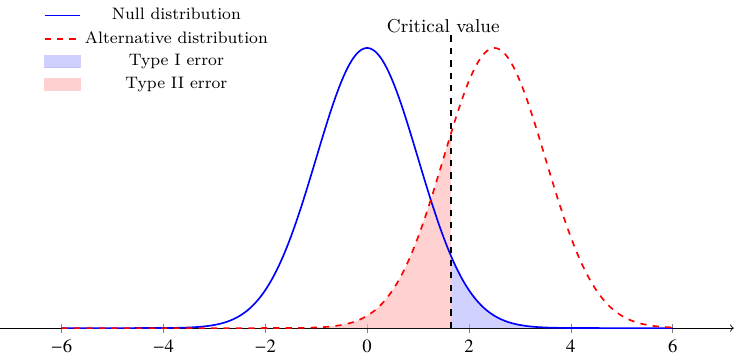}
          \caption{\footnotesize {\textit{Distributions of the test statistic under $H_0$ and $H_1$, and the trade-off between two types of errors.}}}
			\label{fig:two types of errors}
		\end{figure}

 It is important to note that Type~I and Type~II errors cannot, in general, be minimized simultaneously.
Figure~\ref{fig:two types of errors} illustrates this trade-off by showing how the $\alpha$-determined critical value influences both types of errors, as reducing one typically increases the other. In the asymptotic setting, the significance level~$\alpha$ is typically fixed in advance, making this trade-off less pronounced: traditional methods require the so-called \emph{consistency}, that is, the Type~II error vanishes for any fixed alternative as the sample size grows, so that the role of~$\alpha$ becomes relatively minor. However, in finite-sample settings--especially when the sample size is small--this balance becomes much more delicate. A more refined non-asymptotic analysis is therefore needed to characterize this relationship and to provide theoretical guidance on appropriate significance levels and tolerable Type~II errors in practice. 

In particular, we focus on a stronger criterion that evaluates the testing power \emph{ uniformly} over a class of alternatives based on finitely many samples. Given the null distribution $P_0$, let $\mathcal{C}$ denote a collection of distributions satisfying certain regularity conditions, 
$\rho(P,P_0)$ a measure of discrepancy between $P$ and $P_0$, and $\Delta>0$.
Define
\begin{equation}\label{distribution-class}
    \mathcal{P}(\mathcal{C}, \rho, \Delta) := 
    \left\{ P \in \mathcal{C} : \rho(P,P_0) \geq \Delta \right\},
\end{equation}
as the class of distributions in $\mathcal{C}$ that are at least $\Delta$ away from $P_0$ in terms of $\rho$. 
The \emph{detection boundary} of an $\alpha$-level test $\phi_{\alpha}$ based on the sample $x^n$ over $\mathcal{P}(\mathcal{C}, \rho, \Delta)$, 
as formalized in the following definition, is the minimal separation $\Delta := \Delta_n$ that can depend on the sample size $n$ such that the Type~II error can be uniformly controlled.

 \begin{definition}[Detection boundary and optimality] 
 \label{definition: detection boundary}
Let $\delta>0$ be a prescribed tolerance for the Type~II error. 
For an $\alpha$-level test $\phi_{\alpha}$ based on the sample $x^n$, and a distribution class $\mathcal{P}(\mathcal{C}, \rho, \Delta)$, the \emph{detection boundary} is defined as 
\[
\Delta_{n}(\phi_{\alpha}; \mathcal C,\rho,\delta)
:= \inf\left\{\Delta >0 : \sup_{P \in \mathcal{P}(\mathcal{C}, \rho, \Delta)} e_{n}^{(\mathrm{II})}(\phi_{\alpha}) \leq \delta \right\}.
\]    
Moreover, the \emph{optimal detection boundary} over the class $\Phi_{n,\alpha}$ of all $\alpha$-level tests based on the sample $x^n$ is given by 
\[
\Delta^{\ast}_n(\Phi_{n,\alpha};\mathcal C,\rho,\delta)
:= \inf_{\phi_{\alpha} \in \Phi_{n,\alpha}} \Delta_{n}(\phi_{\alpha}; \mathcal{C},\rho,\delta).
\]
\end{definition}
 
Besides the test $\phi_\alpha$, the detection boundary depends heavily on the discrepancy measure $\rho$, such as $\chi^2$-divergence,  $L^2$-distance between density functions, K-L divergence, and others.
 It is straightforward to see that different choices of $\rho$ require fundamentally different design philosophy of $\phi_\alpha$, which in turn induce distinct alternative distribution classes $\mathcal{P}(\mathcal{C}, \rho, \Delta)$, leading to different forms or rates of the detection boundary. For instance, \citet{balasubramanian2021optimality} considered  $\rho$ as the
 $\chi^2$-divergence  and studied the detection boundary over the alternative distribution class 
\begin{equation}
    \label{Yuanming: alternative space}
   {\widetilde{\mathcal{P}}}_{\theta, \Delta} := \left\{ P: \frac{dP}{dP_0} -1 \in \mathcal{F}(\theta; M),\ \chi^2(P,P_0) = \left\|\frac{dP}{dP_0} -1 \right\|_{L^2(P_0)}^2 \geq \Delta \right\},
\end{equation}
where
$$
\begin{aligned}
    \mathcal{F}(\theta;M):=&\left\{f\in L^{2}(P_{0}):\text{ for any }R>0,\exists f_{R}\in\mathcal{H}_K\text{ such that }\|f_{R}\|_{K}\leq R, \right.\\
    &\qquad\qquad\quad \left.\text{and}\left\|f-f_{R}\right\|_{L^{2}(P_{0})}\leq MR^{-1/\theta}\right\},
\end{aligned}
$$
 $\theta > 0$, and $\mathcal H_K$ is the reproducing kernel Hilbert space (RKHS) associated with a Mercer kernel $K$ satisfying $ \sup_{x}\sqrt{K(x,x)} \leq \kappa$ for some $\kappa >0$. Moreover, \citet{hagrass-gof-test} also employed the $\chi^2$-divergence to define the  alternative space  as
 \begin{equation}
 	\label{Hagrass: alternative space}
 \widetilde{\mathcal{P}}_{r,\Delta}:= \left\{P: \frac{dP}{dP_0} - 1\in \mathrm{Ran}\left(L_{\bar{K}}^{r}\right), \chi^2(P,P_0) = \left\|\frac{dP}{dP_0} - 1 \right\|_{L^2(P_0)}^2 \geq \Delta\right\},
 \end{equation}
 where $r>0$, the integral operator $L_{K}:L^2(P_0)\rightarrow L^2(P_0)$ (also $\mathcal H_{K}\rightarrow\mathcal H_{K}$ if no confusion is made) is defined by 
\begin{equation}\label{integral operator}
	L_{K}(f) := \int_{\mathcal{X}} f(x){K}_x d P_0(x),
\end{equation}
the operator $L_{\bar{K}}$ in \eqref{Hagrass: alternative space} is defined through the centered kernel function $\bar{K}(x,y):=\langle K_x-\mu_0,\,K_y-\mu_0\rangle_K$, $K_x=K(x,\cdot)$, $\mu_P=\int_{\mathcal X} K_xdP(x)$, $\mu_0=\mu_{P_0}$, and $\mathrm{Ran}(A)$ denotes the range of an operator $A$.
 
\subsection{General Framework: From Test Construction to Detection Boundaries}
\label{subsec:general-framework}
Given a prescribed probability metric $\rho$, one can construct a test from an empirical estimate of $\rho(P,P_0)$ for the GOF testing problem in~\eqref{GOF test}.
Specifically, let $\widehat{T}_{\lambda}(P,P_0)$ denote a regularized estimator of $\rho(P,P_0)$ based on observations from $P$, where $\lambda >0$ is a user-defined regularization parameter.
The overall construction of a GOF test can then be outlined in the following four steps:

\begin{itemize}
    \item \textbf{Significance level:} Fix a significance level $\alpha > 0$ to prescribe the desired Type~I error control. 

    \item \textbf{Test statistic:} Given the sample $x^n$ from $P$ and the known distribution $P_0$, construct a regularized statistic $\widehat{T}_{\lambda}(P,P_0)$ as an estimator of $\rho(P,P_0)$. 

    \item \textbf{Critical value:} Determine the threshold $\widehat{c}_{\alpha,\lambda,n}$ by appropriate calibration procedures depending on $\widehat{T}_{\lambda}$, $P_0$, and $x^n$ such that
    \[
        \mathbb{P}\!\left\{ \widehat{T}_{\lambda}(P_0,P_0) \geq  \widehat{c}_{\alpha,\lambda,n} \right\} \leq \alpha.
    \]

    \item \textbf{Decision rule:} Define the $\alpha$-level test as
    \begin{equation}
        \label{equ: a test}
        \phi_{\alpha}^{T_{\lambda}} := \mathbf{1}\!\left\{\widehat{T}_{\lambda}(P,P_0) \geq \widehat{c}_{\alpha,\lambda,n} \right\},
    \end{equation}
    where $\mathbf{1}\{\mathcal{A}\}$ denotes the indicator function of the event $\mathcal{A}$.
\end{itemize}
 The statistical properties of the test $\phi_{\alpha}^{T_{\lambda}}$ hinge on two key quantities: the critical value $\widehat{c}_{\alpha,\lambda,n}$ and the estimation error of $\widehat{T}_{\lambda}(P,P_0)$ for $\rho(P,P_0)$. 
The critical value not only directly determines the valid control of the Type~I error, but also affects the power performance in conjunction with the estimation error of $\widehat{T}_{\lambda}(P,P_0)$. 
In what follows, we develop a general framework illustrating how these two quantities can be translated into detection boundaries of the test in~\eqref{equ: a test} over the distribution class $\mathcal{P}(\mathcal{C},\rho,\Delta)$ in~\eqref{distribution-class}.

For any $0<\eta<1$, assume that
\begin{align}\label{equ: estimation error}
    \mathbb{P}\!\left[ \rho(P,P_0) - \widehat{T}_{\lambda}(P,P_0) 
    > \mathcal{U}_1\{n,\lambda,\eta,\rho(P,P_0)\} \right] \leq \frac{\eta}{2},
\end{align}
and
\begin{align}\label{equ: upper bound critical value}
    \mathbb{P}\!\left[ \widehat{c}_{\alpha,\lambda,n} 
    > \mathcal{U}_2\{n,\lambda,\alpha,\eta,\rho(P,P_0)\} \right] \leq \frac{\eta}{2}.
\end{align}
Here, $\mathcal{U}_1\{n,\lambda,\eta,\rho(P,P_0)\}>0$ is a deterministic quantity 
that characterizes the one-sided estimation error of $\widehat{T}_{\lambda}(P,P_0)$ for $\rho(P,P_0)$; 
condition~\eqref{equ: estimation error} guarantees that $\widehat{T}_{\lambda}(P,P_0)$ is unlikely to severely underestimate $\rho(P,P_0)$. 
Similarly, $\mathcal{U}_2\{n,\lambda,\alpha,\eta,\rho(P,P_0)\}>0$ is a deterministic bound providing, with high probability, an upper control on the critical value $\widehat{c}_{\alpha,\lambda,n}$. Building on these bounds, it is easy to derive  the following lemma to connect  $\mathcal{U}_1$ and $\mathcal{U}_2$ 
to the detection boundary.

\begin{lemma}  \label{lemma: general framework}    Assume that \eqref{equ: estimation error}, \eqref{equ: upper bound critical value} and     \begin{equation}       \label{equ: detection boundary inequality}   \rho(P,P_0) \;\geq\; \mathcal{U}_1\{n,\lambda,\delta,\rho(P,P_0)\} + \mathcal{U}_2\{n,\lambda,\alpha,\delta,\rho(P,P_0)\},    \end{equation}  hold for any $  P \in \mathcal{P}(\mathcal{C}, \rho, \Delta)$, then there holds
    \begin{equation}
        \label{equ: uniform type II error}
        \sup_{P\in\mathcal{P}(\mathcal{C},\rho,\Delta)}e_n^{(\mathrm{II})}\left(\phi_{\alpha}^{T_{\lambda}}\right)\leq \delta.
    \end{equation}
\end{lemma}
\begin{proof} 
For any $P \in \mathcal{P}(\mathcal{C},\rho,\Delta)$ satisfying \eqref{equ: detection boundary inequality}, there holds 
$$
\begin{aligned}
    &\quad \ \mathbb{P}\left\{\widehat{T}_{\lambda}(P,P_0) < \widehat{c}_{\alpha,\lambda,n}\right\}  \\
    &= \mathbb{P}\left\{ \rho(P,P_0) <\widehat{c}_{\alpha,\lambda,n} + \rho(P,P_0)- \widehat{T}_{\lambda}(P,P_0)\right\} \\
    & \leq \mathbb{P}\left[  \mathcal{U}_1\{n,\lambda,\delta, \rho(P,P_0)\} + \mathcal{U}_2\{n,\lambda,\alpha,\delta,\rho(P,P_0)\}<\widehat{c}_{\alpha,\lambda,n} + \rho(P,P_0)- \widehat{T}_{\lambda}(P,P_0)\right] \\
    &\leq \mathbb{P}\left[ \mathcal{U}_1\{n,\lambda,\delta, \rho(P,P_0)\}   < \rho(P,P_0)- \widehat{T}_{\lambda}(P,P_0)\right] +  \mathbb{P}\left[ \mathcal{U}_2\{n,\lambda,\alpha,\delta,\rho(P,P_0)\}   <\widehat{c}_{\alpha,\lambda,n} \right]\\
    & \leq \delta,
\end{aligned}
$$
where the last step follows from \eqref{equ: estimation error} and \eqref{equ: upper bound critical value} by taking $\eta = \delta$. Taking the supremum among $\mathcal{P}(\mathcal{C},\rho,\Delta)$ completes the proof.
\end{proof}

  Lemma~\ref{lemma: general framework} provides a general approach for characterizing the detection boundary of a test over a given distribution class. Once the $\alpha$-level test $\phi_{\alpha}^{T_{\lambda}}$ in~\eqref{equ: a test} is established, the derivation of its detection boundary  over the distribution class $\mathcal{P}(\mathcal{C},\rho,\Delta)$ in \eqref{distribution-class}
 involves a two-step procedure. 
First, the optimal regularization parameter $\lambda^{\ast}$ is chosen by minimizing the right-hand side of~\eqref{equ: detection boundary inequality}. 
Then, the resulting $\alpha$-level test $\phi_{\alpha}^{T_{\lambda^{\ast}}}$ achieves a detection boundary that does not exceed the smallest solution of~\eqref{equ: detection boundary inequality}.

\section{Spectral Regularized Kernel GOF Tests}\label{sec:  spectral regularized GOF test}
This section introduces our spectral regularized kernel testing method for GOF problems.

 \subsection{Kernel-based GOF Tests }
 A Mercer kernel is said to be \emph{characteristic} if 
 the kernel embedding  mapping $ P \mapsto \mu_P$  is injective. Typical examples of characteristic kernels include the Gaussian, Laplace, inverse-multiquadratic, and Matérn kernels, among others. A detailed discussion on the characteristic property of positive definite kernels can be found in \citep{sriperumbudur2011universality, simon2018kernel}. Kernel-based GOF tests are often constructed using characteristic kernels through kernel-based discrepancy measures. 
A classical example is the \emph{maximum mean discrepancy} (MMD) between $P$ and $P_0$~\citep{gretton2007kernel,gretton2012kernel}, defined as
\begin{equation}\label{MMD}
	\mathrm{MMD}^2(P,P_0) := \|\mu_P - \mu_{0} \|_{K}^2 
	= \mathbb{E}\!\left[K(X,X^{\prime})\right] + \mathbb{E}\!\left[K(Y,Y^{\prime})\right] - 2\,\mathbb{E}\!\left[K(X,Y)\right],
\end{equation}
where $X,X^{\prime}\sim P$ and $Y,Y^{\prime}\sim P_0$. 
When the kernel is characteristic, we have $\mathrm{MMD}(P,P_0) = 0$ if and only if $P=P_0$.

For a fixed significance level $\alpha > 0$, given observations $x^n$ from $P$ and the known distribution $P_0$, an unbiased estimator of MMD is
\begin{equation}\label{MMD: GOF}
	\widehat{\mathrm{MMD}}^2(P, P_0) := \frac{1}{n(n-1)}\sum_{i\neq j} K(x_i,x_j) - \frac{2}{n} \sum_{i=1}^n \mu_{0}(x_i) + \|\mu_{0}\|_K^2.
\end{equation}
Once the critical value $\widehat{c}_{\alpha,n}$ is determined, the MMD-based test is defined as
\begin{equation}\label{MMD-GOF}
    \phi_{\alpha}^{\mathrm{MMD}}:=\mathbf{1}\left\{\widehat{\mathrm{MMD}}^2(P, P_0)\geq \widehat{c}_{\alpha,n}\right\}.
\end{equation}
Despite the well-established finite-sample and asymptotic properties~\citep{gretton2012kernel,gretton2007kernel}, 
MMD-based tests may suffer from low power due to their inherent structure. 
In particular, the operator representation of~\eqref{MMD} takes the form
\begin{equation}\label{MMD-Integral}
     \mathrm{MMD}^2(P,P_0) =\sum_{k \geq 1} \lambda_k \left[\mathbb E_P\varphi_k(X) - \mathbb E_{P_0}\varphi_k(X)\right]^2, 
\end{equation}
which can be substantially smaller than the $\chi^2$-divergence
\begin{equation}\label{chi-diveregence-operator}
     \chi^2(P,P_0)=\left\|\frac{dP}{dP_0} - 1 \right\|_{L^2(P_0)}^2
     =\sum_{k \geq 1} \left[\mathbb E_P\varphi_k(X) - \mathbb E_{P_0}\varphi_k(X)\right]^2,
\end{equation}
where $\{(\lambda_k,\varphi_k)\}_{k\geq 1}$ denotes a set of normalized eigenpairs of the integral operator $L_K$ defined by \eqref{integral operator}. 
In fact, since the eigenvalues $\lambda_k$ typically decay rapidly, the $\chi^2$-divergence provides a more faithful characterization of the discrepancy between $P$ and $P_0$, especially when the difference is concentrated in the higher-order Fourier coefficients. 
In particular, if $\left[\mathbb{E}_P\varphi_k(X) - \mathbb{E}_{P_0}\varphi_k(X)\right]$ remains significant for sufficiently large $k$, such discrepancies may be severely down-weighted by the small values of $\lambda_k$, as recently pointed out in~\citep{balasubramanian2021optimality,hagrass-two-sample-test,hagrass-gof-test}.

Under the assumption $\mu_0 = 0$, \citet{balasubramanian2021optimality} constructed a new probability measure incorporating Tikhonov regularization in order to mitigate the rapid decay of eigenvalues. Specifically, their proposed regularized kernel distance is defined as
\begin{equation} \label{moderated_MMD}
	\gamma_{\lambda}^2(P,P_0) :=  \sum_{k\geq1} \frac{\lambda_k}{\lambda_k + \lambda} \left[\mathbb E\varphi_k(X)\right]^2=\left\|(L_K+\lambda I)^{-1/2}\mu_P\right\|_K^2,
\end{equation}
 and the empirical counterpart is given by
\begin{equation}
    \label{equ: empirical moderated MMD}
    	\widehat{\gamma}_{\lambda}^2(P,P_0):=\sum_{k\geq1} \frac{\lambda_k}{\lambda_k + \lambda} \left[\frac{1}{n}\sum_{i = 1}^n\varphi_k(X_i)\right]^2.
\end{equation}
\citet{balasubramanian2021optimality} established the asymptotic normality of their proposed statistic \eqref{equ: empirical moderated MMD}, thereby constructing an asymptotically $\alpha$-level test. For the distribution class \eqref{Yuanming: alternative space}, they derived the detection boundary of the test as $n^{-{2}/\{2+(\theta+1)s\}}$ when the eigenvalues of the kernels satisfying $\lambda_k\asymp k^{-1/s}$ with $0<s<1$, where $a_k\asymp b_k$ means that there exists some universal constant $c, c^{\prime}>0$ such that $cb_k\leq a_k \leq c^{\prime} b_k$ for large $k$. Furthermore, they demonstrated the minimax optimality of the proposed test.

To remove the restrictive assumption $\mu_0 = 0$ imposed by~\citet{balasubramanian2021optimality} and to overcome the saturation phenomenon inherent in Tikhonov regularization algorithms~\citep{gerfo2008spectral}, \citet{hagrass-gof-test} proposed a spectral regularized kernel GOF test based on the discrepancy measure
\begin{equation}
	\label{SRD: another representation}
	\eta_{\lambda}^2(P,P_0):= \left\| g_{\lambda}^{1/2}(\Sigma_0)(\mu_P - \mu_0) \right\|_K^2,
\end{equation}
where $\Sigma_0 := \int_{\mathcal{X}} (K_x - \mu_0) \otimes (K_x - \mu_0)\, dP_0(x)$ is the centered covariance operator, and $g_{\lambda}$ denotes a spectral filter that provides a regularized approximation of the inverse map $x\mapsto x^{-1}$.
 In particular, when $g_{\lambda}(x) = (x+\lambda)^{-1}$, the statistic reduces to the Tikhonov regularization form in~\eqref{moderated_MMD}.
 
 Building on the estimation for the covariance operator~\citep{sriperumbudur2022approximate}, \citet{hagrass-gof-test} presented a natural unbiased estimator for $\eta_{\lambda}^2(P,P_0)$. Specifically, let $x^n:=\{x_i\}_{i=1}^n\sim P$ and $y^m:=\{y_j\}_{j=1}^m\sim P_0$ denote i.i.d. samples, and $\hat{\Sigma}_0$ be an consistent estimator for $\Sigma_0$ based on another independent sample from $P_0$. The empirical statistic is given by
\begin{equation}
    \label{equ: emprical SRD}
    \widehat{\eta}_{\lambda}^2(P,P_0):= \frac{1}{n(n-1)}\frac{1}{m(m-1)}\sum_{ i \neq i^{\prime} } \sum_{ j \neq j^{\prime} }\left\langle g_{\lambda}^{1/2}(\hat{\Sigma}_0)(K_{x_i}-K_{y_j}), g_{\lambda}^{1/2}(\hat{\Sigma}_0)(K_{x_{i^{\prime}}}-K_{y_{j^{\prime}}})\right\rangle_K.
\end{equation}
However, there is a critical gap between the covariance operator $\Sigma_0$ and the integral operator $L_K$. 
The transition from~\eqref{moderated_MMD} to~\eqref{equ: emprical SRD} within their framework, introduced certain analytical challenges that necessitate additional conditions on the spectral filter. 
In contrast to the general definition of spectral filters in the literature~\citep{guo2017learning,gerfo2008spectral,bauer2007regularization}, \citet{hagrass-gof-test} introduced an extra assumption on $g_{\lambda}$, requiring the existence of a universal constant $c>0$, independent of $\lambda$, such that
\begin{equation}
	\label{additional assumption}
	\inf_{x} g_{\lambda}(x)(x+\lambda) \;\geq\; c,
\end{equation}
an assumption that plays a key role in their power analysis. Within the alternative space~\eqref{Hagrass: alternative space}, the minimax detection boundary is $n^{-{4r}/{(4r+s)}}$ when the eigenvalues of $\Sigma_0$ decay at the rate $\lambda_i \asymp i^{-1/s}$ with $0<s<1$. 
This rate coincides with the detection boundary obtained by tests based on \eqref{equ: emprical SRD}, thereby demonstrating the optimality of their proposed procedures.

\subsection{Test Statistics with Spectral Regularization}
In our approach, we aim to avoid both the assumption $\mu_0=0$ of \citet{balasubramanian2021optimality} and the spectral assumption \eqref{additional assumption} of \citet{hagrass-two-sample-test,hagrass-gof-test} by constructing the discrepancy measurement based on the integral operator $L_K$ and widely used spectral filters \citep{gerfo2008spectral,bauer2007regularization}, $g_{\lambda}:[0,\kappa^2] \to \mathbb R^{+},$ satisfying
\begin{equation}
		\label{equ: filter property 1}
		\sup_{0\leq x\leq \kappa^2}|g_{\lambda}(x)| \leq \frac{b}{\lambda}, \quad  \sup_{0\leq x\leq \kappa^2}|g_{\lambda}(x)x| \leq b,
	\end{equation}
	and 
	\begin{equation}
		\label{equ: filter property 2}
		\sup_{0\leq x\leq \kappa^2}|1-g_{\lambda}(x)x|x^{\nu} \leq \gamma_{\nu}\lambda^{\nu}, \quad \forall 0\leq \nu\leq\nu_g,
	\end{equation}
	where $\gamma_{\nu} >0$ is a constant depending only on $\nu \in  (0,\nu_g]$ and $b$ is an absolute constant.
To ensure a well-posed problem in (\ref{GOF test}), we make the following assumptions throughout this paper.
\begin{assumption}
\label{assumption: absolutely continuous}
	The alternative distribution $P$ is absolutely continuous with respect to the null distribution $P_0$, and the target function $f:=dP/dP_0-1$ belongs to $L^2(P_0)$.
\end{assumption}

Based on Assumption~\ref{assumption: absolutely continuous} and~\eqref{chi-diveregence-operator}, 
we introduce the following discrepancy measure
\begin{equation}
	\label{corrected MMD}
	\xi_{\lambda}(P,P_0)
	:= \left\langle g_{\lambda}(L_K)L_K f,\, f \right\rangle_{L^2(P_0)}
	= \left\| g^{1/2}_{\lambda}(L_K)(\mu_P - \mu_0)\right\|_{K}^2,
\end{equation}
as an approximation to $\chi^2(P,P_0)$. 
In fact, by~\eqref{equ: filter property 2} we have
\begin{equation}\label{App-limitation}
      \lim_{\lambda \to 0} \xi_{\lambda}(P,P_0) \to \chi^2(P,P_0).
\end{equation}

Given i.i.d. observations $x^n:=\{x_i\}_{i=1}^n \sim P$, 
the test statistic for GOF based on~\eqref{corrected MMD} can be approximated by
\[
 \left\| g^{1/2}_{\lambda}(L_K)\left(\frac{1}{n}\sum_{i=1}^nK_{x_i} - \mu_0\right)\right\|_{K}^2.
\]
Since $P_0$ is known, $\mu_0$ can be estimated using an additional i.i.d. sample $y^m:=\{y_j\}_{j=1}^m\sim P_0$. 
Moreover, it is well known~\citep{caponnetto2007optimal} that the empirical integral operator 
$L_{K,D}:\mathcal H_K \to \mathcal H_K$ defined by 
\begin{equation}\label{empirical integral operator}
      L_{K,D}(f)=\frac1{N}\sum_{\ell=1}^N f(z_\ell)K_{z_\ell},
\end{equation}
with i.i.d. samples $D:=\{z_\ell\}_{\ell=1}^N\sim P_0$, provides a good approximation of $L_K:\mathcal H_K \to \mathcal H_K$. Combining the auxiliary samples $y^m$, the operator $L_{K,D}$, and the requirement of unbiasedness, 
we construct the following statistic:
\begin{equation}
	\label{two sample statistic}
	\widehat{\xi}_{\lambda}(P,P_0)
	:= \frac{1}{n(n-1)}\frac{1}{m(m-1)}\sum_{ i \neq i^{\prime} } \sum_{ j \neq j^{\prime} }
	\left\langle g_{\lambda}^{1/2}(L_{K,D})(K_{x_i}-K_{y_j}),\, 
	g_{\lambda}^{1/2}(L_{K,D})(K_{x_{i^{\prime}}}-K_{y_{j^{\prime}}})\right\rangle_K.
\end{equation}

Define kernel matrices  $K_{NN} = [K(z_i,z_j)]_{i,j\in [N]}$,   $K_{nN} = [K(x_i,z_j)]_{i\in [n],j\in[N]}$, $K_{mN} = [K(y_i,z_j)]_{i\in [m],j\in[N]}$. The following proposition, whose proof is given in Appendix~\ref{section: appendix B}, shows that the two-sample statistic (\ref{two sample statistic}) can be computed through some simple matrix manipulations, and its computational complexity is comparable to the regularized statistic proposed in~\citep{hagrass-two-sample-test,hagrass-gof-test}.  
\begin{proposition}
	\label{computation}
	Let $\widehat{\xi}_{\lambda}(P,P_0)$ be defined in (\ref{two sample statistic}). Denote by $\{(\widehat{\lambda}_i,\widehat{\boldsymbol{\alpha}}_i)\}_{i\in[N]}$ the normalized eigenpairs of the scaled kernel matrix $K_{NN}/N$, and define
	$$
	G_{\lambda, N} = \sum_{i=1}^N\widehat{\lambda}_{\ell}^{-1} g_{\lambda}\left(\widehat{\lambda}_i\right)\widehat{\boldsymbol{\alpha}}_i\widehat{\boldsymbol{\alpha}}_i^{\top}.
	$$
	Then,
	$$
	\begin{aligned}
		\widehat{\xi}_{\lambda}(P,P_0) \quad &=\quad  \frac{1}{n(n-1)N}\left[\mathbf{1}_{n}^{\top}K_{nN}G_{\lambda,N}K_{nN}^{\top}\mathbf{1}_n -\mathrm{Tr}\left(K_{nN}G_{\lambda,N}K_{nN}^{\top} \right)\right]\\
		\quad &\quad\quad+\quad  \frac{1}{m(m-1)N}\left[\mathbf{1}_{m}^{\top}K_{mN}G_{\lambda,N}K_{mN}^{\top}\mathbf{1}_m -\mathrm{Tr}\left(K_{mN}G_{\lambda,N}K_{mN}^{\top} \right)\right]\\
		\quad &\quad\quad-\quad  \frac{2}{nmN}\mathbf{1}_{n}^{\top}K_{nN}G_{\lambda,N}K_{mN}^{\top}\mathbf{1}_m,
	\end{aligned}
	$$
	where $\mathbf{1}_k$ denotes the all-one vectors of dimension $k$ and  $\mathrm{Tr}(A)$ denotes the trace of an operator (or matrix) $A$. 
\end{proposition}

\subsection{Critical Value Determination}
Based on the statistics developed in~\eqref{moderated_MMD}, 
we present two schemes for determining the critical value.

The first one relies on the empirical effective dimension, defined by
\begin{equation}\label{empirical-effective-dimension}
    \mathcal{N}_D(\lambda):= \mathrm{Tr}\left[(K_{NN} + \lambda NI)^{-1}K_{NN}\right],\quad \lambda >0.
\end{equation}
The empirical effective dimension $\mathcal{N}_D(\lambda)$ reflects not only the smoothness of the kernel but also the marginal distributional information of $P_0$. 
More importantly, it provides a natural measure of the variance of the statistic, and has therefore been adopted in~\citep{hagrass-two-sample-test,hagrass-gof-test} for calibrating critical values. 
 In our approach, we set 
\begin{equation}\label{critical-value-1}
    \widehat{c}_{\alpha,\lambda,n,m,D} = 30b \alpha^{-1} \left(\frac{1}{n-1}+\frac{1}{m-1}\right)\left(1 +\frac{8\kappa}{\sqrt{N\lambda}}\log\frac{24}{\alpha}\right)\left\{\mathcal{N}_D(\lambda)\right\}^{1/2}
\end{equation}
as the critical value and then  get a GOF test as 
\begin{equation}
          \label{equ: testing rule based on inequality}
          \phi_{\alpha}^{\xi_{\lambda}}\left(x^n;y^m,D\right)= \mathbf{1}\left\{\widehat{\xi}_{\lambda}(P,P_0) \geq \widehat{c}_{\alpha,\lambda,n,m,D}\right\},
\end{equation}
where $\widehat{\xi}_{\lambda}(P,P_0)$ is given by \eqref{two sample statistic}. We summarize the testing procedure in Algorithm~\ref{alg:concentration_test} in Appendix~\ref{section: appendix A}. The $\mathcal N_D(\lambda)$-based critical value has the advantage of being straightforward to compute that requires $O(N^3)$ time for matrix inversion, while at the same time tightly capturing the variance structure of the statistic to ensure valid Type~I error control. 
However, the constants in~\eqref{critical-value-1} are derived from concentration inequalities under worst-case scenarios, which makes the resulting GOF test conservative in practice, yielding sub-nominal Type-I error and reduced power. 

Our second test is based on the well-known permutation approach~\citep{lehmann2008testing,hagrass-two-sample-test,hagrass-gof-test}, whose basic idea is to recompute the test statistic under randomly permuted sample labels and then use the resulting empirical distribution as a reference for calibration.
 Specifically, recall $x^n = \{x_i\}_{i=1}^n\sim P,y^m = \{y_j\}_{j=1}^m\sim P_0$ and $D = \{z_{\ell}\}_{\ell = 1}^N\sim P_0$, and define the pooled sample as $u^{n+m} = x^n \cup y^m $. Let $\Pi_{n+m}$ denote the set of all possible permutations of $\{1,\cdots, n+m \}$ and $\{\pi_b\}_{b = 1}^B$ be $B$ random permutations drawn from $\Pi_{n+m}$. Denote further $\widehat{\xi}_{\lambda}^{\pi}(P,P_0)$ as the statistic based on the permuted samples $x_{\pi}^n,y_{\pi}^m$ and $D$. For simplicity, write $\widehat{\xi}_{\lambda}^{\pi_b}(P,P_0)$ as $\widehat{\xi}_{\lambda}^{b}(P,P_0)$ for $1\leq b \leq B$, and $\widehat{\xi}_{\lambda}^{0}(P,P_0) \equiv\widehat{\xi}_{\lambda}(P,P_0) $ denotes the statistic based on the original samples without permutation. The empirical permutation distribution function is defined by
\begin{equation}
    \label{empirical permutation distribution}
 \hat{F}_{B,\lambda}(t):= \frac{1}{B+1}\sum_{b = 0}^B \mathbf{1}\left\{\widehat{\xi}_{\lambda}^{b}(P,P_0) \leq t\right\}, \qquad t\geq 0,
\end{equation}
 and the empirical $(1-\alpha)$-th quantile of $\hat{F}_{B,\lambda}$ is defined by
 \begin{equation}
 	\label{empirical quantile}
 	\hat{q}_{1-\alpha}^{B, \lambda}:=\inf \left\{ t: \hat{F}_{B,\lambda}(t) \geq 1-\alpha\right\}.
 \end{equation}
 Then we get a GOF test as
 \begin{equation}
          \label{equ: testing rule based on perm}
          \phi_{\alpha}^{\xi_{\lambda},perm}(x^n;y^m,D):= \mathbf{1}\left\{\widehat{\xi}_{\lambda}(P,P_0) \geq \hat{q}_{1-\alpha}^{B, \lambda}\right\}.
\end{equation}
 We summarize the testing procedure in Algorithm~\ref{alg:permutation_test} in Appendix~\ref{section: appendix A}. Compared to our $\mathcal N_D(\lambda)$-based approach, the permutation method has the advantage of calibrating the test closer to the nominal significance level, thereby avoiding the conservativeness inherent in worst-case concentration bounds and yielding higher empirical power in practice.

 \subsection{Comparisons}
In this subsection, we compare our proposed approach with several representative kernel-based GOF tests from the literature. The idea of distributional comparison based on MMD was initially proposed in~\citep{smola2007hilbert,gretton2007kernel,gretton2012kernel}. These pioneering works establish both the asymptotic properties and the finite-sample guarantees of the associated estimators, including $U$- and $V$-statistics, and have since become a standard baseline for comparison with many kernel-based methods.

A recent work \citep{balasubramanian2021optimality} studied the detection boundary of MMD-based test over the distribution class \eqref{Yuanming: alternative space}, revealing its suboptimality for GOF problems within the minimax framework~\citep{ingster1987minimax,ingster1993asymptotically}, and highlighting the advantages of Tikhonov regularization in improving test power. However, the theoretical benefits of the regularized statistic \eqref{equ: empirical moderated MMD} rest on three restrictions. First, the validity of the proposed test relies on large-sample theory and involves the eigensystem of the kernel function, which is often difficult to compute in practice, although it can be 
theoretically obtained when both the kernel function and $P_0$ are known.
 Then, they imposed the  assumption of the kernel function degenerates at $P_0$, i.e., $\mu_0=0$, and uniform boundedness condition on eigenfunctions of integral operators. These assumptions are overly restrictive and may rule out many kernel functions and probability measures.  Moreover, it is well known that Tikhonov regularization suffers from the saturation phenomenon, and in the testing framework, this means that the detection boundary of the associated test cannot be further improved even if the alternative enjoys higher regularity.

More recently, \citet{hagrass-gof-test} adopted a spectral regularization approach to address the limitations of the test in \citep{balasubramanian2021optimality}. 
In addition to resolving the computational challenges and eliminating the extra assumptions on kernel functions, 
the use of spectral filters further mitigates the saturation phenomenon inherent in Tikhonov regularization.
 In their work, the alternative space~\eqref{Hagrass: alternative space} is considered for theoretical power analysis. 
There, the centered integral operator $L_{\bar{K}}$ is defined on $L^2(P_0)$ or on $\mathcal{H}_{\bar{K}}$, whereas the construction of~\eqref{equ: emprical SRD} relies on the covariance operator $\Sigma_0$, which is defined on $\mathcal{H}_K$.
 Although at the population level, these two different operators can be connected through the centered inclusion operator, defined by
$$
\bar{I}_K: \mathcal{H}_{K} \to L^2(P_0), \quad f\mapsto f-\mathbb{E}_{P_0}f,
$$
the empirical operator  $g_{\lambda}(\hat{\Sigma}_0)$ in \eqref{equ: emprical SRD} cannot be connected with the source condition $dP/dP_0-1 \in \mathrm{Ran}(L_{\bar{K}}^r)$ directly. Hence, the additional assumption in~\eqref{additional assumption} is required to guarantee an \emph{error-free} conversion between $g_{\lambda}(\hat{\Sigma}_0)$ and $g_{\lambda}(\Sigma_0)$ in their \emph{bias} analysis. 
As a consequence, the spectral filter that do not satisfy~\eqref{additional assumption}, such as the spectral cut-off with the filter function \begin{equation}
    \label{equ: cut-off}
    g_{\lambda}(x)=x^{-1}\mathbf{1}\{x\geq \lambda\},
\end{equation} is excluded from their theoretical analysis and would require new framework to be incorporated. 

In contrast, our approach proceeds differently. 
Comparing~\eqref{corrected MMD} with~\eqref{SRD: another representation}, the only apparent difference lies in the choice of operators for spectral computation: \citet{hagrass-gof-test} employed the covariance operator $\Sigma_0$, whereas we rely on the integral operator $L_K$. 
However, this seemingly slight difference gives rise to fundamentally different design philosophy and theoretical analyses in constructing the statistics. 
While inheriting some of the advantages of the methods in \citep{hagrass-two-sample-test,hagrass-gof-test}, our proposed statistic offers more intrinsic benefits, both in terms of methodological flexibility and theoretical tractability.

First, our introduced quantity in~\eqref{SRD: another representation} can be viewed as a  natural extension of~\eqref{moderated_MMD}, as it allows for a broad  class of spectral filters and removes the restrictive assumption $\mu_0=0$, while simultaneously ensuring convergence to the $\chi^2$-divergence as established in~\eqref{App-limitation}. Then, our analysis dispenses with the additional condition~\eqref{additional assumption} on filter functions,  which is particularly reflected in the \emph{bias} analysis.
The key reason is that $L_K$, as an operator on $L^2(P_0)$ and restricted to $\mathcal{H}_K$, can essentially be regarded as the integral operator on $\mathcal{H}_K$, a property not shared between $L_{\bar{K}}$ and $\Sigma_0$. 
Consequently, our constructed statistic can decouple the regularity of $f = dP/dP_0 - 1$ without relying on the restriction~\eqref{additional assumption}. 
By adopting the difference-based error decomposition approach commonly used in kernel regression~\citep{caponnetto2007optimal,lin2020distributed,guo2017learning}, the bias of our statistic can then be characterized explicitly.  Moreover, while facilitating the analysis of bias, it does not mean that the \emph{variance} has increased. The relation 
  \begin{equation}\label{operator-relation}
    	0 \preccurlyeq \Sigma_0= \int_{\mathcal{X}}  (K_x -\mu_0) \otimes (K_x -\mu_0) dP_0(x) =   L_K - \mu_0 \otimes\mu_0\preccurlyeq L_K
\end{equation}
  implies that the eigenvalues of $\Sigma_0$ are uniformly smaller than that of $L_K$.  One can see from the matrix form of the statistics that sample eigenvalues appear in the denominators of both our statistic and those in~\citep{hagrass-two-sample-test,hagrass-gof-test}. However, smaller eigenvalues close to zero may substantially amplify the noise caused by estimation, and thus our proposed statistic can enjoy a smaller variance. 
From the theoretical perspective, under $H_0$, the variance of the statistics in~\citep{hagrass-two-sample-test,hagrass-gof-test} can be expressed directly in terms of the effective dimension, whereas the variance of our statistic is characterized by the Hilbert--Schmidt norm of 
\[
\Theta_{\lambda,P_0,P_0} := (L_{K} + \lambda I)^{-1/2}\,\Sigma_0\,(L_{K} + \lambda I)^{-1/2},
\]
which can be further controlled by the effective dimension $\mathcal{N}(\lambda)$ through the relation~\eqref{operator-relation}. 
When the kernel function is not degenerate at $P_0$ (i.e., $\mu_0 \neq 0$), the variance bound is strictly smaller. It is worth noting that~\citet{hagrass-two-sample-test,hagrass-gof-test} characterized the variance in terms of the effective dimension of $\mathcal{H}_{\bar{K}}$, the RKHS associated with the centered kernel $\bar{K}$, whereas our analysis employs the effective dimension of $\mathcal{H}_K$.
However, recent results~\citep{wang2024learning} show that the effective dimensions of $\mathcal{H}_{K}$ and $\mathcal{H}_{\bar{K}}$ differ by at most one. 
An empirical variance comparison between our statistic~\eqref{two sample statistic} and that of~\citep{hagrass-gof-test} is reported in Figure~\ref{fig:toysimulation}. 
\begin{figure}
			\centering
			\includegraphics[width=1\linewidth]{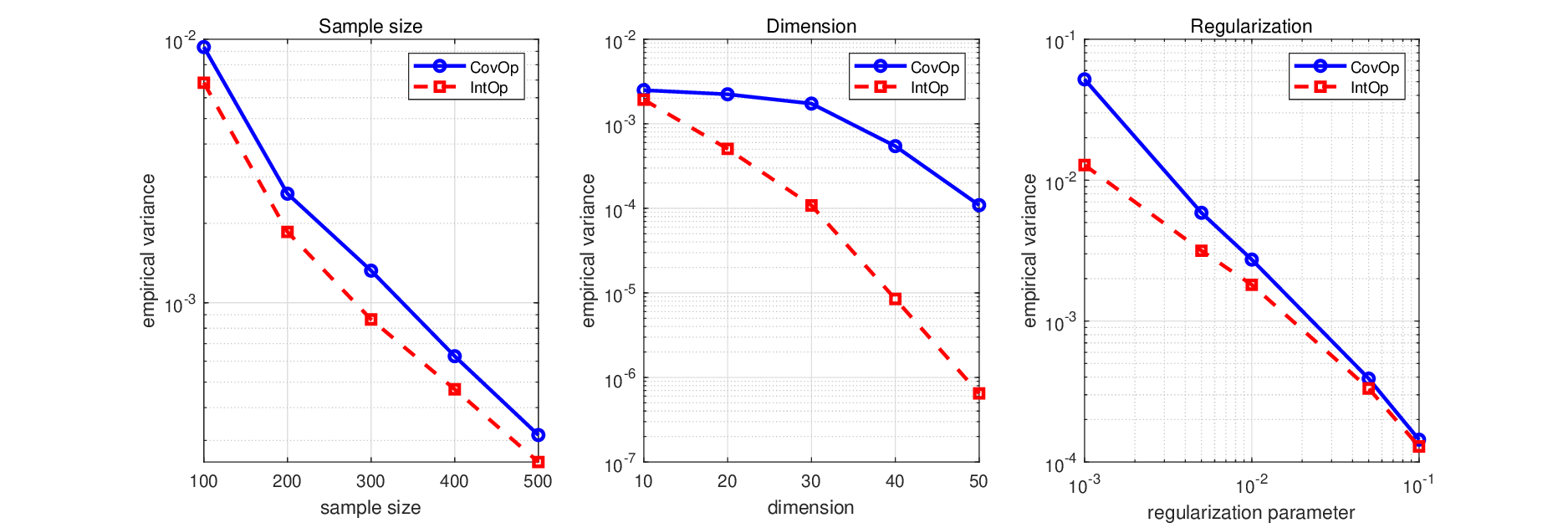}
			\caption{\footnotesize {\textit{Empirical variance comparison of our proposed statistic and the method in~\citep{hagrass-gof-test} under the null hypothesis, based on simulations with standard normal data from $N(0,I_d)$ using a Gaussian kernel function. Three subfigures illustrate the empirical variance variation with respect to sample size (fixing \( d = 10 \),  \( \lambda = 0.01 \)), dimension (fixing \(n = m = 200 \),  \( \lambda = 0.01 \)), and regularization parameter (fixing  \(n = m = 200 \),  \( d = 10 \)). The estimates for the integral operator and the centralized covariance operator are both based on $200$ samples.}}}
			\label{fig:toysimulation}
		\end{figure}
This reduction in variance indicates that the detection boundary of our proposed tests derived from the bias-variance analysis is no worse than those in \citep{hagrass-two-sample-test,hagrass-gof-test}. As shall be shown in our theoretical analysis of the next section, it in fact coincides with the detection boundary rate established therein.

\section{Theoretical Verification}\label{sec: theory}
This section provides theoretical verifications for the two proposed spectral regularized GOF tests, \eqref{equ: testing rule based on inequality} and \eqref{equ: testing rule based on perm}. 
In light of~\eqref{App-limitation}, we employ the $\chi^2$-divergence to quantify the discrepancy between $P$ and $P_0$, and evaluate the performance of the GOF tests over the alternative space 
\begin{equation}
	\label{ours: alternative space}
	\mathcal{P}_{r,\Delta}:= \left\{P\in \mathcal{C}_K^r: \chi^2(P,P_0) = \left\|f \right\|_{L^2(P_0)}^2 \geq \Delta\right\},
\end{equation}
where 
\begin{equation}\label{space-priori}
    \mathcal{C}_K^r:=\left\{P:\frac{dP}{dP_0}-1\in \mathrm{Ran}(L_K^r)\right\}.
\end{equation}
In the family of distributions $\mathcal{P}_{r,\Delta}$, the smoothness of target function $f$ is described through the index $r$, and larger index $r$ implies better regularity of $f$. In particular,  $r \geq 1/2$ implies $f \in \mathcal{H}_K$ while $0<r<1/2$ means that $f$ lies in the interpolation space between $\mathcal{H}_K$ and $L^2(P_0)$. It should be noted that the alternative space in~\eqref{ours: alternative space} generally differs from that in~\eqref{Hagrass: alternative space}. 
The latter is defined via the integral operator \(L_{\bar{K}}\), constructed from the centered kernel function \( \bar{K}(x,y) = \langle K_x - \mu_0, K_y - \mu_0 \rangle_{K} \), 
whereas our formulation employs \(L_K\), which can be built from any Mercer kernel without requiring centering. 
These two distribution classes coincide only when the kernel \(K\) is intrinsically centered (i.e., \(K=\bar{K}\)) or when \(K\) degenerates on \(P_0\).

To explicitly characterize the detection boundary, we introduce the population-level effective dimension, defined as
\[
\mathcal{N}(\lambda) := \mathrm{Tr}\left[(L_K + \lambda I)^{-1}L_K\right], \quad \lambda > 0,
\]
which is commonly used in the literature \citep{caponnetto2007optimal,lin2020distributed,guo2017learning} to measure the complexity of the hypothesis space and regularity of the distribution $P_0$. We make the following assumption on the effective dimension.
\begin{assumption}
	\label{assumption: effective dimension}
	There exists some $s \in (0,1]$ such that
	\begin{equation}
		\label{equ: effective dimension assumption}
		\mathcal{N}(\lambda) \leq C_0 \lambda^{-s},
	\end{equation}
	for some universal constant $C_0 \geq 1$. Moreover, assume that 
$\mathcal{N}_{D}(\lambda) \geq 1$ and $\mathcal{N}(\lambda) \geq 1$.
\end{assumption}
In Assumption \ref{assumption: effective dimension}, condition \eqref{equ: effective dimension assumption} with $s =1$ is always satisfied by taking the constant $C_0 = \kappa^2$. For $0<s<1$, \eqref{equ: effective dimension assumption} is equivalent to the eigenvalue decaying assumption in the literature~\citep{hagrass-two-sample-test,hagrass-gof-test,caponnetto2007optimal}.
With these helps, we first present our theoretical verification of the GOF test in \eqref{equ: testing rule based on inequality}.
 \begin{theorem}
\label{theorem: testing rule based on inequality} Let $0<\alpha,\delta<1$, $0<\lambda \leq 1$ and $n,m\geq 3$. Under Assumptions~\ref{assumption: absolutely continuous}-\ref{assumption: effective dimension}, if 
$N\lambda \geq \max\{\,16C_1^2(\log(eC_0)+s\log(\lambda^{-1}))\log^2(12\alpha^{-1}),\;4\kappa^2\}$, the test $\phi_{\alpha}^{\xi_{\lambda}}$ provided in \eqref{equ: testing rule based on inequality} is an $\alpha$-level test,  where the constant $C_1= \max\{(\kappa^2+1)/3,2\sqrt{\kappa^2 + 1}\}$.
 If in addition,  $\lambda = n^{-\frac{2}{4r+s}}$, and  $m,N>n$ satisfying
 $N\geq C_1^{\prime}  n^{\frac{2}{4r+s}}\log n \log^2(24\delta^{-1})$,
then
\begin{equation}\label{first-boundary-1}
    \Delta_{n}(\phi_{\alpha}^{\xi_{\lambda}}; 	\mathcal{C}_K^r,\chi^2,\delta)\leq C^{\ast}(\alpha,\delta) n^{-\frac{4r}{4r+s}},\qquad\forall 1/2 \leq r \leq \nu_g,
\end{equation}
where   $C^{\ast}(\alpha,\delta)=O( \max\{\delta^{-1}, \log(\delta^{-1})\}+\alpha^{-1}\log(\alpha^{-1})\log(\delta^{-1}) )$ is independent of $n$, and the constant $C_1^{\prime} =\max\{32C_1^2(\log(eC_0)+{2s}{(4r+s)^{-1}}),4\kappa^2\} $.
\end{theorem}

Theorem~\ref{theorem: testing rule based on inequality} provides a theoretical guarantee for $\phi_{\alpha}^{\xi_{\lambda}}$ defined in \eqref{equ: testing rule based on inequality}, in terms of Type I error control and detection boundary rate. The conclusion of Theorem~\ref{theorem: testing rule based on inequality} also holds for a broad class of spectral filters that do \emph{not} satisfy the additional condition~\eqref{additional assumption}.
In particular, the theoretical guarantees in~\eqref{first-boundary-1} remain valid for the spectral cut-off filter defined in~\eqref{equ: cut-off}, which violates~\eqref{additional assumption} because $g_{\lambda}(0)=0$.
It is well-known that the spectral cut-off filter possesses an infinite qualification, thereby  avoiding the saturation effect inherent in Tikhonov regularization~\citep{balasubramanian2021optimality}.
Hence, our result goes beyond existing studies~\citep{balasubramanian2021optimality,hagrass-two-sample-test,hagrass-gof-test},
highlighting the greater generality of our approach and the broader inclusiveness of its theoretical guarantees. We provide further explanations and remarks to clarify its implications.

In our proof framework, the variance of the statistic (\ref{two sample statistic}) under $H_0$ is characterized by the effective dimension $\mathcal{N}(\lambda)$. The empirical effective dimension $\mathcal{N}_D(\lambda)$ defined in \eqref{empirical-effective-dimension} provides a direct estimate of $\mathcal{N}(\lambda)$. By replacing $\mathcal{N}(\lambda)$ with its empirical counterpart and applying Markov’s inequality, one can show that  $\phi_{\alpha}^{\xi_{\lambda}}$ is indeed an $\alpha$-level test. 
To ensure valid control of the Type~I error, our method requires the auxiliary sample size drawn from $P_0$ for estimating $L_K$ to satisfy $N \gtrsim \lambda^{-1}\log(\lambda^{-1})\log^2(\alpha^{-1})$, where $a\gtrsim b$ means that there exists some constant $c>0$ such that $a\geq c b$. 
Theorem~\ref{theorem: testing rule based on inequality} also shows that the detection boundary of $\phi_{\alpha}^{\xi_{\lambda}}$ over the distribution class $\mathcal{P}_{r,\Delta}$ is of order at most $n^{-{4r}/{(4r+s)}}$.
 The explicit form of $C^{\ast}(\alpha,\delta)$ illustrated in \eqref{first-boundary-1} is not easy to express. Nevertheless, this quantity can be regarded as a decreasing function of both $\alpha$ and $\delta$. On the one hand, when the significance level $\alpha$ is reduced, the test must become more conservative to ensure a smaller probability of falsely rejecting $H_0$. This requires a stronger signal (i.e., larger $\chi^2(P,P_0)$) to cross the higher critical value, thereby increasing the detection boundary constant. This phenomenon substantially reflects the inherent trade-off between two types of errors.  On the other hand, when the tolerance $\delta$ for the Type~II error is reduced, the test must achieve higher power against alternatives, which again demands a larger detectable signal and thus increases the detection boundary constant. We note that the rate $n^{-{4r}/{(4r+s)}}$ in~\eqref{first-boundary-1} agrees with the detection boundary derived in~\citep{balasubramanian2021optimality} for the distribution class in~\eqref{Yuanming: alternative space}, under the condition that $f = dP/dP_0 - 1 \in \mathcal{H}_K$ (corresponding to $r=1/2$ in our setting and $\theta=0$ in theirs). This rate also coincides with the boundary obtained in~\citep{hagrass-gof-test} for the class defined in~\eqref{Hagrass: alternative space}, when the kernel eigenvalues decay polynomially. Furthermore, for any fixed $\delta > 0$, achieving this detection boundary requires $N\gtrsim n^{2/(4r+s)}\log n \log^2(\delta^{-1})$, which is comparable to those imposed in \citep{hagrass-gof-test}.

From a practical perspective, once the observed sample $x^n$ is given, one still needs to determine the auxiliary sample sizes $m$ and $N$ as well as the regularization parameter $\lambda$ in order to improve the testing performance. Our theoretical results suggest that taking $m \geq n$ and $N \gtrsim n\log n$ is sufficient to guarantee the Type I error control. However, achieving the detection boundary additionally requires a proper choice of the regularization parameter $\lambda$. Since the optimal choice of $\lambda$ depends on the smoothness parameters $r$ and $s$ that are unknown in practice, a more delicate issue arises regarding data-driven parameter selection. In the literature, various adaptive strategies have been proposed, such as sample splitting strategies~\citep{NIPS2012dbe272ba,liu2020learning} and aggregation methods~\citep{schrab2023mmd,fromont2013two,balasubramanian2021optimality, hagrass-two-sample-test,hagrass-gof-test}. The former is based on partitioning the sample, where one part is used to learn approximately optimal parameters, and the other part is then employed to construct the test statistic using these learned parameters. This reduces the risk of overfitting and ensures valid inference, but at the cost of efficiency since only a subset of the data is used for testing. The latter resembles an ensemble approach, in which a grid of candidate parameters (such as kernel bandwidths or regularization levels) is explored, and the resulting statistics are combined to form a more powerful test.  We also note that our proposed test can be combined with the second strategy to enhance empirical performance. Although the Type~I error control of this adaptive procedure can be established with relative ease, we do not provide a theoretical analysis of its effect on the detection boundary. Nevertheless, by following arguments similar to those in \citep{hagrass-two-sample-test,hagrass-gof-test}, one can obtain rigorous verification, typically at the expense of an additional \emph{logarithmic} factor in the detection boundary. A complete theoretical treatment of this issue is beyond the scope of the present work and left for future research.

Similar to Theorem~\ref{theorem: testing rule based on inequality}, we also provide the theoretical guarantee for the permutation-based test $\phi_{\alpha}^{\xi_{\lambda},perm}$ in \eqref{equ: testing rule based on perm}.
\begin{theorem}
    \label{theorem: testing rule based on permutation}
    Let $0<\alpha<e^{-1}$, $0<\delta<1$, $n,m\geq 3$, and $B\geq 1$. Under Assumptions~\ref{assumption: absolutely continuous}-\ref{assumption: effective dimension}, the test $\phi_{\alpha}^{\xi_{\lambda},perm}$ provided in \eqref{equ: testing rule based on perm} is  an $\alpha$-level test. If in addition, $m,N>n$ satisfying 
    $
    N\geq C_2^{\prime}n^{\frac{2}{4r+s}}\log n\log^2(28\delta^{-1}),
    $
    and $B\geq \frac{3}{\alpha^2}
 (\log(14\delta^{-1}) +\alpha(1-\alpha))$, then
    \begin{equation}\label{first-boundary-2}
    \Delta_{n}(\phi_{\alpha}^{\xi_{\lambda},perm}; 	\mathcal{C}_K^r,\chi^2,\delta)\leq C^{\ast\ast}(\alpha,\delta) n^{-\frac{4r}{4r+s}},\qquad\forall 1/2 \leq r \leq \nu_g,
\end{equation}
where $C^{\ast\ast}(\alpha,\delta)=O( \max\{\delta^{-1}, \log(\delta^{-1})\}+\delta^{-1}\log(\alpha^{-1}))$ is independent of $n$, the constant $C_2^{\prime} = 32C_1^2\{\log(eC_0)+{2s}{(4r+s)^{-1}}\}$, and $C_1$ is defined in Theorem~\ref{theorem: testing rule based on inequality}.
\end{theorem}

We compare Theorem~\ref{theorem: testing rule based on permutation} with Theorem~\ref{theorem: testing rule based on inequality} to highlight the distinctions between two tests in \eqref{equ: testing rule based on inequality} and \eqref{equ: testing rule based on perm}. First, to achieve the Type I error control, the permutation test $\phi_{\alpha}^{\xi_{\lambda},{perm}}$ does not impose additional constraints on the sample size, the number of permutations, or the choice of the regularization parameter, as long as the procedure is well-defined. The key idea relies on the exchangeability of the test statistic under the null hypothesis: for any permutation $\pi \in \Pi_{n+m}$, the statistics used to construct the empirical permutation distribution~\eqref{empirical permutation distribution} are identically distributed given $D$. Then, to achieve the same detection boundary as in Theorem~\ref{theorem: testing rule based on inequality}, additional requirement of the permutation times $B$ needs to be made, but the restriction on $N$ for can be slightly weaker. Moreover, it is worth emphasizing that, although Theorem~\ref{theorem: testing rule based on inequality} provides theoretical guarantees for $\phi_{\alpha}^{\xi_{\lambda}}$, the test based on the empirical effective dimension is typically conservative in controlling the Type~I error, which in turn reduces its empirical power. 
In contrast, the permutation test controls the Type~I error much closer to the nominal significance level $\alpha$, while simultaneously maintaining satisfactory power in practice. 
This difference can be further understood by comparing the factors $C^{\ast}(\alpha,\delta)$ and $C^{\ast\ast}(\alpha,\delta)$: the former involves a polynomial dependence of order $\alpha^{-1}$, leading to conservative Type~I error control, whereas the latter only grows logarithmically in $\alpha$, resulting in a more explicit calibration and improved empirical power performance.

From an implementation viewpoint, the aforementioned adaptive strategies can also be incorporated into the permutation test $\phi_{\alpha}^{\xi_{\lambda},{perm}}$, and is expected to yield theoretical results similar to those in \citep{hagrass-two-sample-test,hagrass-gof-test}.
Since the permutation-based approach demonstrates better empirical performance compared with the $\mathcal{N}_D(\lambda)$-based test $\phi_{\alpha}^{\xi_{\lambda}}$, we mainly adopt the permutation test combined with the adaptive strategy to demonstrate the effectiveness of the proposed statistic in our numerical experiments; see Section~\ref{sec: numerical experiment} for details.

\section{Numerical Study}
\label{sec: numerical experiment}
In this section, we conduct three sets of simulations to assess the performance of our proposed testing procedures. The first set examines the effect of regularization on test performance through a specific illustrative example. The second set benchmarks our methods against several state-of-the-art approaches, including the aggregated MMD test~\citep{schrab2023mmd}, the Energy test~\citep{szekely2004testing}, and the spectral regularized goodness-of-fit test of~\citep{hagrass-gof-test}. The third set explores the robustness of our approach with respect to different spectral filter functions, highlighting its applicability across a broad class of filters used in constructing the test statistic. In all simulations, the significance level is fixed at $\alpha = 0.05$, and samples $x^n := \{x_i\}_{i=1}^n$ are independently drawn from the distribution $P$, while additional samples $y^m := \{y_j\}_{j=1}^m$ and $D := \{z_\ell\}_{\ell=1}^N$ are independently generated from the null distribution $P_0$. Although our theory suggests that both $m,N$ should be no smaller than $n$, we observe in practice that the testing performance remains satisfactory even when $N<n$, a similar phenomenon was also reported in \citep{hagrass-two-sample-test,hagrass-gof-test}. We therefore adopt relatively small values of $N$ in our experiments to reduce the computational cost without noticeable loss of empirical performance.

\subsection{Effects of Regularization on Test Performance}
In this subsection, we design a tailored family of distributions that aligns with our theoretical assumptions. This controlled setup allows us to clearly isolate the role of regularization and to illustrate how different regularization parameters affect the resulting test performance. Specifically, let $\mathcal{X} = [0,1]^d$ and define the Sobolev kernel function
\[
K_{\mathrm{sob}}(x,y) = \prod_{j = 1}^d \min\{x_j, y_j\}, \quad x,y \in [0,1]^d.
\]
The first-order Sobolev space on $[0,1]^d$ with zero boundary at the origin is defined by
\[
H_0^1([0,1]^d) := \Biggl\{ f:[0,1]^d \to \mathbb{R} \ \Big| \ f(0,\cdots,0) = 0, \ 
\int_{[0,1]^d} \sum_{j=1}^d \Bigl| \frac{\partial f}{\partial x_j}(x) \Bigr|^2 dx < \infty \Biggr\},
\]
corresponding the RKHS induced by the Sobolev kernel. Define the function
\begin{equation}
    \label{specific density} 
    p_{c_1,\dots, c_d}(x) :=\prod_{j=1}^d \left\{1 +  m_j(x_j)\right\} \ \mathrm{with} \ |c_j| \leq 1,\ \forall  j =1,\cdots, d,
\end{equation}
where $m_j(x) = 2c_jx\mathbf{1}{\{0\leq x<0.5\}} +2c_j(x-1)\mathbf{1}{\{0.5\leq x\leq 1\}} $. It is straightforward to verify that $\int_{[0,1]^d} p_{c_1,\dots,c_d}(x)dx = 1$, and the constraint $|c_j| \le 1$ for each $1\leq j\leq d$ guarantees that $p_{c_1,\dots,c_d}(x) \ge 0$. Hence, the functions $p_{c_1,\cdots,c_d}(x)$ defined in \eqref{specific density} form a family of probability density functions. Since the derivative of a polynomial is again a polynomial, it is square-integrable on any bounded interval. Therefore, for the density family defined in \eqref{specific density}, the Radon–Nikodym derivative with respect to the uniform distribution on $[0,1]^d$ satisfies
\[
\frac{dP_{c_1,\dots,c_d}}{dP_0} - 1 = p_{c_1,\dots,c_d} - 1 \in H_0^1([0,1]^d).
\]
\begin{figure}[t]
    \centering
    \begin{subfigure}[b]{0.48\textwidth}
        \centering
        \includegraphics[width=\textwidth]{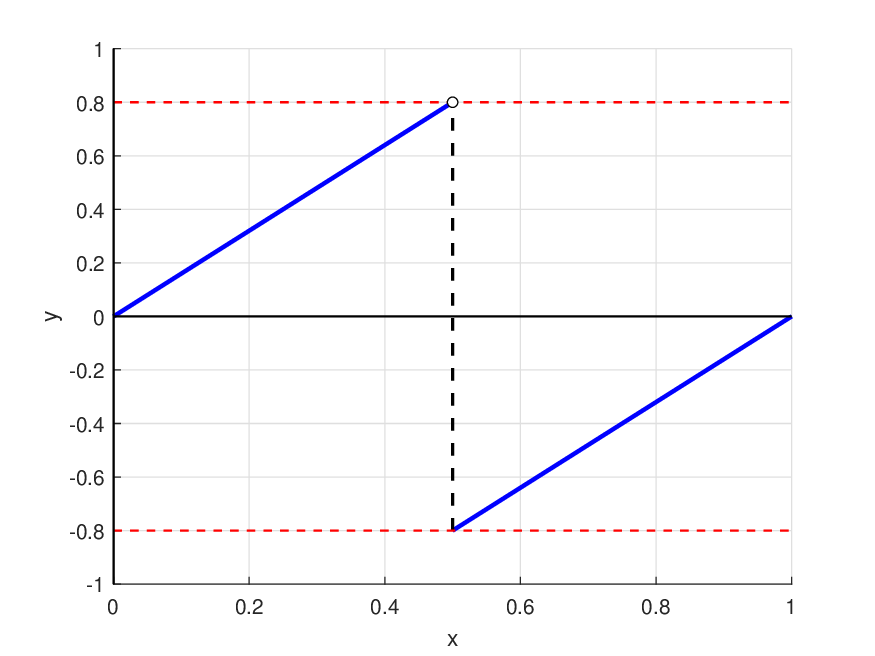}
        \caption{ }
        \label{fig: mx}
    \end{subfigure}
    \hfill
    \begin{subfigure}[b]{0.48\textwidth}
        \centering
        \includegraphics[width=\textwidth]{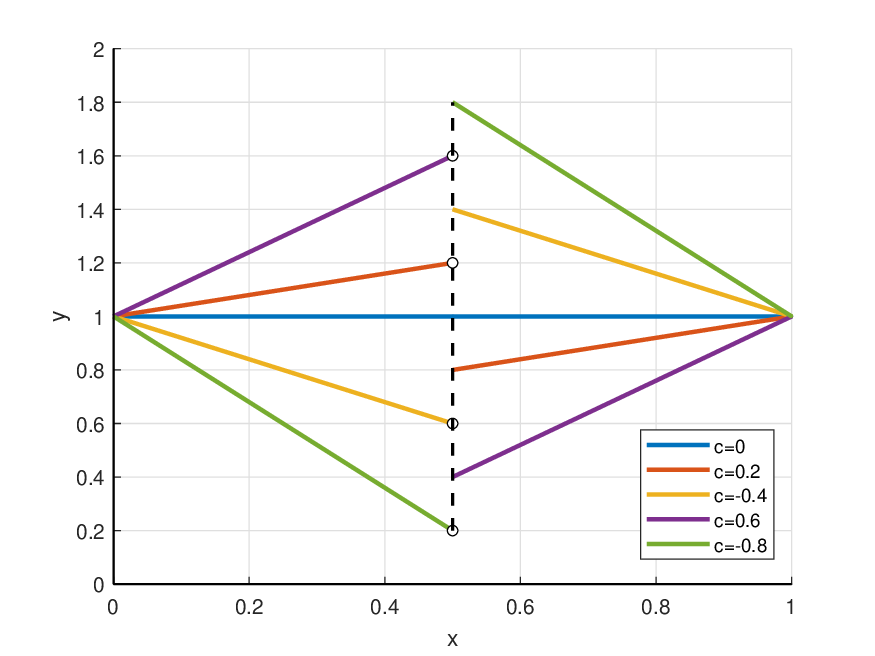}
        \caption{}
        \label{fig: density_functions}
    \end{subfigure}
    \caption{(a) $m(x)$ for $c = 0.8$; (b) Alternative density functions for various $c$ values.}
    \label{fig: density functions}
\end{figure}

In this example, the null distribution $P_0$ corresponds to the uniform distribution on $[0,1]^d$, and the alternatives are defined by the density in \eqref{specific density}. For simplicity, we take $c_1 = \cdots = c_d = c$ with $|c| \le 1$ and denote $P_{c_1,\dots,c_d}$ and $p_{c_1,\dots,c_d}$ by $P_c$ and $p_c$, respectively. The magnitude of $|c|$ controls the deviation of $p_c$ from the uniform density: larger $|c|$ yields a density further from $P_0$, whereas $c=0$ recovers the uniform case. Figure~\ref{fig: density functions} shows the density functions of the one-dimensional alternatives for different values of $c$. For the simulations, we set the number of permutations to $B = 400$ to determine the critical value, the data dimension $d $ takes values in $\{20,50\}$, and the sample sizes $(n,m,N) = (500,1000,100)$. All simulations are repeated $200$ times to compute empirical sizes and powers.

The resulting power curves are reported in Figure~\ref{fig: effect of regularization}. We observe that all tests maintain adequate Type I error control. Moreover, different levels of regularization (i.e., different choices of $\lambda$) affect the decay of the kernel eigenvalues in distinct ways, which in turn leads to varying improvements in power. This illustrates the important role of regularization in enhancing the performance of the test.

\begin{figure}[ht]
\centering
\includegraphics[width=0.9\textwidth]{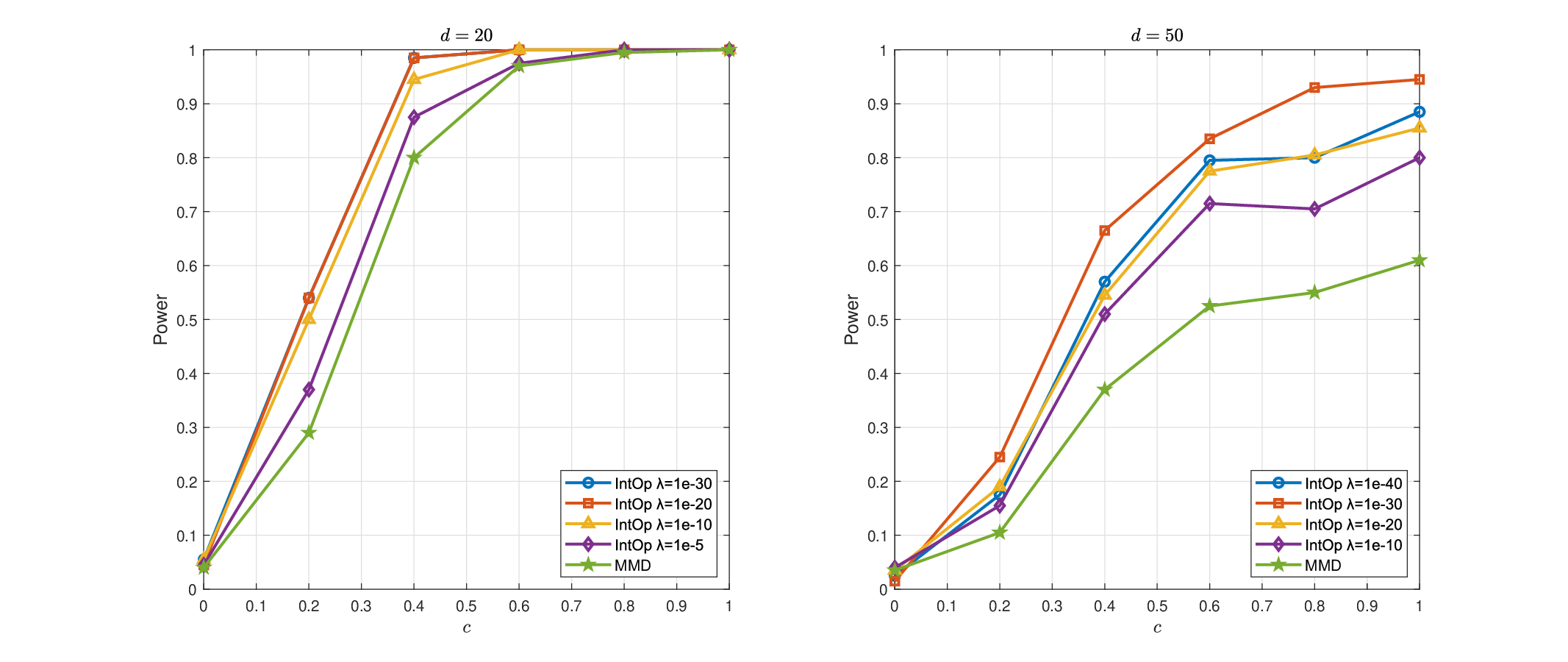}
\caption{Power comparison using Tikhonov regularization under different regularization parameters.}
\label{fig: effect of regularization}
\end{figure}

\subsection{Comparison with State-of-the-Art Testing Methods}
\label{subsec: experiment II}
In this subsection, we conduct four  experimental scenarios to further access the testing performance of the proposed method on general data types. As illustrated in the our theoretical results, the optimal regularization parameter depends on the regularities about the underlying distributions and the kernel function, which are typically unattainable in advance. To address this issue, we adopt a kind of adaptive strategy proposed in~\citep{hagrass-two-sample-test,hagrass-gof-test}. Specifically, we use the Gaussian kernel for these simulations defined as $K(x,y) = \exp\{-\|x-y\|^2/(2h)\}$, where $h$ is the bandwidth. For our statistic and that proposed in~\citep{hagrass-gof-test}, we construct the adaptive test by taking the
union of tests jointly over $\lambda \in \Lambda$ and $h \in H$ suggested in (Hagrass et al., 2024), that is, if $T_{\lambda, h}$ is the statistic computed on $\lambda$ and $h$, then we reject $H_0$ if and only if $T_{\lambda, h}\geq \hat{q}_{1-\frac{\alpha}{|\Lambda||H|}}$ for any $(\lambda, h) \in \Lambda \times H$, where $\Lambda = \{\lambda_L, 2\lambda_L,\cdots, \lambda_U\}$, $H = h_m\cdot \{w_L, 2w_L,\cdots, w_U\}$ and $h_m$ is determined by the median heuristic $h_m:=\mathrm{median}\{\|a-a^{\prime}\|_2^2: a,a^\prime \in D \cup y^m, a\neq a^{\prime}\}$ suggested in~\citep{gretton2012kernel}. The grid of parameters are set to $(\lambda_L, \lambda_U) = (10^{-6}, 5)$ and $(w_L, w_U) = (10^{-2}, 10^2)$, respectively. Tikhonov regularization approach is employed. For further comparison, we also consider several existing GOF tests from the literature, including the aggregated MMD test~\citep{schrab2023mmd} and the Energy test~\citep{szekely2004testing}. The aggregated MMD test adopts the product Gaussian kernel function in a more precise form of $K(x,y) = \prod_{i = 1}^d\exp\{-(x_i-y_i)^2/\sigma_i^2\}$. Since \citep{schrab2023mmd} presents various versions
of the MMD test, we compare our results to the uniform version of MMD, as down in~\citep{hagrass-gof-test}.  We set the number of permutations $(B_1, B_2, B_3) = (500,500,100)$ for the aggregated MMD test as suggested in (Schrab et al, 2023), and $B = 400$ for anther three methods to determine the critical value.  The sample size is set to $(n,m,N)=(200, 400, 100)$. All simulations are repeated $200$ times to compute empirical sizes and powers.

The experimental data are drown from $d$-dimensional distributions with $d $ taking values in $ \{10, 20, 50, 100\}$. Four types of data distribution are considered as follows:

\begin{enumerate}
    \item[(1)] \textbf{Mean shift in Gaussian distribution:} samples are drawn from $N(\mu, I_d)$, where the null hypothesis $H_0$ corresponds to $\mu = 0$ and the alternative $H_1$ corresponds to $\mu \neq  0$. 
    
    \item[(2)] \textbf{Variance change in Gaussian distribution:} samples follow $N(0, \sigma^2 I_d)$, where the null hypothesis $H_0$ corresponds to $\sigma = 1$ and the alternative $H_1$ corresponds to  $\sigma \neq  1$.
    
    \item[(3)]  \textbf{Support expansion in uniform distribution:} samples follow $\mathrm{Uniform}([0,\theta]^d)$, where $H_0$ corresponds to $\theta = 1$ and $H_1$ to $\theta \neq  1$. 
    
    \item[(4)]  \textbf{Concentration change in von Mises–Fisher (vMF) distribution:} samples follow the vMF distribution with density $f_d(x;\mu,\kappa) = C_d(\kappa) \exp(\kappa \mu^{\top} x)$ for $x \in \mathbb{S}^{d-1}$, where $\mu = 1_d/\sqrt{d}$, $\kappa \ge 0$, and $
    C_d(\kappa) = 2\pi \left(\frac{\kappa}{2\pi}\right)^{d/2-1} B_{d/2-1}^{-1}(\kappa)$
    with $B_v$ denoting the modified Bessel function of the first kind of order $v$. Here, $H_0$ corresponds to $\kappa = 0$, representing the uniform distribution on the sphere, while $H_1$ corresponds to $\kappa > 0$, representing increasing concentration around $\mu$. 
\end{enumerate}
The empirical power curves are shown in Figures~\ref{fig: new_big_four_case_mu}--\ref{fig: new_big_four_case_kappa}. Under the null hypothesis, all tests are based on permutation procedures and thus achieve satisfactory control of the Type I error. Under the alternative hypothesis, across different data types and dimensions, our proposed method and the method of \citet{hagrass-gof-test} demonstrate similar power, generally outperforming the MMD test while performing comparably to the Energy test in most scenarios. In particular, for the Gaussian variance alternatives and uniform alternatives shown in Figure~\ref{fig: new_big_four_case_sigma2}-\ref{fig: new_big_four_case_theta}, our method clearly demonstrates higher power than the other three methods. Overall, the results indicate that our method maintains robust Type~I error control while achieving competitive power across a variety of general data settings and dimensions compared with existing state-of-the-art methods, and even surpasses them in certain cases.

 \begin{figure}[ht]
    \centering
    \includegraphics[width=0.8\textwidth]{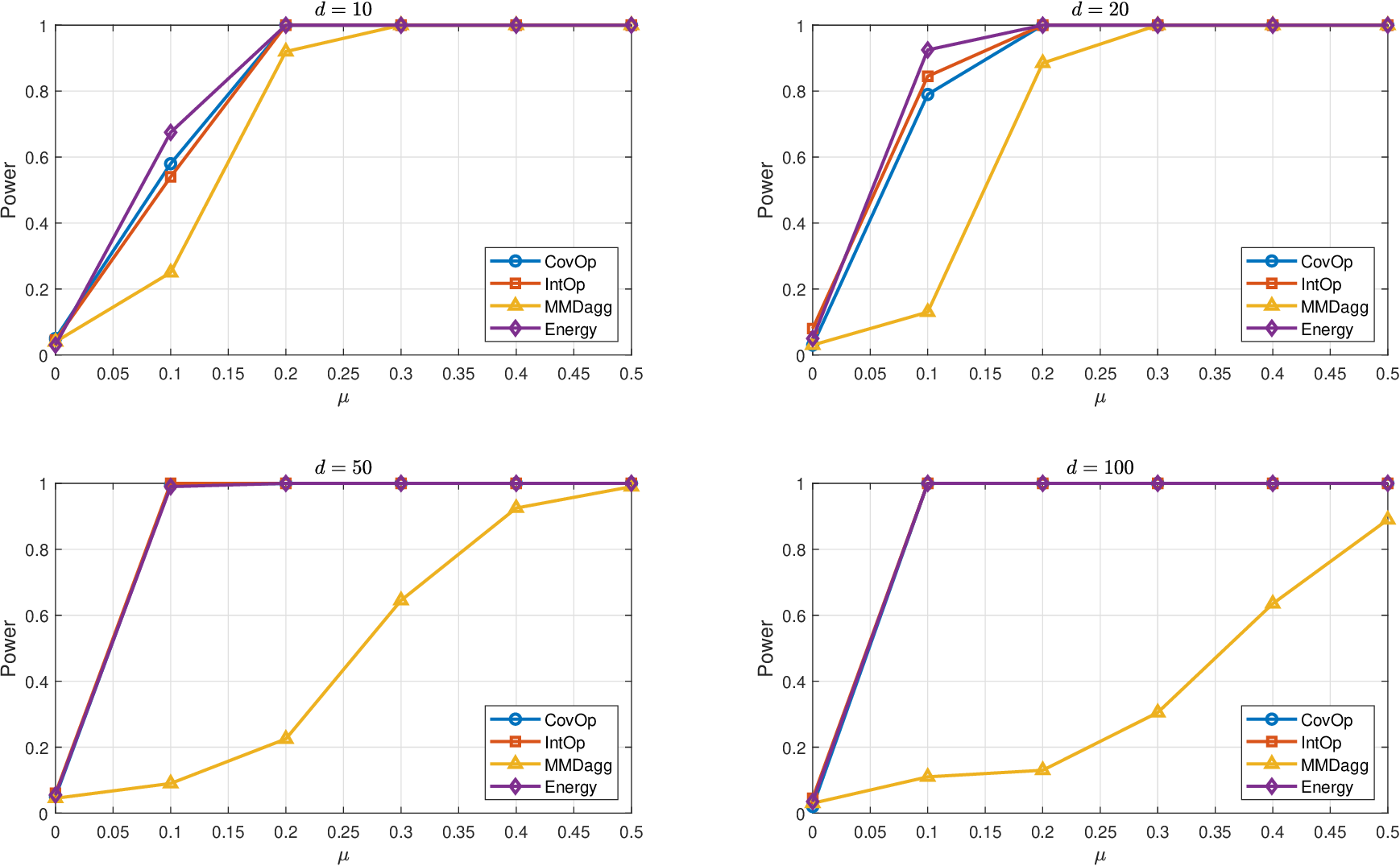} 
    \caption{Power comparison across dimensions under Gaussian mean alternatives.}
    \label{fig: new_big_four_case_mu}
\end{figure}
 \begin{figure}[ht]
    \centering
    \includegraphics[width=0.8\textwidth]{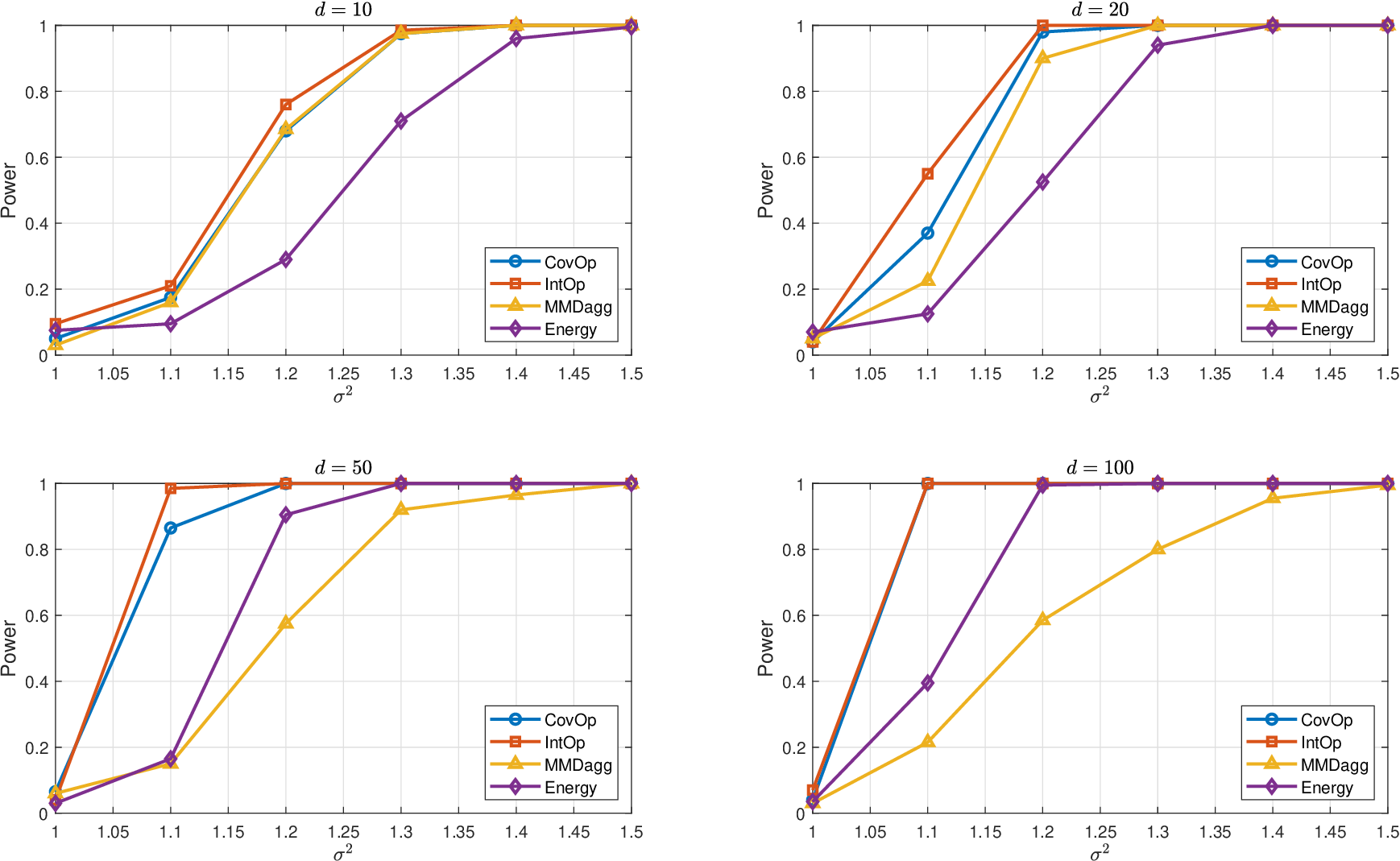} 
    \caption{Power comparison across dimensions under Gaussian variance alternatives.}
    \label{fig: new_big_four_case_sigma2}
\end{figure}
 \begin{figure}[ht]
    \centering
    \includegraphics[width=0.8\textwidth]{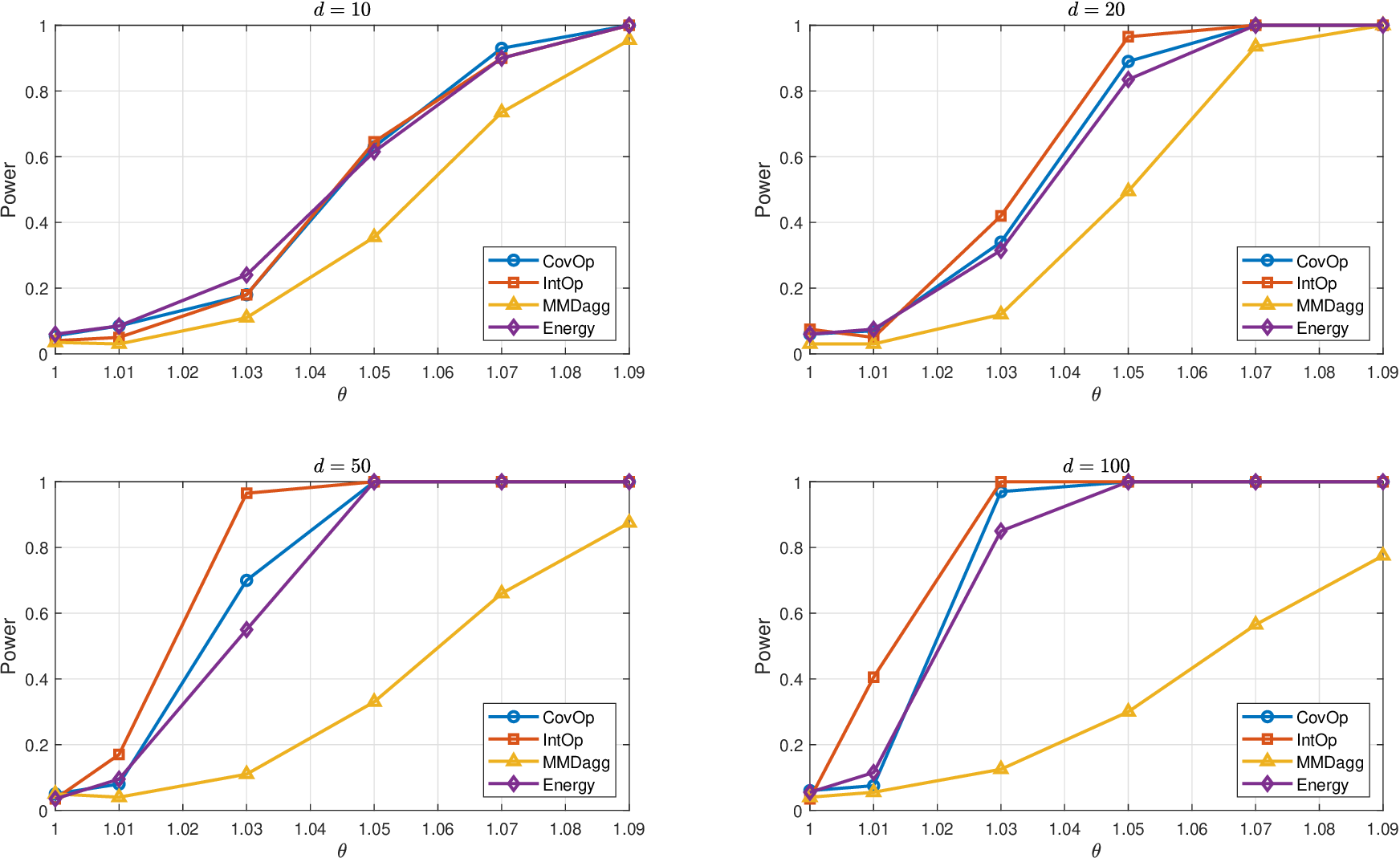} 
    \caption{Power comparison across dimensions under uniform scale alternatives.}
    \label{fig: new_big_four_case_theta}
\end{figure}
 \begin{figure}[ht]
    \centering
    \includegraphics[width=0.8\textwidth]{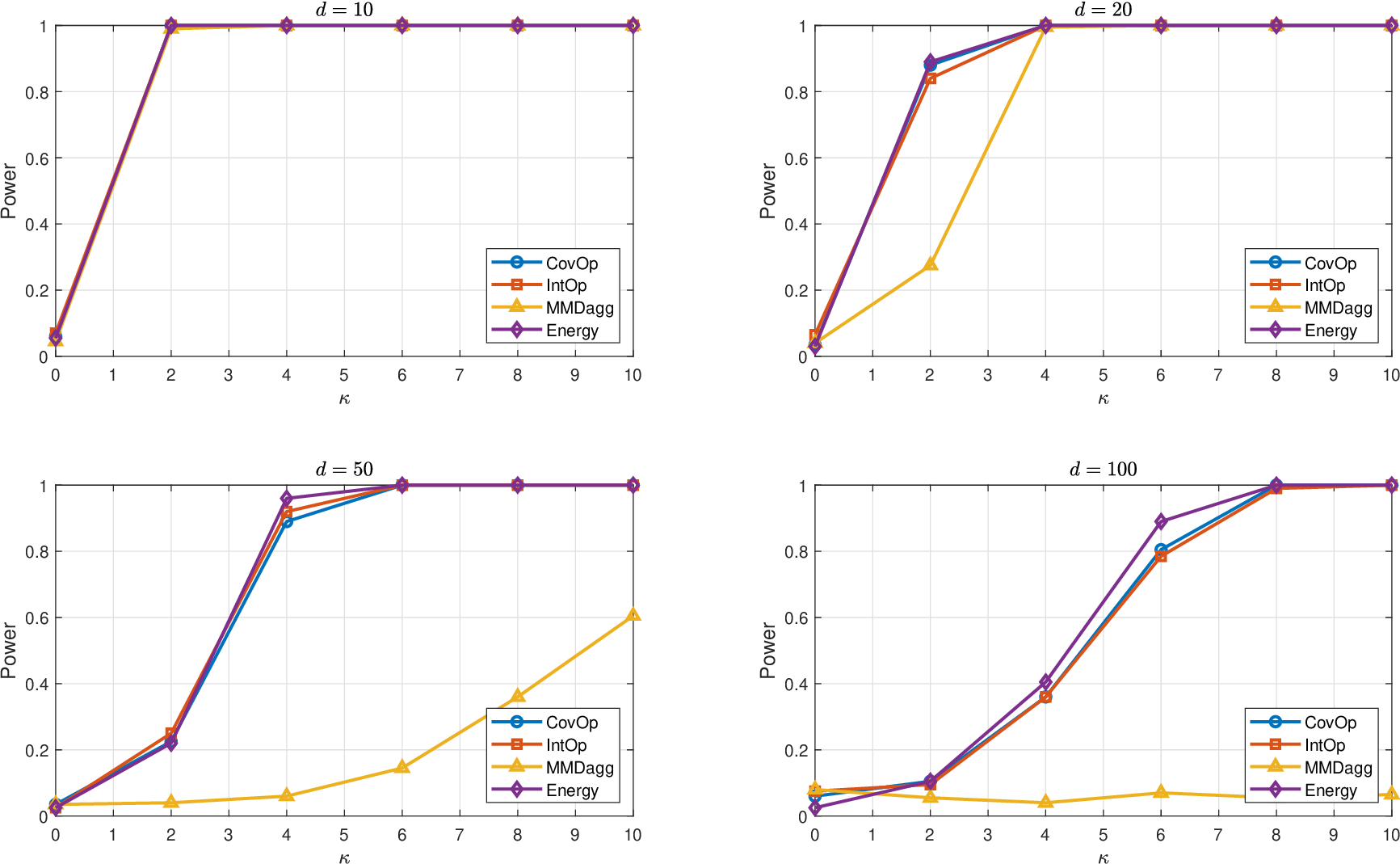} 
    \caption{Power comparison across dimensions under vMF concentration alternatives.}
    \label{fig: new_big_four_case_kappa}
\end{figure}
\subsection{Generality across Spectral Filters}
In this subsection, we examine whether our testing framework remains effective when equipped with different spectral filter functions.
Specifically, we consider two additional and widely used filters: the spectral cut-off filter defined in~\eqref{equ: cut-off}, and the Landweber iteration filter with
$g_{\lambda}(x) = \sum_{\ell=0}^{1/\lambda-1} (1-x)^\ell$.
We follow the same adaptive testing procedure as in the previous subsection, adjusting the regularization parameter grid to $(\lambda_L, \lambda_U) = (10^{-6}, 10^{-3})$ for these two filters.
The comparison of our approach with different filters, along with the aggregated MMD and Energy tests,  is conducted under uniform alternatives. The results, reported in Figure~\ref{fig: different_filter_results_theta}, indicate that the testing performance is largely insensitive to the choice of spectral filter. These experiments further confirm the generality of the proposed methods.
In particular, the inclusion of the spectral cut-off filter, which violates the additional assumption~\eqref{additional assumption} required in prior studies~\citep{hagrass-two-sample-test, hagrass-gof-test}, shows that our approach remains theoretically valid and empirically effective.
 Our approach achieves competitive power across different filters and scenarios, confirming its generality beyond specific regularization schemes.

\begin{figure}[ht]
    \centering
    \includegraphics[width=0.8\textwidth]{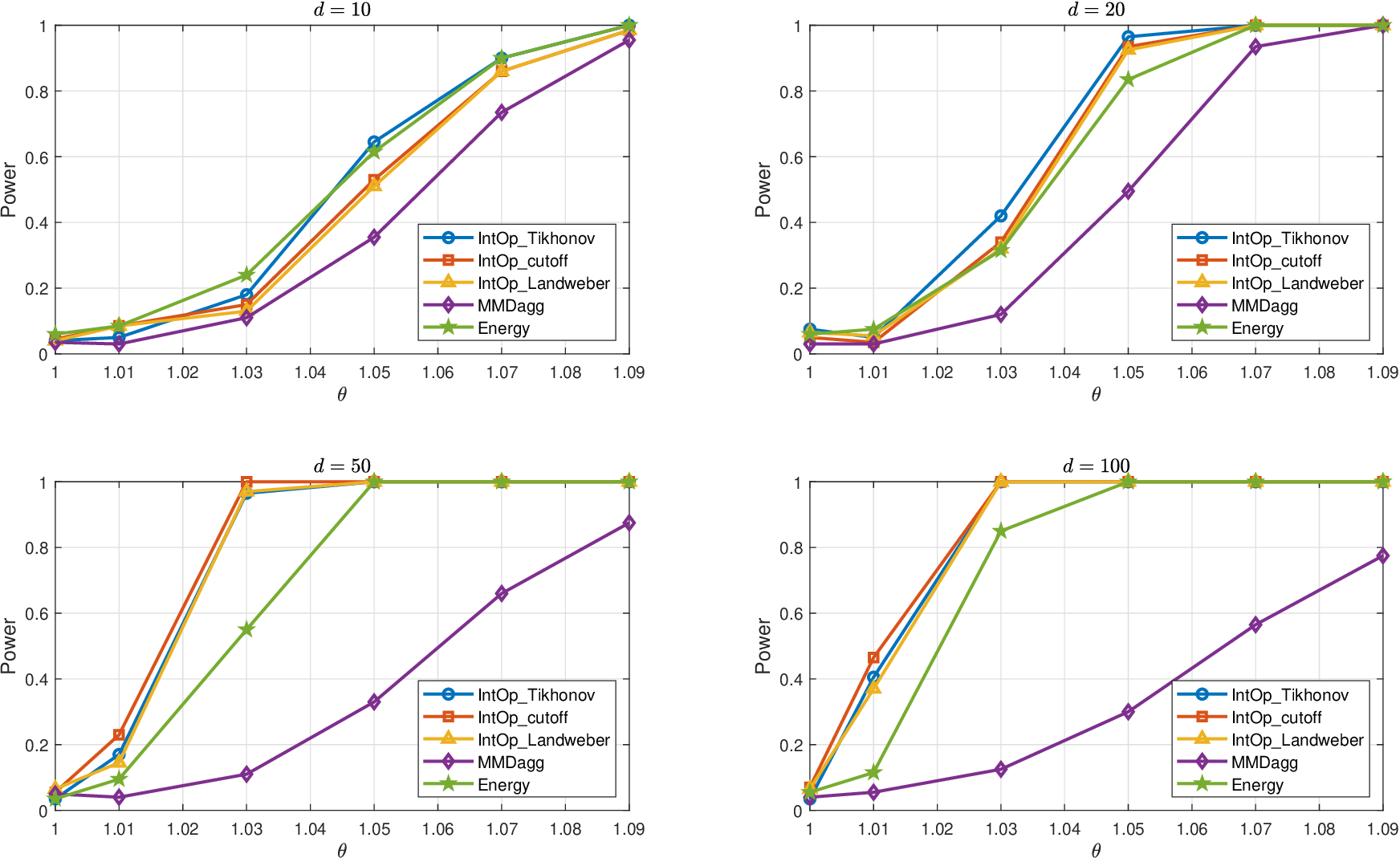} 
    \caption{Power comparison across dimensions and filter functions under uniform scale alternatives.}
    \label{fig: different_filter_results_theta}
\end{figure}

\section{Error Decomposition}
\label{sec: error decomposition}
The proposed two-sample statistic $\widehat{\xi}_{\lambda}(P,P_0)$ serves as a regularized estimator of $\chi^2(P,P_0)$. In this section, we provide a detailed error decomposition of this statistic. We begin by defining a semi-population version of the statistic:
\begin{equation*}
\xi_{\lambda}^{\star}(P,P_0) := \mathbb{E}_D\left\{\widehat{\xi}_{\lambda}(P,P_0)\right\} = \left\langle g_{\lambda}^{1/2}(L_{K,D})(\mu_P - \mu_0), , g_{\lambda}^{1/2}(L_{K,D})(\mu_P - \mu_0)\right\rangle_{K},
\end{equation*}
where $\mathbb{E}_D(\cdot) := \mathbb{E}(\cdot \mid D)$ denotes the conditional expectation given the sample $D$. This leads to a natural decomposition of the total estimation error into two components:
\begin{equation}
\label{equ: error decomposition}
\widehat{\xi}_{\lambda}(P,P_0) - \chi^2(P,P_0)\quad =\quad
\underbrace{\widehat{\xi}_{\lambda}(P,P_0) - \xi^{\star}_{\lambda}(P,P_0)}_{\textit{Sample error}} \quad +\quad 
\underbrace{\xi^{\star}_{\lambda}(P,P_0) - \chi^2(P,P_0)}_{\textit{Approximation error}}.
\end{equation}
The \emph{sample error}, which mainly reflects the \emph{variance} due to the randomness in estimating kernel mean embeddings, will be analyzed in Subsection~\ref{subsec: sample error}. The \emph{approximation error}, which primarily corresponds to the \emph{bias} introduced by the regularization scheme, will be characterized in Subsection~\ref{subsec: approximation error}. Both error terms can be further described via two key quantities: the \emph{distributional discrepancy} between $P$ and $P_0$, and the \emph{operator similarity} between the empirical operator $L_{K,D}$ and its population counterpart $L_K$. These fundamental quantities will be discussed in Subsection~\ref{subsec: distributional discrepancy and operator similarity}. 

\subsection{Sample Error}
\label{subsec: sample error}
 Let $P_0$ be the null distribution, and $Q$ be another probability distribution defined on $\mathcal{X}$, where $Q$ may coincide with either $P$ or $P_0$. Denote by $\mu_Q := \int_{\mathcal{X}}K_xdQ(x)$ the kernel mean embedding of $Q$. We further define
 \begin{align}
 \label{equ: ThetaldQP0}
     \Theta_{\lambda,Q,P_0}:= (L_K + \lambda I)^{-1/2}\Sigma_{Q}(L_K + \lambda I)^{-1/2},
 \end{align}
 where $\Sigma_Q:= \int_{\mathcal{X}}(K_x - \mu_Q)\otimes(K_x - \mu_Q) dQ(x)$ is the covariance operator on $\mathcal{H}_K$ with respect to the distribution $Q$. 

The following lemma plays a central role in our analysis of the sample error, as it provides deterministic bounds for certain types of $U$-statistics that appear when decomposing the sample error via conditional expectation. The proof mainly relies on moment calculations of $U$-statistics and follows similar arguments to Lemma~A.1 in \citep{hagrass-gof-test} and Lemmas~A.4--A.6 in \citep{hagrass-two-sample-test}.

 \begin{lemma}
 \label{lemma: sample error}
  Let $\{x_i\}_{i=1}^n \sim Q$, $\{y_j\}_{j=1}^m \sim P_0$ and  $D := \{z_\ell\}_{\ell=1}^N\sim P_0$ be independent samples, and $n,m \geq 2$.
 Define
 $$
\begin{aligned}
    U_{Q,P_0,1} &:= \frac{1}{n(n-1)}\sum_{i=1}^n \sum_{ j\neq i}^n\left\langle g_{\lambda}^{1/2}(L_{K,D})(K_{x_i} -\mu_Q),  g_{\lambda}^{1/2}(L_{K,D})(K_{x_j} - \mu_Q)\right \rangle_{K},\\
    U_{Q,P_0,2} &:= \frac{1}{nm}\sum_{i=1}^n \sum_{ j=1 }^m\left\langle g_{\lambda}^{1/2}(L_{K,D})(K_{x_i} -\mu_Q),  g_{\lambda}^{1/2}(L_{K,D})(K_{y_j} - \mu_0)\right \rangle_{K},\\
    U_{Q,P_0,3} &:= \frac{1}{n}\sum_{i=1}^n \left\langle g_{\lambda}^{1/2}(L_{K,D})(K_{x_i} -\mu_Q),  g_{\lambda}^{1/2}(L_{K,D})(\mu_{P} -\mu_{0})\right \rangle_{K},
\end{aligned}
$$
where $L_{K,D}:\mathcal{H}_K \to \mathcal{H}_K$ is defined in \eqref{empirical integral operator}. Then, there holds
\begin{enumerate}
    \item [(1)] $\mathbb E_D(|U_{Q,P_0,1}|) \leq \frac{3b}{n-1}\| \Theta_{\lambda,Q,P_0}\|_{HS}\mathscr{P}_{D,\lambda}^2$;
    \item [(2)] $\mathbb E_D(|U_{Q,P_0,2}|) \leq \frac{2b}{\sqrt{nm}}\| \Theta_{\lambda,Q,P_0}\|_{HS}^{1/2}\|\Theta_{\lambda,P_0,P_0}\|_{HS}^{1/2}\mathscr{P}_{D,\lambda}^2$;
    \item [(3)] $\mathbb E_D(|U_{Q,P_0,3}|) \leq \frac{2b}{\sqrt{n}}\| \Theta_{\lambda,Q,P_0}\|^{1/2}\|f\|_{L^2(P_0)}\mathscr{P}_{D,\lambda}^2$,
\end{enumerate}
where 
\begin{equation}
    \label{equ: PDL}
    \mathscr{P}_{D,\lambda}:= \left\|(L_{K,D} + \lambda I)^{-1/2}(L_K + \lambda I)^{1/2} \right\|.
\end{equation}
 \end{lemma}

\begin{proof}
By independence across different data sets, $U_{Q,P_0,1}$, $U_{Q,P_0,2}$, and $U_{Q,P_0,3}$ all have zero mean. By Lemma A.3 in \citep{hagrass-two-sample-test}, the conditional variances of these three terms given $D$ can be expressed explicitly as follows:
\begin{align*}
    n(n-1)/2\cdot \mathbb{E}_D(U_{Q,P_0,1}^2) &=  \mathbb E_D\left[\left\langle g_{\lambda}^{1/2}(L_{K,D})(K_{x} -\mu_Q), g_{\lambda}^{1/2}(L_{K,D})(K_{x^{\prime}} - \mu_Q)\right \rangle_{K}^2\right],\\
   nm\cdot \mathbb{E}_D(U_{Q,P_0,2}^2) &= \mathbb{E}_D \left[\left \langle g_{\lambda}^{1/2}(L_{K,D})(K_{x} - \mu_Q) ,g_{\lambda}^{1/2}(L_{K,D})(K_{y} - \mu_0) \right \rangle_K^2\right],\\
    n\cdot \mathbb{E}_D(U_{Q,P_0,3}^2) &=\mathbb{E}_D\left[\left \langle g_{\lambda}^{1/2}(L_{K,D})(K_{x} - \mu_Q) ,g_{\lambda}^{1/2}(L_{K,D})(\mu_P - \mu_0) \right \rangle_K^2\right],
\end{align*}
where $x,x^{\prime} \sim Q$ and $y\sim P_0$. Now we compute these conditional variances explicitly.

For the first term, we have
\begin{equation}
\label{equ: condtional variance part 1}
\begin{aligned}
 & \quad \,\,\mathbb{E}_D \left[\left \langle g_{\lambda}^{1/2}(L_{K,D})(K_{x} - \mu_Q) ,g_{\lambda}^{1/2}(L_{K,D})(K_{x^{\prime}} - \mu_Q )\right \rangle_K^2\right]\\
	& =\mathbb{E}_D \left[ \left \langle g_{\lambda}^{1/2}(L_{K,D})(K_{x} - \mu_Q)\otimes (K_{x} - \mu_Q) g_{\lambda}^{1/2}(L_{K,D}),\right. \right. \\
	& \left.\left.\qquad \qquad  g_{\lambda}^{1/2}(L_{K,D})(K_{x^{\prime}} - \mu_Q)\otimes (K_{x^{\prime}} - \mu_Q) g_{\lambda}^{1/2}(L_{K,D})\right \rangle_{HS}\right]\\
	& = \left \langle g_{\lambda}^{1/2}(L_{K,D})\mathbb{E}(K_{x} - \mu_Q)\otimes (K_{x} - \mu_Q) g_{\lambda}^{1/2}(L_{K,D}) ,\right.  \\
	& \left.\qquad \qquad g_{\lambda}^{1/2}(L_{K,D})\mathbb{E}(K_{x^{\prime}} - \mu_Q)\otimes (K_{x^{\prime}} - \mu_Q) g_{\lambda}^{1/2}(L_{K,D})\right \rangle_{HS}\\
	&  = \left\| g_{\lambda}^{1/2}(L_{K,D})(L_K + \lambda I)^{1/2}(L_K + \lambda I)^{-1/2}\Sigma_Q (L_K + \lambda I)^{-1/2}(L_K + \lambda I)^{1/2}g_{\lambda}^{1/2}(L_{K,D}) \right\|_{HS}^2\\
	& \leq \left\| g_{\lambda}^{1/2}(L_{K,D})(L_{K,D} + \lambda I)^{1/2}\right\|^4\cdot \left\|(L_{K,D} + \lambda I)^{-1/2}(L_K + \lambda I)^{1/2}\right\|^4\cdot \\
	& \quad \,\left\| (L_K + \lambda I)^{-1/2}\Sigma_Q (L_K + \lambda I)^{-1/2}\right\|_{HS}^2\\
	& \leq 4b^2  \left\|\Theta_{\lambda,Q,P_0}\right\|_{HS}^2 \mathscr{P}_{D,\lambda}^4,
\end{aligned}
\end{equation}
where the first step follows using Lemma~\ref{lemma: Balazs2005}, and the last step follows by \eqref{equ: Cordes inequality} of Lemma~\ref{lemma: basic operator inequalities} and \eqref{equ: filter property 1}. For the second term, we have
\begin{equation}
\label{equ: condtional variance part 2}
\begin{aligned}
 & \quad \,\,\mathbb{E}_D \left[\left \langle g_{\lambda}^{1/2}(L_{K,D})(K_{x} - \mu_Q) ,g_{\lambda}^{1/2}(L_{K,D})(K_{y} - \mu_0) \right \rangle_K^2\right]\\
	& =\mathbb{E}_D \left[ \left \langle g_{\lambda}^{1/2}(L_{K,D})(K_{x} - \mu_Q)\otimes (K_{x} - \mu_Q) g_{\lambda}^{1/2}(L_{K,D}),\right. \right. \\
	& \left.\left.\qquad \qquad  g_{\lambda}^{1/2}(L_{K,D})(K_{y} - \mu_0)\otimes (K_{y} - \mu_0) g_{\lambda}^{1/2}(L_{K,D})\right \rangle_{HS}\right]\\
	& = \left \langle g_{\lambda}^{1/2}(L_{K,D})\mathbb{E}(K_{x} - \mu_Q)\otimes (K_{x} - \mu_Q) g_{\lambda}^{1/2}(L_{K,D}) ,\right.  \\
	& \left.\qquad \qquad g_{\lambda}^{1/2}(L_{K,D})\mathbb{E}(K_{y} - \mu_0)\otimes (K_{y} - \mu_0) g_{\lambda}^{1/2}(L_{K,D})\right \rangle_{HS}\\
	& \leq  \left\| g_{\lambda}^{1/2}(L_{K,D})(L_K + \lambda I)^{1/2}(L_K + \lambda I)^{-1/2}\Sigma_Q (L_K + \lambda I)^{-1/2}(L_K + \lambda I)^{1/2}g_{\lambda}^{1/2}(L_{K,D}) \right\|_{HS}\cdot\\
	& \quad \,\left\| g_{\lambda}^{1/2}(L_{K,D})(L_K + \lambda I)^{1/2}(L_K + \lambda I)^{-1/2}\Sigma_0 (L_K + \lambda I)^{-1/2}(L_K + \lambda I)^{1/2}g_{\lambda}^{1/2}(L_{K,D}) \right\|_{HS}\\
	& \leq \left\| g_{\lambda}^{1/2}(L_{K,D})(L_{K,D} + \lambda I)^{1/2}\right\|^4\cdot \left\|(L_{K,D} + \lambda I)^{-1/2}(L_K + \lambda I)^{1/2}\right\|^4\cdot \\
	& \quad \,\left\| (L_K + \lambda I)^{-1/2}\Sigma_Q (L_K + \lambda I)^{-1/2}\right\|_{HS}\left\|(L_K + \lambda I)^{-1/2}\Sigma_0 (L_K + \lambda I)^{-1/2}\right\|_{HS}\\
	& \leq 4b^2  \left\|\Theta_{\lambda,Q,P_0}\right\|_{HS}\left\|\Theta_{\lambda,P_0,P_0}\right\|_{HS} \mathscr{P}_{D,\lambda}^4,
\end{aligned}
\end{equation}
where the first step uses Lemma~\ref{lemma: Balazs2005}, the thrid step is an application of the Cauchy-Schwarz's inequality, and the last step follows by \eqref{equ: Cordes inequality} of Lemma~\ref{lemma: basic operator inequalities} and \eqref{equ: filter property 1}. For the third term, we have
\begin{equation}
\label{equ: condtional variance part 3}
\begin{aligned}
 & \quad \,\,\mathbb{E}_D\left[\left \langle g_{\lambda}^{1/2}(L_{K,D})(K_{x} - \mu_Q) ,g_{\lambda}^{1/2}(L_{K,D})(\mu_P - \mu_0) \right \rangle_K^2\right]\\
	& = \mathbb{E}_D \left[ \left \langle g_{\lambda}^{1/2}(L_{K,D})(K_{x} - \mu_Q)\otimes (K_{x} - \mu_Q) g_{\lambda}^{1/2}(L_{K,D}),\right. \right. \\
	& \left.\left.\qquad \qquad  g_{\lambda}^{1/2}(L_{K,D})(\mu_P - \mu_0)\otimes (\mu_P - \mu_0) g_{\lambda}^{1/2}(L_{K,D})\right \rangle_{HS}\right]\\
	& = \left \langle g_{\lambda}^{1/2}(L_{K,D})\mathbb{E}(K_{x} - \mu_Q)\otimes (K_{x} - \mu_Q) g_{\lambda}^{1/2}(L_{K,D}) ,\right.  \\
	& \left.\qquad \qquad g_{\lambda}^{1/2}(L_{K,D})(\mu_P - \mu_0)\otimes (\mu_P - \mu_0) g_{\lambda}^{1/2}(L_{K,D})\right \rangle_{HS}\\
	&= \mathrm{Tr}\left[ g_{\lambda}^{1/2}(L_{K,D})(\mu_P - \mu_0)\otimes (\mu_P - \mu_0) g_{\lambda}^{1/2}(L_{K,D})g_{\lambda}^{1/2}(L_{K,D})\Sigma_Q g_{\lambda}^{1/2}(L_{K,D}) \right]\\
	&= \mathrm{Tr}\left[(L_K + \lambda I)^{1/2} g_{\lambda}(L_{K,D})(\mu_P - \mu_0)\otimes (\mu_P - \mu_0) g_{\lambda}(L_{K,D})(L_K + \lambda I)^{1/2}\right.\\
	& \qquad\ \ \left.(L_K + \lambda I)^{-1/2}\Sigma_Q (L_K + \lambda I)^{-1/2} \right]\\
	& \leq \left\|(L_K + \lambda I)^{1/2} g_{\lambda}(L_{K,D})(L_K + \lambda I)^{1/2}(L_K + \lambda I)^{-1/2}(\mu_P - \mu_0)\right\|_K^2\left\|\Theta_{\lambda, Q,P_0}\right\|\\
	& \leq \left\|g_{\lambda}^{1/2}(L_{K,D})(L_{K,D} +\lambda I)^{1/2}(L_{K,D} +\lambda I)^{-1/2}(L_{K} +\lambda I)^{1/2} \right\|^4\cdot\\
	&\quad \,\left\|(L_K + \lambda I)^{-1/2}L_{K}^{1/2}L_{K}^{1/2}f\right\|_K^2 \left\|\Theta_{\lambda, Q,P_0}\right\|\\
	& \leq 4b^2 \left\|\Theta_{\lambda, Q,P_0}\right\|\|f\|_{L^2(P_0)}^2\mathscr{P}_{D,\lambda}^4,
\end{aligned}
\end{equation}
where the first step  is by applying Lemma~\ref{lemma: Balazs2005}, the fifth step by applying Lemma~\ref{lemma: trace inequality} with $k = 1$, the sixth step follows by $L_Kf = \mu_P - \mu_0$ and the last steps follows by $\|(L_K+ \lambda I)^{-1/2}L_K^{1/2}\| \leq 1$, \eqref{equ: Cordes inequality} of Lemma~\ref{lemma: basic operator inequalities} and \eqref{equ: filter property 1}. These three estimates combining with the Jensen's inequality $\mathbb{E}(|X|) \leq [\mathbb{E}(X^2)]^{1/2}$ complete the proof.
\end{proof}


Based on Lemma~\ref{lemma: sample error}, we derive the decomposition of the sample error term in \eqref{equ: error decomposition} under conditional expectation. 
\begin{proposition}[Sample error]
\label{proposition: sample error decomposition} Let $n,m \geq 2$. Then, the sample error in \eqref{equ: error decomposition} satisfies
\begin{equation}
\label{equ: sample error decomposition}
    \begin{aligned}
            E_D\left[|\widehat{\xi}_{\lambda}(P,P_0)-{\xi}^{\star}_{\lambda}(P,P_0)|\right]&\leq 5b\mathscr{P}_{D,\lambda}^2\left(\frac{\|\Theta_{\lambda,P,P_0}\|_{HS}}{n-1}+\frac{\|\Theta_{\lambda,P_0,P_0}\|_{HS}}{m-1}\right)\\
            &  + 4b\mathscr{P}_{D,\lambda}^2\left(\frac{\|\Theta_{\lambda,P,P_0}\|^{1/2}}{\sqrt{n}}+\frac{\|\Theta_{\lambda,P_0,P_0}\|^{1/2}}{\sqrt{m}}\right)\|f\|_{L^2(P_0)}.
    \end{aligned}
\end{equation}

\end{proposition}
\begin{proof}
We also adopt the decomposition approach that used in~\citep{hagrass-two-sample-test,hagrass-gof-test} as follows,
$$
	\widehat{\xi}_{\lambda}(P,P_0)-{\xi}^{\star}_{\lambda}(P,P_0) = I_1 +I_2- I_3+I_4+I_5,
$$
where
\begin{equation}
	\label{equ: five terms}
	\begin{aligned}
	 I_1 &:=  \frac{1}{n(n-1)} \sum_{i=1}^n \sum_{ j\neq i}^n\left\langle g_{\lambda}^{1/2}(L_{K,D})(K_{x_i} -\mu_P),  g_{\lambda}^{1/2}(L_{K,D})(K_{x_j} - \mu_P)\right \rangle_{K},\\
		I_2 &:= \frac{1}{m(m-1)}\sum_{i=1}^m \sum_{ j\neq i}^m\left\langle  g_{\lambda}^{1/2}(L_{K,D})(K_{y_i} -\mu_0),  g_{\lambda}^{1/2}(L_{K,D})(K_{y_j} - \mu_0)\right \rangle_{K},\\
		I_3 &:=\frac{2}{nm} \sum_{i=1}^n\sum_{j=1}^m \left \langle  g_{\lambda}^{1/2}(L_{K,D})(K_{x_i} - \mu_P) , g_{\lambda}^{1/2}(L_{K,D})(K_{y_j} - \mu_0) \right \rangle_K, \\
	I_4 &:=  \frac{2}{n}\sum_{i=1}^n \left\langle  g_{\lambda}^{1/2}(L_{K,D})(K_{x_i} - \mu_P), g_{\lambda}^{1/2}(L_{K,D})(\mu_P -\mu_0)\right\rangle_K,\\
	I_5&:= \frac{2}{m}\sum_{i=1}^m \left\langle  g_{\lambda}^{1/2}(L_{K,D})(K_{y_i} - \mu_0),  g_{\lambda}^{1/2}(L_{K,D})(\mu_0 -\mu_P)\right\rangle_K.
	\end{aligned}
\end{equation}
Taking $Q = P$ in Lemma~\ref{lemma: sample error}, we have 
\begin{align*}
        E_{D}(|I_1|) &\leq \frac{3b}{n-1}\|\Theta_{\lambda,P,P_0}\|_{HS}\mathscr{P}_{D,\lambda}^2;\\
        E_{D}(|I_3|) &\leq \frac{4b}{\sqrt{nm}}\|\Theta_{\lambda,P,P_0}\|_{HS}^{1/2}\|\Theta_{\lambda,P_0,P_0}\|_{HS}^{1/2}\mathscr{P}_{D,\lambda}^2;\\
         E_{D}(|I_4|) & \leq \frac{4b}{\sqrt{n}}\| \Theta_{\lambda,P,P_0}\|^{1/2}\|f\|_{L^2(P_0)}\mathscr{P}_{D,\lambda}^2.
\end{align*}
Taking $Q = P_0$ and $n = m$ in Lemma~\ref{lemma: sample error}, we have 
\begin{align*}
    E_{D}(|I_2|) &\leq \frac{3b}{m-1}\|\Theta_{\lambda,P_0,P_0}\|_{HS}\mathscr{P}_{D,\lambda}^2;\\
     E_{D}(|I_5|) & \leq \frac{4b}{\sqrt{m}}\| \Theta_{\lambda,P_0,P_0}\|^{1/2}\|f\|_{L^2(P_0)}\mathscr{P}_{D,\lambda}^2.
\end{align*}
By the triangle inequality and $\sqrt{ab} \leq (a+b)/2$ for $a,b \geq 0$, we complete the proof.
\end{proof}

\subsection{Approximation Error}
\label{subsec: approximation error}
Our analysis for the approximation error in \eqref{equ: error decomposition} is essentially different from those in \citep{balasubramanian2021optimality, hagrass-two-sample-test, hagrass-gof-test}. To better illustrate the difference,
we further define
$$
\breve{\xi}_{\lambda}(P,P_0):
=  \left\|(L_{K} + \lambda I)^{-1/2}L_K f \right\|_K^2\quad \mathrm{and} \quad 
\widetilde{\mathscr{P}}_{D,\lambda} =\left\|(L_K+\lambda I)^{-1/2}(L_{K,D} + \lambda I)^{1/2} \right\|.
$$
When we take $g_{\lambda}(x)=(x+\lambda)^{-1}$, it follows that
\begin{equation}
\label{equ: upper bound}
\xi^{\star}_{\lambda}(P,P_0) 
= \left\langle g_{\lambda}(L_{K,D})L_K f,\, L_K f \right\rangle_{K}
= \left\|(L_{K,D} + \lambda I)^{-1/2}L_K f \right\|_K^2.
\end{equation}
Due to special spectral properties of Tikhonov regularization, there holds both
$$
 \xi^{\star}_{\lambda}(P,P_0) 
\leq \left\|(L_{K,D} + \lambda I)^{-1/2}(L_K + \lambda I)^{1/2} \right\|^2\breve{\xi}_{\lambda}(P,P_0) = \mathscr{P}_{D,\lambda}^2\breve{\xi}_{\lambda}(P,P_0)
$$
and, in particular,
\begin{equation}
\label{equ: lower bound}
     \xi^{\star}_{\lambda}(P,P_0) 
\geq \left\|(L_K+\lambda I)^{-1/2}(L_{K,D} + \lambda I)^{1/2} \right\|^{-2}\breve{\xi}_{\lambda}(P,P_0) = \widetilde{\mathscr{P}}_{D,\lambda}^{-2}\breve{\xi}_{\lambda}(P,P_0).
\end{equation}
 Both ${\mathscr{P}}_{D,\lambda}$ and $\widetilde{\mathscr{P}}_{D,\lambda}$ 
are bounded with high probability. However, this line of analysis cannot be directly extended to general spectral algorithms. 
Nevertheless, under the additional assumption in \eqref{additional assumption}, the filter functions considered in \citep{hagrass-two-sample-test,hagrass-gof-test} exhibit spectral behaviors analogous to those of Tikhonov regularization. 
In particular, due to the spectral similarity between $\hat{\Sigma}_0$ and $\Sigma_0$, \eqref{equ: filter property 1} yields a high-probability bound for
\[
\left\| g_{\lambda}^{1/2}(\hat{\Sigma}_0)\,(\Sigma_0 + \lambda I)^{1/2} \right\|,
\]
which plays essentially the same role as $\mathscr{P}_{D,\lambda}$, leading to a similar upper bound as  \eqref{equ: upper bound}. On the other hand, assumption \eqref{additional assumption} provides a uniform lower bound on $g_\lambda$, ensuring that the factor
\[
\left\|(\Sigma_0 + \lambda I)^{-1/2}\, g_{\lambda}^{-1/2}(\hat{\Sigma}_0)\right\|
\]
remains well controlled with high probability, in analogy with $\widetilde{\mathscr{P}}_{D,\lambda}$, thereby leading to the analogous lower bound in \eqref{equ: lower bound}. 
 Overall, in characterizing the approximation error, the main ideas in \citep{balasubramanian2021optimality} and in \citep{hagrass-two-sample-test,hagrass-gof-test} both first employ a \emph{multiplication-based} approach to establish the similarity between $\xi_{\lambda}^{\star}(P,P_0)$ and $ \breve{\xi}_{\lambda}(P,P_0)$. The subsequent analysis then reduces to controlling the similarity between $\breve{\xi}_{\lambda}(P,P_0)$ and $\chi^2(P,P_0)$, which can be handled directly at the population level.
 
 However, in the absence of the additional assumption~\eqref{additional assumption}, this strategy is no longer valid. To overcome this limitation, we adopt a \emph{difference-based } approach to directly characterize the similarity between $\xi_{\lambda}^{\star}(P,P_0)$ and $\chi^2(P,P_0)$.
The following lemma plays a central role in deriving the approximation error. The proof strategy follows techniques in the literature of kernel regression~\citep{guo2017learning}, together with inequalities for operator differences~\citep{Dicker2017}. 
For readability, the complete proof is provided in the Appendix~\ref{section: appendix B}.

\begin{lemma}
\label{lemma: core of approximation error}
Assume that $0<\lambda \leq 1$. For $0\leq u \leq \nu_g$ and $v >0$, define
$$
\mathscr{T}_{D,\lambda,u,v}:= \left\|(L_{K,D} + \lambda I)^{u}(g_{\lambda}(L_{K,D})L_{K,D} -I) (L_K + \lambda I)^v\right\|.
$$
Then, there holds
\begin{enumerate}
    \item [(1)]$\mathscr{T}_{D,\lambda,u,v} \leq C_{u,v} \mathscr{P}_{D,\lambda}^{2v}\lambda^{\min\{u+v,\nu_g\}}$ when $0\leq v \leq 1/2$;
    \item [(2)]$\mathscr{T}_{D,\lambda,u,v} \leq 2^uv(\gamma_u + b +1) \lambda^{u+v-1}\mathscr{W}_D + C_{u,v}\lambda^{\min\{u+v,\nu_g\}}$ when $1/2< v <1$;
    \item [(3)]$\mathscr{T}_{D,\lambda,u,v} \leq 2^{u+1}v(\gamma_u + b +1) (1+\kappa^2)^{v-1}\lambda^{u}\mathscr{W}_D + C_{u,v}\lambda^{\min\{u+v,\nu_g\}}$ when $v\geq 1$,
\end{enumerate}
where $C_{u,v} = 2^{\min\{u+v,\nu_g\}}(\gamma_{\min\{u+v,\nu_g\}} +b + 1) (1 + \kappa^2)^{\max\{u+v-\nu_g,0\}}$ is independent of $\lambda$, $\mathscr{P}_{D,\lambda}$ is defined in \eqref{equ: PDL} and 
\begin{equation}
\label{equ: WD}
\mathscr{W}_{D}:= \|L_K - L_{K,D}\|.
\end{equation}
\end{lemma}

Based on Lemma~\ref{lemma: core of approximation error}, we now provide the decomposition for the approximation error.
\begin{proposition}[Approximation error]
\label{proposition: approximation error decomposition}
Assume $0<\lambda \leq 1$, and $f=dP/dP_0-1=L_{K}^r(u)$ for some $u\in L^2(P_0)$ with $r\geq 1/2$. Then the approximation error $|{\xi}^{\star}_{\lambda}(P,P_0)-\chi^2(P,P_0)|$ in \eqref{equ: error decomposition} is bounded by
 \begin{flalign*}
        &\textit{(1)}\ C_{r,\kappa,b}\left\{(\mathscr{P}_{D,\lambda}^{2r}\lambda^{\min\{r,\nu_g\}}+\mathscr{P}_{D,\lambda}^2  \mathscr{Q}_{D,\lambda} \lambda^{1/2})\|f\|_{L^2(P_0)} +\mathscr{P}_{D,\lambda}^2  \mathscr{Q}_{D,\lambda}\|f\|_{L^2(P_0)}^2 \right\};\qquad  \qquad\\
    &\textit{(2)}\ C_{r,\kappa,b}\left\{ (\mathscr{P}_{D,\lambda} \mathscr{W}_D\lambda^{\min\{r-1,1/2\}} +\mathscr{P}_{D,\lambda}\lambda^{\min\{r,\nu_g\}}+\mathscr{P}_{D,\lambda}^2  \mathscr{Q}_{D,\lambda}\lambda^{1/2})\|f\|_{L^2(P_0)} \right.\\
    & \qquad \qquad \ \left.+ \mathscr{P}_{D,\lambda}^2  \mathscr{Q}_{D,\lambda} \|f\|_{L^2(P_0)}^2 \right\},
 \end{flalign*}
for (1) $1/2\leq r \leq 1$ and (2) $r>1$ respectively. Here, 
\begin{equation}
\label{equ: QDL}
\mathscr{Q}_{D,\lambda}:=\|(L_{K}+ \lambda I)^{-1/2}(L_K - L_{K,D})(L_{K}+ \lambda I)^{-1/2}\|,
\end{equation}
$\mathscr{P}_{D,\lambda}$ and $\mathscr{W}_D$ are defined in \eqref{equ: PDL} and \eqref{equ: WD}, respectively, and the constant $C_{r,\kappa,b}$ is explicitly determined in the proof.
 
\end{proposition}
\begin{proof}
On the one hand, since $L_K f = \mu_P - \mu_0$, we have
\[
\xi^{\star}_{\lambda}(P,P_0) 
= \left\langle g_{\lambda}(L_{K,D})L_K f,\, L_K f \right\rangle_{K}.
\]
On the other hand, since $f   \in \mathcal{H}_K$ and 
$\chi^2(P,P_0) = \|f\|_{L^2(P_0)}^2$, it follows that
\[
\chi^2(P,P_0) = \langle f, L_K f \rangle_{K}.
\]
The approximation error in~\eqref{equ: error decomposition} can be decomposed as 
$$
{\xi}^{\star}_{\lambda}(P,P_0)-\chi^2(P,P_0) = T_1 + T_2,
$$
where
\begin{equation}
\label{equ: two terms}
\begin{aligned}
       T_1 &:=\left \langle g_{\lambda}(L_{K,D})L_{K,D}f-f, L_Kf \right\rangle_K,\\
       T_2 &:= \left \langle g_{\lambda}(L_{K,D})(L_K-L_{K,D})f, L_Kf \right\rangle_K.\\
\end{aligned}
\end{equation}
For the first term, we have
\begin{equation*}
		\begin{aligned}
			|T_1| &= \left|  \left\langle L_K^{1/2}(g_{\lambda}(L_{K,D})L_{K,D}-I) f, L_K^{1/2}f \right\rangle_K \right|\\
			& \leq \left\| (L_K + \lambda I)^{1/2}(L_{K,D} + \lambda I)^{-1/2}  (L_{K,D}+\lambda I)^{1/2}\left( g_{\lambda}(L_{K,D})L_{K,D}-I\right) f\right\|_K \cdot \left\| L_K^{1/2}f \right\|_K\\
			& \leq \left\| (L_{K,D}+\lambda I)^{1/2}\left( g_{\lambda}(L_{K,D})L_{K,D}-I\right) L_K^{r-1/2}L_K^{1/2}u\right\|_K \cdot \mathscr{P}_{D,\lambda} \|f\|_{L^2(P_0)}\\
			& \leq \left\| (L_{K,D}+\lambda I)^{1/2}\left( g_{\lambda}(L_{K,D})L_{K,D}-I\right) (L_K+\lambda I)^{r-1/2}\right\| \cdot \|u\|_{L^2(P_0)}\mathscr{P}_{D,\lambda}  \|f\|_{L^2(P_0)},
		\end{aligned}
	\end{equation*} 
where the first step follows by the self-adjoint of $L_K^{1/2}$ on $\mathcal{H}_K$, the second step follows using the Cauchy-Schwarz's inequality and $\|L_K^{1/2}(L_K + \lambda I)^{-1/2}\| \leq1$, and the last step follows using  $\|L_K^{r-1/2}(L_K + \lambda I)^{-(r-1/2)}\| \leq1$ for $r\geq 1/2$. Taking $u= 1/2$ and $v = r-1/2 \geq 0$ in Lemma~\ref{lemma: core of approximation error}, we have
\begin{enumerate}
    \item [\textit{(1)}]$|T_1| \leq C_{r,\kappa,b}^{\prime} \mathscr{P}_{D,\lambda}^{2r}\lambda^{\min\{r,\nu_g\}}\|f\|_{L^2(P_0)}$ when $1/2\leq r \leq 1$;
    \item [\textit{(2)}]$|T_1| \leq C_{r,\kappa,b}^{\prime}\mathscr{P}_{D,\lambda} (\lambda^{\min\{r-1,1/2\}}\mathscr{W}_D +\lambda^{\min\{r,\nu_g\}})\|f\|_{L^2(P_0)}$ when $r>1$,
\end{enumerate}
where $C_{r,\kappa,b}^{\prime}=\{C_{1/2,r-1/2} + 2^{3/2}(r-1/2)(\gamma_{1/2} + b +1) (1+\kappa^2)^{r}\}\|u\|_{L^2(P_0)}$, and $C_{1/2,r-1/2}$ is explicitly determined in Lemma~\ref{lemma: core of approximation error}. For the second term,  we have
	\begin{equation*}
		\begin{aligned}
			|T_2| &=\left| \left\langle g_{\lambda}^{1/2}(L_{K,D})L_{K}f, g_{\lambda}^{1/2}(L_{K,D}) (L_{K}-L_{K,D})f \right\rangle_{K}\right|\\
			& \leq \left\|g_{\lambda}^{1/2}(L_{K,D})(L_{K,D} + \lambda I)^{1/2}(L_{K,D} + \lambda I)^{-1/2}(L_{K} + \lambda I)^{1/2}(L_{K} + \lambda I)^{-1/2}L_{K}^{1/2}L_{K}^{1/2}f\right\|_{K}\cdot\\
			& \quad \, \left\|g_{\lambda}^{1/2}(L_{K,D})(L_{K,D} + \lambda I)^{1/2}(L_{K,D} + \lambda I)^{-1/2}(L_{K} + \lambda I)^{1/2}\cdot \right.\\
			&\qquad \left.(L_{K} + \lambda I)^{-1/2}(L_{K}-L_{K,D})(L_{K} + \lambda I)^{-1/2}(L_{K} + \lambda I)^{1/2}f\right\|_{K}\\
			& \leq  2b\mathscr{P}_{D,\lambda}^2  \mathscr{Q}_{D,\lambda}    \|(L_{K} + \lambda I)^{1/2}f\|_K\|f\|_{L^2(P_0)}\\
			& \leq2b(1+\|f\|_{K})\mathscr{P}_{D,\lambda}^2  \mathscr{Q}_{D,\lambda} \left\{\sqrt{\lambda}\|f\|_{L^2(P_0)} + \|f\|_{L^2(P_0)}^2\right\} ,
		\end{aligned}
	\end{equation*}
	where first step follows using the Cauchy-Schwarz's inequality, the second step follows by \eqref{equ: Cordes inequality} of Lemma~\ref{lemma: basic operator inequalities} and \eqref{equ: filter property 1} and the bound $\|(L_{K} + \lambda I)^{-1/2}L_{K}^{1/2}\|\leq 1$, and the last step follows using Lemma~\ref{lemma: (LKD+ldI)1/2f}. Combining these two estimates and letting $C_{r,\kappa,b} = C_{r,\kappa,b}^{\prime} +2b(1+\|f\|_{K})$ complete the proof.
\end{proof}
\subsection{Distributional Discrepancy and Operator Similarity}
\label{subsec: distributional discrepancy and operator similarity}
As revealed by the Propositions~\ref{proposition: sample error decomposition}-\ref{proposition: approximation error decomposition}, both the sample error and approximation error can be fundamentally characterized by two key quantities: the distributional discrepancy between $P$ and $P_0$, and the similarity between the empirical operator $L_{K,D}$ and the population operator $L_K$. 

The discrepancy between distributions is captured by the $L^2(P_0)$-norm of the target function (i.e., the $\chi^2$-divergence between $P$ and $P_0$),  together with norms of $\Theta_{\lambda,P,P_0}$ defined in~\eqref{equ: ThetaldQP0}. The following lemma unifies these two quantities via the notion of effective dimension. Since the argument is similar to Lemma A.9 in \citep{hagrass-two-sample-test} and Lemma A.3 in \citep{hagrass-gof-test} and  relies only on population-level calculations, we defer the proof to the Appendix~\ref{section: appendix B}.
\begin{lemma}
\label{lemma: norm of ThetaldQP0}
Let $\Theta_{\lambda,P,P_0}$ be defined in~\eqref{equ: ThetaldQP0}. 
  Then, there holds
	\begin{enumerate}
    	\item[(1)] $\|\Theta_{\lambda,P,P_0}\|_{HS} \leq  \kappa \lambda^{-1/2}\left[\mathcal{N}(\lambda)\right]^{1/4} \|f\|_{L^2(P_0)} + \left[\mathcal{N}(\lambda)\right]^{1/2}$;
    	\item[(2)] $\|\Theta_{\lambda,P,P_0}\| \leq 1+ \kappa \lambda^{-1/2}\left[\mathcal{N}(\lambda)\right]^{1/4}\|f\|_{L^2(P_0)}$.
    	\end{enumerate}
    \end{lemma}

One classical approach to characterize the similarity between $L_K$ and $L_{K,D}$ is to derive tight bounds on the operator difference $\mathscr{W}_{D}$, defined in \eqref{equ: WD}, which has been extensively studied in the literature~\citep{bauer2007regularization,caponnetto2007optimal,lin2017distributed,blanchard2016convergence,guo2017learning}. For instance, for $0<\eta<1$, there holds
\begin{equation}
    \label{equ: ineq for WD}
    \mathscr{W}_{D} \leq \frac{4\kappa^2}{\sqrt{N}} \log\frac{2}{\eta}
\end{equation}
with confidence $1-\eta$. We also consider other operator products, as well as mixtures of operator products and differences, such as $\mathscr{P}_{D,\lambda}$ defined in~\eqref{equ: PDL} and $\mathscr{Q}_{D,\lambda}$ in~\eqref{equ: QDL}, to further describe the similarity. The following lemma is introduced to formalize these results, and the proof can be found in~\citep{lin2020distributed}.
\begin{lemma}
\label{lemma: DKRR communication}
	Assume that $0<\lambda \leq 1$ and $\mathcal{N}(\lambda)\geq 1$. For $0<\eta <1$, there holds
	\begin{equation}
	    \label{equ: ineq for QDL}
	\mathscr{Q}_{D,\lambda} \leq C_1 \mathcal{B}_{D,\lambda} \log\frac{4}{\eta}
	\end{equation}
     with confidence $1-\eta$, where  $C_1 = \max\{(\kappa^2+1)/3,2\sqrt{\kappa^2 + 1}\}$ and   
\begin{equation}
	    \label{equ: BDL}
	    \mathcal{B}_{D,\lambda}:= \frac{1+ \log\mathcal{N}(\lambda)}{N\lambda} + \sqrt{\frac{1+ \log\mathcal{N}(\lambda)}{N\lambda}}.
\end{equation}     
       Additionally, if $\eta \geq 4 \exp\{-1/(2C_1 \mathcal{B}_{D,\lambda})\}$, there holds
    \begin{equation}
        \label{equ: ineq for PDL}
        \mathscr{P}_{D,\lambda} \leq \sqrt{2}
    \end{equation}
	with confidence $1-\eta$.
	\end{lemma}

 \section{Proofs}
 \label{sec: proofs}
 In this section, we prove our main theoretical results.

 The following theorem establishes the total estimation error of $\widehat{\xi}_{\lambda}(P,P_0)$ in probability. Its proof combines the decomposition results for the sample error and the approximation error (Propositions~\ref{proposition: sample error decomposition}--\ref{proposition: approximation error decomposition}) with the results in Subsection~\ref{subsec: distributional discrepancy and operator similarity}. This result plays a central role in deriving the detection boundary, as it provides the key quantity $\mathcal{U}_1$ required in the general framework of Subsection~\ref{subsec:general-framework}.
  
 \begin{theorem}
     \label{theorem: main estimation error} Let $0<\lambda \leq 1$, $n\geq 2$ and $m,N>n$. For $ 4 \exp\{-1/(2C_1 \mathcal{B}_{D,\lambda})\}\leq \eta <1$, there holds
     \begin{equation}
         \label{equ: total error probability bound}
|\widehat{\xi}_{\lambda}(P,P_0)-\chi^2(P,P_0)| \leq C_4\left[F\{n,\lambda, \|f\|_{L^2(P_0)} \}+G\{n,\lambda, \|f\|_{L^2(P_0)} \} \right]\cdot \max\left\{\eta^{-1},\log\frac{4}{\eta} +1\right\}
     \end{equation}
     with confidence $1-5\eta$, where
     \begin{equation}
         \label{equ: F}
         F\{n,\lambda,\|f\|_{L^2(P_0)}\} :=\frac{\{\mathcal{N}(\lambda)\}^{1/2}}{n-1} + \left[\frac{1}{\sqrt{n}} + \frac{\{\mathcal{N}(\lambda)\}^{1/4}}{(n-1)\lambda^{1/2}}\right]\|f\|_{L^2(P_0)}  + \frac{\{\mathcal{N}(\lambda)\}^{1/8}}{\sqrt{n}\lambda^{1/4}}\|f\|_{L^2(P_0)}^{3/2};
     \end{equation}
     $G\{n,\lambda,\|f\|_{L^2(P_0)}\}$ is partitioned as
     \begin{equation}
         \label{equ: G}
         \begin{aligned}
         \begin{cases}
            (\lambda^{\min\{r,\nu_g\}}+\mathcal{B}_{D,\lambda}\lambda^{1/2})\|f\|_{L^2(P_0)} +\mathcal{B}_{D,\lambda}\|f\|_{L^2(P_0)}^2\ for \ 1/2\leq r\leq 1;
           &\\
           & \\(n^{-1/2}\lambda^{\min\{r-1,1/2\}} +\lambda^{\min\{r,\nu_g\}}+\mathcal{B}_{D,\lambda}\lambda^{1/2})\|f\|_{L^2(P_0)} +\mathcal{B}_{D,\lambda}\|f\|_{L^2(P_0)}^2  \ for \ r>1. &
         \end{cases}
     \end{aligned}
     \end{equation}
     Here, $\mathcal{B}_{D,\lambda}$ is defined in \eqref{equ: BDL}, and the constant $C_4$ is explicitly determined in the proof.
 \end{theorem}
 
 \begin{proof}
Taking $P = P_0$ (so that $\|f\|_{L^2(P_0)} = 0$) in Lemma~\ref{lemma: norm of ThetaldQP0}, we obtain
\begin{equation} \label{equ: norms of ThetaldP0P0} \|\Theta_{\lambda,P_0,P_0}\|_{HS} \leq \left\{\mathcal{N}(\lambda)\right\}^{1/2},\quad \|\Theta_{\lambda,P_0,P_0}\| \leq 1. \end{equation}
Substituting the bounds involving $\Theta_{\lambda,P_0,P_0}$ in \eqref{equ: norms of ThetaldP0P0} together with those involving $\Theta_{\lambda,P,P_0}$ from Lemma~\ref{lemma: norm of ThetaldQP0} into \eqref{equ: sample error decomposition} of Proposition~\ref{proposition: sample error decomposition}, and using the fact that $m>n$, we obtain
    \begin{equation}
    \begin{aligned}
            & \quad\,\,\mathbb{E}_D\left\{|\widehat{\xi}_{\lambda}(P,P_0)-{\xi}^{\star}_{\lambda}(P,P_0)|\right\}\\
            &\leq 5b\mathscr{P}_{D,\lambda}^2\left[\frac{\kappa\lambda^{-1/2}\{\mathcal{N}(\lambda)\}^{1/4}\|f\|_{L^2(P_0)}+\{\mathcal{N}(\lambda)\}^{1/2} }{n-1}+\frac{\{\mathcal{N}(\lambda)\}^{1/2}}{m-1}\right]\\
            &\quad \,+ 4b\mathscr{P}_{D,\lambda}^2\left[\frac{\kappa^{1/2}\lambda^{-1/4}\{\mathcal{N}(\lambda)\}^{1/8}\|f\|_{L^2(P_0)}^{1/2}+1 }{\sqrt{n}}+\frac{1}{\sqrt{m}}\right]\|f\|_{L^2(P_0)}\\
            &\leq C_2\mathscr{P}_{D,\lambda}^2 \left(\frac{\{\mathcal{N}(\lambda)\}^{1/2}}{n-1} + \left[\frac{1}{\sqrt{n}} + \frac{\{\mathcal{N}(\lambda)\}^{1/4}}{(n-1)\lambda^{1/2}}\right]\|f\|_{L^2(P_0)}  + \frac{\{\mathcal{N}(\lambda)\}^{1/8}}{\sqrt{n}\lambda^{1/4}}\|f\|_{L^2(P_0)}^{3/2}\right)\\
            & =C_2\mathscr{P}_{D,\lambda}^2 F\left\{n,\lambda,\|f\|_{L^2(P_0)}\right\}, 
    \end{aligned} 
    \end{equation}
    where $C_2 = 5b (\kappa^{1/2} + \kappa +2)$. By the conditional Markov's inequality, there holds
    \begin{equation*}
        |\widehat{\xi}_{\lambda}(P,P_0)-{\xi}^{\star}_{\lambda}(P,P_0)| \leq C_2\eta^{-1}\mathscr{P}_{D,\lambda}^2 F\left\{n,\lambda,\|f\|_{L^2(P_0)}\right\}
    \end{equation*}
    with confidence $1-\eta$. Furthermore, by inserting the estimate for $\mathscr{P}_{D,\lambda}$ in \eqref{equ: ineq for PDL}, for $\eta \geq 4 \exp\{-1/(2C_1 \mathcal{B}_{D,\lambda})\}$, there holds
    \begin{equation}
        \label{equ: sample error probability bound}|\widehat{\xi}_{\lambda}(P,P_0)-{\xi}^{\star}_{\lambda}(P,P_0)| \leq 2C_2\eta^{-1}F\left\{n,\lambda,\|f\|_{L^2(P_0)}\right\}
    \end{equation}
    with confidence $1-2\eta$.
    
    Based on Proposition~\ref{proposition: approximation error decomposition}, by inserting the estimates for $\mathscr{P}_{D,\lambda}$, $\mathscr{Q}_{D,\lambda}$, and $\mathscr{W}_{D}$ from Lemma~\ref{lemma: DKRR communication} and \eqref{equ: ineq for WD}, and using $N>n$, we obtain that for $\eta \geq 4 \exp\{-1/(2C_1 \mathcal{B}_{D,\lambda})\}$, with confidence at least $1-3\eta$, the error term $|{\xi}^{\star}_{\lambda}(P,P_0)-\chi^2(P,P_0)|$ in \eqref{equ: error decomposition} is bounded by
   \begin{equation}
       \label{equ: approximation error probability bound}
       \begin{aligned}
            &(1)\ C_{3}\left\{(\lambda^{\min\{r,\nu_g\}}+\mathcal{B}_{D,\lambda}\lambda^{1/2})\|f\|_{L^2(P_0)} +\mathcal{B}_{D,\lambda}\|f\|_{L^2(P_0)}^2 \right\}(\log\frac{4}{\eta} + 1);\\
           & (2)\ C_{3}\left\{(n^{-1/2}\lambda^{\min\{r-1,1/2\}} +\lambda^{\min\{r,\nu_g\}}+\mathcal{B}_{D,\lambda}\lambda^{1/2})\|f\|_{L^2(P_0)} +\mathcal{B}_{D,\lambda}\|f\|_{L^2(P_0)}^2 \right\}(\log\frac{4}{\eta} + 1),
       \end{aligned}
   \end{equation}
for \textit{(1)} $1/2\leq r \leq 1$ and \textit{(2)} $r>1$, respectively, where $C_3 = (4^r + 2C_1 + 4\sqrt{2}\kappa^2)C_{r,\kappa,b}$. 

Combining the probabilistic bounds in \eqref{equ: sample error probability bound} and \eqref{equ: approximation error probability bound}, and setting $C_4 = 2C_2 + C_3$, we complete the proof.
 \end{proof}

  \subsection{Proof of Theorem~\ref{theorem: testing rule based on inequality}}
Based on Theorem~\ref{theorem: main estimation error}, we are now ready to prove 
Theorem~\ref{theorem: testing rule based on inequality}. 
Before doing so, we require the following lemma in order to bound the critical value, which quantifies the discrepancy between 
$\mathcal{N}(\lambda)$ and $\mathcal{N}_D(\lambda)$; see Proposition A.1 in~\citep{blanchard2019lepskii}.

\begin{lemma}
	\label{lemma: lepskii}
	Let $\eta\in(0,1)$, $L_{\eta} = 2\log(8/\eta)$ and $\mathcal{N}(\lambda) ,\mathcal{N}_D(\lambda) \geq 1$. If $N\lambda \geq {4\kappa^2}$, then there holds
	$$
	\max\left\{\frac{\mathcal{N}(\lambda)}{\mathcal{N}_D(\lambda)},\frac{\mathcal{N}_D(\lambda)}{\mathcal{N}(\lambda)}\right\} \leq \left(1+ \frac{4\kappa L_{\eta}}{\sqrt{N\lambda}}\right)^2.
	$$
    with confidence $1-\eta$.
\end{lemma}
 \begin{proof}[Proof of Theorem~\ref{theorem: testing rule based on inequality}] The proof is divided into two parts: the control of the Type~I error under $H_0$, and the characterization of the detection boundary under $H_1$.
 \vspace{1em}
 
     \underline{\textit{Type I error control.}}  Under the null hypothesis, $\xi^{\star}_{\lambda}(P,P_0) =0$ and $\|f\|_{L^2(P_0)} = 0$. By Assumption~\ref{assumption: effective dimension}, if $N\lambda \geq 16C_1^2\{\log eC_0 +s\log(\lambda^{-1})\}\log^2(4\alpha^{-1})$, the condition $\alpha \geq  4 \exp\{-1/(2C_1 \mathcal{B}_{D,\lambda})\}$ holds. Combining \eqref{equ: sample error probability bound} and \eqref{equ: ineq for PDL}, Markov's inequality implies that
     \begin{equation}
         \label{equ: sample error probability bound under H0} 
         \widehat{\xi}_{\lambda}(P,P_0)\leq 10b\alpha^{-1}\{\mathcal{N}(\lambda)\}^{1/2} \left(\frac{1}{n-1}+\frac{1}{m-1}\right)
     \end{equation}
     holds with probability $1-2\alpha$. Furthermore, by Lemma~\ref{lemma: lepskii}, for $N\lambda \geq 4\kappa^2$, we have
\begin{equation} \label{equ: effective dimension control} \left\{\mathcal{N}(\lambda)\right\}^{1/2}\leq \left\{\mathcal{N}_D(\lambda)\right\}^{1/2}\left(1 +\frac{8\kappa}{\sqrt{N\lambda}}\log\frac{8}{\alpha}\right) \end{equation}
with probability $1-\alpha$. Combining \eqref{equ: sample error probability bound under H0} with \eqref{equ: effective dimension control}, we obtain
     $$
     \widehat{\xi}_{\lambda}(P,P_0)\leq 10b\alpha^{-1} \left(\frac{1}{n-1}+\frac{1}{m-1}\right)\left(1 +\frac{8\kappa}{\sqrt{N\lambda}}\log\frac{8}{\alpha}\right)\left\{\mathcal{N}_D(\lambda)\right\}^{1/2}=\widehat{c}_{\alpha/3,\lambda,n,m,D}
     $$
     with probability $1-3\alpha$. After rescaling $3\alpha$ to $\alpha$, we conclude that the test in \eqref{equ: testing rule based on inequality} is an $\alpha$-level test.
     \vspace{1em}
     
     \underline{\textit{Detection boundary.}} Taking $\lambda = n^{-\frac{2}{4r+s}}$, Assumption~\ref{assumption: effective dimension} implies
 \begin{align}
 \label{equ: bound NL}
 \mathcal{N}(\lambda)\leq C_0n^{\frac{2s}{4r+s}}.
 \end{align}
 Note that $n\geq 3$. For $C_5 = 32C_1^2\left\{\log(eC_0)+\frac{2s}{4r+s}\right\}$, if $N\geq C_5n^{\frac{2}{4r+s}}\log n\log^2(4\delta^{-1})$, we obtain from \eqref{equ: BDL} that
\begin{equation}
    \label{equ: bound BDL}
 \mathcal{B}_{D,\lambda}\leq C_5 N^{-1/2}n^{\frac{1}{4r+s}}\sqrt{\log n},
\end{equation}
 and the condition 
   \begin{equation}
       \label{equ: condition for BDL}
      \delta \geq 4 \exp\{-1/(2C_1 \mathcal{B}_{D,\lambda})\}
   \end{equation}
  is satisfied. 
  
   By $N>n$, substituting \eqref{equ: bound NL}, \eqref{equ: bound BDL} and $\lambda = n^{-\frac{2}{4r+s}}$ into \eqref{equ: F} and \eqref{equ: G}, and taking $\eta = \delta$, we obtain from Theorem~\ref{theorem: main estimation error} that
 \begin{equation}
     \label{equ: bound estimation error}
     \begin{aligned}
         |\widehat{\xi}_{\lambda}(P,P_0)-\chi^2(P,P_0)|\leq \mathcal{U}_1\{n,\delta, \|f\|_{L^2(P_0)}\}
     \end{aligned}
 \end{equation}
 holds with probability $1-5\delta$, where  \begin{equation}
     \label{equ: our U1}
     \begin{aligned}
         &\quad \,\,\mathcal{U}_1\{n,\delta,\|f\|_{L^2(P_0)}\}\\&= C_6\left\{n^{-\frac{4r}{4r+s}} +  \left(n^{-\frac{2r}{4r+s}} + n^{-\frac{4r-1+s/2}{4r+s}}+n^{-1/2}\sqrt{\log n}\right)\|f\|_{L^2(P_0)}+ n^{-\frac{2r-1/2+s/4}{4r+s} }\|f\|_{L^2(P_0)}^{3/2}\right.\\
         &\left.\quad \quad \ \ \ + n^{-\frac{2r-1+s/2}{4r+s}}\sqrt{\log n}\|f\|_{L^2(P_0)}^2\right\}\cdot  \max\left\{\delta^{-1},\log\frac{4}{\delta} +1\right\},
     \end{aligned}
 \end{equation}
and the constant $C_6= 4\sqrt{C_0}C_4C_5$.

Since $m,N>n$, if $N\lambda \geq4\kappa^2$, then by taking $\eta=\delta$ in Lemma~\ref{lemma: lepskii}, we obtain
\begin{equation}
    \label{equ: bound the critical value}
    \widehat{c}_{\alpha,\lambda,n,m ,D} \leq \frac{60b}{(n-1)\alpha}\left(1+ \frac{8\kappa}{\sqrt{n\lambda}}\log\frac{24}{\alpha}\right)\left(1+ \frac{8\kappa}{\sqrt{n\lambda}}\log\frac{8}{\delta}\right)\left\{\mathcal{N}(\lambda)\right\}^{1/2}
\end{equation}
with probability $1-\delta$. Combining \eqref{equ: bound NL}
and $\lambda = n^{-\frac{2}{4r+s}}$, there holds
$$
\widehat{c}_{\alpha,\lambda,n,m,D} \leq  \mathcal{U}_2(n,\alpha,\delta)
$$
with probability $1-\delta$, where
\begin{align}
    \label{equ: Our U2}
    \mathcal{U}_2(n,\alpha,\delta)=120b\sqrt{C_0}\alpha^{-1}\left(1+ 8\kappa\log\frac{24}{\alpha}\right)\left(1+ 8\kappa\log\frac{8}{\delta}\right)n^{-\frac{4r}{4r+s}}.
\end{align}

Substituting \eqref{equ: our U1} and \eqref{equ: Our U2} into \eqref{equ: detection boundary inequality}, taking $\rho(P,P_0)=\|f\|_{L^2(P_0)}^2$, and rescaling $6\delta$ to $\delta$, we obtain from Lemma~\ref{lemma: general framework} that the detection boundary of the test in~\eqref{equ: testing rule based on inequality} over the alternative space in \eqref{ours: alternative space} for $1/2 \leq r \leq \nu_g$ is at most $C^{\ast}(\alpha,\delta)n^{-\frac{4r}{4r+s}}$, where $C^{\ast}(\alpha,\delta)=O( \max\{\delta^{-1}, \log(\delta^{-1})\}+\alpha^{-1}\log(\alpha^{-1})\log(\delta^{-1}) )$ is independent of the sample size. The proof is complete.
 \end{proof}

  \subsection{Proof of Theorem~\ref{theorem: testing rule based on permutation}}
The validity of the test $\phi_{\alpha}^{\xi_{\lambda},{perm}}$ defined in \eqref{equ: testing rule based on perm} in controlling the Type~I error relies on the exchangeability of the proposed statistics, in line with Theorem~10 in \citep{hagrass-gof-test}. Similar to Theorem~\ref{theorem: testing rule based on inequality}, the study of the detection boundary of $\phi_{\alpha}^{\xi_{\lambda},{perm}}$ also builds on Theorem~\ref{theorem: main estimation error} together with the general framework of Lemma~\ref{lemma: general framework}. Since the total estimation error has already been explicitly derived in \eqref{equ: bound estimation error}, it remains to bound the critical value $\hat{q}_{1-\alpha}^{B,\lambda}$ generated by the permutation procedure.

We first introduce the following lemma; it can be found in Lemma~14 in the arXiv version of \citep{hagrass-two-sample-test} and E.4 in~\citep{schrab2023mmd}. For completeness, we place its proof in Appendix~\ref{section: appendix B}.

\begin{lemma}
\label{lemma: DKW for quantile}
Let 
\begin{equation}
    \label{equ: true permuation distribution}
    F_{\lambda}(x) := \frac{1}{(n+m)!}\sum_{\pi \in \Pi_{n+m}} 
\mathbf{1}\!\left\{\widehat{\xi}^{\pi}_{\lambda}(P,P_0) \leq x\right\}
\end{equation}
be the permutation distribution function of $\widehat{\xi}_{\lambda}(P,P_0)$, 
and define its $(1-\alpha)$ quantile as
\begin{equation}
    \label{equ: qld}
    q_{1-\alpha}^{\lambda} := \inf\{q \in \mathbb{R} : F_{\lambda}(q) \geq 1-\alpha\}.
\end{equation}
The empirical $(1-\alpha)$-th quantile $\hat{q}_{1-\alpha}^{B,\lambda}$ is defined in \eqref{empirical quantile}. Then, for any $\alpha>0$ and $\eta>0$, if $B\geq \frac{3}{\alpha^2}
 (\log(2\eta^{-1}) +\alpha(1-\alpha))$, there holds
	 \begin{equation}
	     \label{equ: qbld<qld}
	     \hat{q}_{1-\alpha}^{B,\lambda} \leq q_{1-\alpha/2}^{\lambda}
	 \end{equation}
	 with probability $1-\eta$.
	\end{lemma}

Lemma~\ref{lemma: DKW for quantile} shows that, for sufficiently large $B$, the empirical quantile $\hat{q}_{1-\alpha}^{B,\lambda}$ can be bounded by the quantile of $F_{\lambda}$. Hence, it suffices to bound $q_{1-\alpha}^{\lambda}$. Theorem 6.1 and Lemma H.1 in \citep{kim2022minimax} establish a connection between the quantiles of the permutation 
distribution and the variance of the associated $U$-statistics with finite samples, and $n$ and $m$ of the same order. Continuing with the notation used in this paper, we summarize it in the following lemma.

\begin{lemma}
\label{lemma: kim2022}
For $0<\alpha <e^{-1}$, there exist some constant $\tilde{C}>0$ such that
\begin{equation}
    \label{equ: kim2022}
	    q^{\lambda}_{1-\alpha}\leq \tilde CV\log\frac{1}{\alpha}
	\end{equation}
almost surely, where $V>0$ is defined by
\begin{equation}
\label{equ: Vdagger}
V^2 
:= \frac{1}{nm(n-1)(m-1)}\left\{
   \sum_{i=1 }^n \sum_{\substack{i^{\prime}=1 \\ i^{\prime}\neq i}}^n g^2(x_i,x_{i^{\prime}})
 + \sum_{j=1}^m \sum_{\substack{j^{\prime}=1 \\ j^{\prime}\neq j}}^m g^2(y_j,y_{j^{\prime}})
 + 2\sum_{i=1}^n\sum_{j=1}^m g^2(x_i,y_j)
 \right\}
\end{equation}
with
$$
	g(t_1,t_2) = \left\langle g_{\lambda}^{1/2}(L_{K,D})(K_{t_1} - \mu_P),g_{\lambda}^{1/2}(L_{K,D})(K_{t_2} - \mu_P) \right\rangle_{K},
	$$
and $q_{1-\alpha}^{\lambda}$ is defined in \eqref{equ: qld}.
\end{lemma}

  The next result provides a high-probability bound for $q_{1-\alpha}^{\lambda}$, which plays a role similar to that of Lemma~A.15 in \citep{hagrass-two-sample-test}.

\begin{lemma}
	\label{lemma: bound the quantile} Let $m,N>n\geq 2$. For $0<\alpha < e^{-1}$ and $0<\eta<1$ there holds
	   \begin{equation}
       \label{equ: bound the quantile}
   \begin{aligned}
        q_{1-\alpha}^{\lambda}
        & \leq \frac{C_7}{n\eta}\log\left(\frac{1}{\alpha}\right)\left( \{\mathcal{N}(\lambda)\}^{1/2} +\left[1+ \lambda^{-1/2}\{\mathcal{N}(\lambda)\}^{1/4}\right]\|f\|_{L^2(P_0)} + \right. \\ &\left.\quad \ \, \,\lambda^{-1/4}\{\mathcal{N}(\lambda)\}^{1/8}\|f\|_{L^2(P_0)}^{3/2}+ \|f\|_{L^2(P_0)}^2\right )\cdot \mathscr{P}_{D,\lambda}^2
   \end{aligned}
   \end{equation}
	with confidence $1-\eta$, where the constant $C_7$ is explicitly determined in the proof.
\end{lemma}

\begin{proof} By \eqref{equ: Vdagger} in Lemma~\ref{lemma: kim2022}, it follows that
	$$
	\begin{aligned}
		& \quad\ \quad \,nm(n-1)(m-1)V^2 \\& \ \ =\ \sum_{i= 1}^n \sum_{j \neq i}^n \left\langle g_{\lambda}^{1/2}(L_{K,D})(K_{x_i} - \mu_P),g_{\lambda}^{1/2}(L_{K,D})(K_{x_j} - \mu_P) \right\rangle_{K}^2\\
		& \ \ + \ \sum_{i= 1}^m\sum_{j \neq i}^m \left\langle g_{\lambda}^{1/2}(L_{K,D})(K_{y_i} - \mu_P),g_{\lambda}^{1/2}(L_{K,D})(K_{y_j} - \mu_P) \right\rangle_{K}^2 \\
		&\ \ + 2\sum_{i= 1}^n \sum_{j =1}^m \left\langle g_{\lambda}^{1/2}(L_{K,D})(K_{x_i} - \mu_P), g_{\lambda}^{1/2}(L_{K,D})(K_{y_j} - \mu_P) \right\rangle_{K}^2\\
        & =:V_1+V_2+2V_3.
	\end{aligned}
	$$
	By the basic inequality $|\sum_{k\geq 1}a_k|^2\leq k\sum_{k\geq 1}|a_k|^2$, we have
    $$
    \begin{aligned}
      V_2/4& \leq \sum_{i= 1}^m \sum_{j \neq i}^m \left\langle g_{\lambda}^{1/2}(L_{K,D})(K_{y_i} - \mu_0),g_{\lambda}^{1/2}(L_{K,D})(K_{y_j} - \mu_0) \right\rangle_{K}^2  \\
      &+ 2(m-1)\sum_{i= 1}^m \left\langle g_{\lambda}^{1/2}(L_{K,D})(K_{y_i} - \mu_0),g_{\lambda}^{1/2}(L_{K,D})(\mu_P - \mu_0) \right\rangle_{K}^2  \\
      &+ m(m-1)\left\langle g_{\lambda}^{1/2}(L_{K,D})(\mu_P- \mu_0),g_{\lambda}^{1/2}(L_{K,D})(\mu_P - \mu_0) \right\rangle_{K}^2, \\ 
      V_3/2& \leq \sum_{i= 1}^n \sum_{j =1}^m \left\langle g_{\lambda}^{1/2}(L_{K,D})(K_{x_i} - \mu_P),g_{\lambda}^{1/2}(L_{K,D})(K_{y_j} - \mu_0) \right\rangle_{K}^2 \\
      &+ m\sum_{i= 1}^n \left\langle g_{\lambda}^{1/2}(L_{K,D})(K_{x_i} - \mu_P),g_{\lambda}^{1/2}(L_{K,D})(\mu_P - \mu_0) \right\rangle_{K}^2.
    \end{aligned}
    $$
    Combining with \eqref{equ: condtional variance part 1}, \eqref{equ: condtional variance part 2}, and \eqref{equ: condtional variance part 3} in the proof of Lemma~\ref{lemma: sample error}, together with $m>n$, we have
   $$
   \begin{aligned}
        \mathbb{E}_{D}(V^2) &\leq \frac{32b^2}{(n-1)^2}\left\{ \|\Theta_{\lambda,P,P_0}\|_{HS}^2+\|\Theta_{\lambda,P_0,P_0}\|_{HS}^2\right. \\ &\left.\quad \ \,\left(\|\Theta_{\lambda,P,P_0}\|+\|\Theta_{\lambda,P_0,P_0}\|\right)\|f\|_{L^2(P_0)}^2 + \|f\|_{L^2(P_0)}^4\right \}\cdot \mathscr{P}_{D,\lambda}^4.
   \end{aligned}
   $$
   By Lemma~\ref{lemma: norm of ThetaldQP0} and the Jensen's inequality $\mathbb{E}(|V|) \leq \{\mathbb{E}(V^2)\}^{1/2}$, as well as the basic inequality $|\sum_{k\geq 1}a_k|^{1/2}\leq \sum_{k\geq 1}|a_k|^{1/2}$, we further have
   \begin{equation}
       \label{equ: bound V}
       \begin{aligned}
        \mathbb{E}_{D}(|V|) 
        & \leq \frac{\tilde{C}_7}{n}\left( \{\mathcal{N}(\lambda)\}^{1/2} +\left[1+ \lambda^{-1/2}\{\mathcal{N}(\lambda)\}^{1/4}\right]\|f\|_{L^2(P_0)} + \right. \\ &\left.\quad \ \, \,\lambda^{-1/4}\{\mathcal{N}(\lambda)\}^{1/8}\|f\|_{L^2(P_0)}^{3/2}+ \|f\|_{L^2(P_0)}^2\right )\cdot \mathscr{P}_{D,\lambda}^2,
   \end{aligned}
   \end{equation}
   where $\tilde{C}_7 = 8\sqrt{2}b(2+\kappa + \kappa^{1/2})$. Based on \eqref{equ: bound V}, using the conditional Markov's inequality and combining with \eqref{equ: kim2022}, the proof is complete.
\end{proof}
 \begin{proof}[Proof of Theorem~\ref{theorem: testing rule based on permutation}]
 The proof is divided into two parts: the control of the Type~I error under $H_0$, and the characterization of the detection boundary under $H_1$.
 \vspace{1em}
 
   \underline{\textit{Type I error control.}} Under the null hypothesis $H_0: P=P_0$, for any $\pi \in \Pi_{n+m}$, the test statistic computed from $(x^n; y^m, D)$ is identically distributed to that from $(x_\pi^n; y_\pi^m, D)$. Since 
\[
\hat{q}_{1-\alpha}^{B,\lambda} = \inf\{t \in \mathbb{R} : \hat{F}_{B,\lambda}(t) \geq 1-\alpha\},
\]
we have
\[
1-\alpha \leq \hat{F}_{B,\lambda}(\hat{q}_{1-\alpha}^{B,\lambda})
= \frac{1}{B+1}\sum_{b=0}^B \mathbf{1}\{\widehat{\xi}_\lambda^b(P,P_0) \leq \hat{q}_{1-\alpha}^{B,\lambda}\}.
\]
Taking expectations yields
\[
1-\alpha \leq \frac{1}{B+1}\sum_{b=0}^B \mathbb{P}\{\widehat{\xi}_\lambda^b(P,P_0) \leq \hat{q}_{1-\alpha}^{B,\lambda}\}
= \mathbb{P}\{\widehat{\xi}_\lambda(P,P_0) \leq \hat{q}_{1-\alpha}^{B,\lambda}\},
\]
where the last equality follows since $\{\widehat{\xi}_\lambda^b(P,P_0)\}_{b=0}^B$ are i.i.d. under $H_0$.  
This shows that the test in \eqref{equ: testing rule based on perm} is an $\alpha$-level test.
\vspace{1em}

 \underline{\textit{Detection boundary.}} According to the proof of Theorem~\ref{theorem: testing rule based on inequality}, for $\lambda = n^{-\frac{2}{4r+s}}$,if $N\geq C_5n^{\frac{2}{4r+s}}\log n\log^2(4\delta^{-1})$, the condition \eqref{equ: condition for BDL} holds. Based on Lemma~\ref{lemma: DKRR communication}, inserting the bound \eqref{equ: bound NL} into \eqref{equ: bound the quantile} and taking $\eta = \delta$, there holds
\begin{equation}
\label{equ: bound q1-alpha/2}
    \begin{aligned}
        q_{1-\alpha/2}^{\lambda}&\leq \frac{C_8}{ \delta}\log\frac{2}{\alpha} \left\{n^{-\frac{4r}{4r+s}} + n^{-\frac{1}{2}-\frac{2r-1}{4r+s}}\|f\|_{L^2(P_0)} + n^{-\frac{3}{4}-\frac{r-1/2}{4r+s}}\|f\|_{L^2(P_0)}^{3/2} + n^{-1}\|f\|_{L^2(P_0)}^2\right\}\\
        & =:\mathcal{U}_2^{\prime
}\{n,\alpha,\delta,\|f\|_{L^2(P_0)}\}
    \end{aligned}
\end{equation}
with probability $1-\delta$, where $C_8 = 4\sqrt{C_0}C_7$. Based on Lemma~\ref{lemma: DKW for quantile}, taking $\eta = \delta$, for $B\geq \frac{3}{\alpha^2}
 (\log(2\delta^{-1}) +\alpha(1-\alpha))$, there holds
\begin{equation}
    \label{equ: qBld < qld}
    \hat{q}_{1-\alpha}^{B,\lambda} \leq q_{1-\alpha/2}^{\lambda}
\end{equation}
with probability $1-\delta$. \eqref{equ: bound q1-alpha/2} and \eqref{equ: qBld < qld} imply that
$$
\hat{q}_{1-\alpha}^{B,\lambda} \leq \mathcal{U}_2^{\prime
}\{n,\alpha,\delta,\|f\|_{L^2(P_0)}\}
$$
holds with probability $1-2\delta$. 

Substituting $\mathcal{U}_1\{n,\delta,\|f\|_{L^2(P_0)}\}$ in \eqref{equ: our U1} and $\mathcal{U}^{\prime}_2\{n,\alpha,\delta,\|f\|_{L^2(P_0)}\}$ in \eqref{equ: bound q1-alpha/2} into \eqref{equ: detection boundary inequality}, taking $\rho(P,P_0)=\|f\|_{L^2(P_0)}^2$ and rescaling $7\delta$ to $\delta$, we obtain from Lemma~\ref{lemma: general framework} that the detection boundary of the test in~\eqref{equ: testing rule based on inequality} over the alternative space in \eqref{ours: alternative space} for $1/2 \leq r \leq \nu_g$ is at most $C^{\ast\ast}(\alpha,\delta)n^{-\frac{4r}{4r+s}}$,where $C^{\ast\ast}(\alpha,\delta)=O( \max\{\delta^{-1}, \log(\delta^{-1})\}+\delta^{-1}\log(\alpha^{-1}))$ is independent of the sample size. The proof is complete.
\end{proof}

\appendix
\section{Testing Procedures}
\label{section: appendix A}
 In this appendix, we present the detailed implementation of two testing methods provided in Section~\ref{sec:  spectral regularized GOF test}.

\begin{algorithm}[H]
 	\caption{Spectral regularized GOF test via empirical effective dimension}
 	\label{alg:concentration_test}
 	\begin{algorithmic}[1]
 		\Require Independent samples $x^n= \{x_i\}_{i=1}^n \sim P$, $y^m = \{y_j\}_{j=1}^m \sim P_0$; reference dataset $D=\{z_{\ell}\}_{\ell=1}^N \sim P_0$; regularization parameter $\lambda$.
 		\Ensure Decision to accept or reject $H_0$.
 			\State \textbf{Step 1: Determine the significance level} \newline
 		Choose the significance level $\alpha>0$.
 			\State \textbf{Step 2: Determine the discrepancy measurement} \newline
 		Choose the probability distance ${\xi}_{\lambda}(P, P_0)$ defined in (\ref{corrected MMD}).
 		\State \textbf{Step 3: Compute the test statistic} \newline
 		Compute the  test statistic $\widehat{\xi}_{\lambda}(P, P_0)$ defined in (\ref{two sample statistic}).
 		\State \textbf{Step 4: Compute the critical value} \newline
 		Compute the critical value $\widehat{c}_{\alpha,\lambda,n,m,D}$ defined in \eqref{critical-value-1}.
 		\State \textbf{Step 5: Make the decision} \newline
 		Reject  $H_0$ if and only if $\widehat{\xi}_{\lambda}(P,P_0) \geq \widehat{c}_{\alpha,\lambda,n,m,D}$.
 	\end{algorithmic}
 \end{algorithm} 
 
\begin{algorithm}[H]
  \caption{Spectral regularized GOF test via permutation approach}
  \label{alg:permutation_test}
  \begin{algorithmic}[1]
    \Require Independent samples $x^n= \{x_i\}_{i=1}^n \sim P$, $y^m = \{y_j\}_{j=1}^m \sim P_0$; reference dataset $D=\{z_{\ell}\}_{\ell=1}^N \sim P_0$; number of permutations $B$; regularization parameter $\lambda$.
    \Ensure Decision to accept or reject $H_0$.
    \State \textbf{Step 1: Determine the significance level} \newline
 		Choose the significance level $\alpha>0$.
 			\State \textbf{Step 2: Determine the discrepancy measurement} \newline
 		Choose the probability distance ${\xi}_{\lambda}(P, P_0)$ defined in (\ref{corrected MMD}).
    \State \textbf{Step 3: Compute the test statistic} \newline
      Calculate the observed statistic $\widehat{\xi}_{\lambda}(P,P_0)$ as in (\ref{two sample statistic}).
    \State \textbf{Step 4: Compute the critical value} 
      \begin{enumerate}[]
        \item Pool the samples: $u^{n+m} = x^n \cup y^m$.
        \item For each $b=1,\dots,B$: 
          \begin{itemize}
                \item Randomly permute indices $\{1,\dots,n+m\}$ to obtain $\pi_b$.
                \item Form permuted samples $x_{\pi_b}^n, y_{\pi_b}^m$.
                \item Compute the permuted statistic $\widehat{\xi}_{\lambda}^b(P,P_0)$.
              \end{itemize}
        \item Construct the empirical permutation distribution $\hat{F}_{B,\lambda}$ in \eqref{empirical permutation distribution} and obtain the empirical $(1-\alpha)$-th quantile $\hat{q}^{B,\lambda}_{1-\alpha}$ in \eqref{empirical quantile}.
      \end{enumerate}
    \State \textbf{Step 5: Make the decision} \newline
      Reject $H_0$ if and only if $ \widehat{\xi}_{\lambda}(P,P_0) \;\geq\; \hat{q}^{B,\lambda}_{1-\alpha}.$
  \end{algorithmic}
\end{algorithm}

\section{Auxiliary Proofs and Lemmas}
\label{section: appendix B}
\begin{proof}[Proof of Propostion~\ref{computation}]
	Recall $D=\{z_{\ell}\}_{\ell = 1}^N$. Define the scaled sampling operator $S_D : \mathcal{H}_K \to \mathbb R^N$ as 
	$$
	S_Df := \frac{1}{\sqrt{N}}\left(f(z_1),\cdots,f(z_N)\right)^{\top}, \quad \forall f \in \mathcal{H}_K,
	$$
	and the associated adjoint operator $S_D^{\ast} :\mathbb R^N \to \mathcal{H}_K  $  is defined by
	$$
	S_D^{\ast} \boldsymbol{a} := \frac{1}{\sqrt{N}}\sum_{\ell = 1}^N a_{\ell}K_{z_{\ell}}. \quad \forall \boldsymbol{a} \in\mathbb R^N.
	$$
	Based on this definition, the empirical integral operator $L_{K,D} = S_D^{\ast}S_D$, and the scaled kernel matrix $K_{NN}/N = S_DS_D^{\ast}$. Since $\{(\widehat{\lambda}_i,\widehat{\boldsymbol{\alpha}}_i)\}_{i\in[N]}$ is the normalized eigenpairs of $K_{NN}/N$, for any $i\in [N]$, we have $S_DS_D^{\ast}\widehat{\boldsymbol{\alpha}}_i = \widehat{\lambda}_i\widehat{\boldsymbol{\alpha}}_i$. Acting $S_D^{\ast}$ on both sides of this equation, we have $L_{K,D}(S_D^{\ast}\widehat{\boldsymbol{\alpha}}_i) = \widehat{\lambda}_i(S_D^{\ast}\widehat{\boldsymbol{\alpha}}_i)$. Define $\widehat{\phi}_i = S_D^{\ast}\widehat{\boldsymbol{\alpha}}_i/\sqrt{\widehat{\lambda}_i}$.   Then, it can be verified that $\{(\widehat{\lambda}_i, \widehat{\phi}_i)\}_{i\in[N]}$ are the normalized eigenpairs of $L_{K,D}$. In fact, we have
	$$
	\left\langle \widehat{\phi}_i,\widehat{\phi}_j \right\rangle_{K}=\left(\widehat{\lambda}_i\widehat{\lambda}_j\right)^{-1/2}\left\langle S_D^{\ast}\widehat{\boldsymbol{\alpha}}_i,S_D^{\ast}\widehat{\boldsymbol{\alpha}}_j \right\rangle_{K} = \left(\widehat{\lambda}_i\widehat{\lambda}_j\right)^{-1/2}\left\langle S_DS_D^{\ast}\widehat{\boldsymbol{\alpha}}_i,\widehat{\boldsymbol{\alpha}}_j \right\rangle_{2}=\sqrt{\frac{\widehat{\lambda}_i}{\widehat{\lambda}_j}}\left\langle \widehat{\boldsymbol{\alpha}}_i,\widehat{\boldsymbol{\alpha}}_j \right\rangle_{2}=\delta_{ij},
	$$
	where $\delta_{ij}$ is the Kronecker delta. Hence, we can write
	$$
	g_{\lambda}(L_{K,D}) = \sum_{ \ell = 1}^N g_{\lambda}\left(\widehat{\lambda}_{\ell}\right)\widehat{\phi}_i\otimes \widehat{\phi}_i.
	$$
	
	The two-sample statistic $\widehat{\xi}_{\lambda}(P,P_0)$ can be decomposed as 
	$$
	\begin{aligned}
			\widehat{\xi}_{\lambda}(P,P_0)&\quad = \quad \frac{1}{n(n-1)}\left(\sum_{i=1}^n\sum_{j=1}^n \left\langle g_{\lambda}(L_{K,D})K_{x_i}, K_{x_j}\right \rangle_{K} -  \sum_{i=1}^n \left\langle g_{\lambda}(L_{K,D})K_{x_i}, K_{x_i}\right \rangle_{K}\right)\\
			& \quad \quad\quad +\quad \frac{1}{m(m-1)}\left(\sum_{i=1}^m\sum_{j=1}^m \left\langle g_{\lambda}(L_{K,D})K_{y_i}, K_{y_j}\right \rangle_{K} -  \sum_{i=1}^m \left\langle g_{\lambda}(L_{K,D})K_{y_i}, K_{y_i}\right \rangle_{K}\right)\\
			& \quad \quad \quad-\quad \frac{2}{nm}\sum_{i=1}^n\sum_{j=1}^m \left\langle g_{\lambda}(L_{K,D})K_{x_i}, K_{y_j}\right \rangle_{K}\\
			&\quad =: \quad \frac{1}{n(n-1)}\left(V_1 -V_2\right)+\frac{1}{m(m-1)}\left(V_3-V_4\right)- \frac{2}{nm}V_5.
	\end{aligned}
	$$
	For any $f,h \in \mathcal{H}_K$, 
	$$
    \begin{aligned}
    		\left\langle g_{\lambda}(L_{K,D})f, h\right \rangle_{K} &= \sum_{\ell \geq 1}g_{\lambda}\left(\widehat{\lambda}_{\ell}\right) \left \langle \left(\widehat{\phi}_{\ell}\otimes \widehat{\phi}_{\ell}\right)f,h\right\rangle_K\\
    		& = \sum_{\ell \geq 1}g_{\lambda}\left(\widehat{\lambda}_{\ell}\right) \left \langle  \widehat{\phi}_{\ell},f\right\rangle_K\left \langle  \widehat{\phi}_{\ell},h\right\rangle_K\\
    			& = \sum_{\ell \geq 1}\widehat{\lambda}_{\ell}^{-1}g_{\lambda}\left(\widehat{\lambda}_{\ell}\right) \left \langle  \widehat{\boldsymbol{\alpha}}_{\ell},S_Df\right\rangle_{2} \left \langle  \widehat{\boldsymbol{\alpha}}_{\ell},S_Dh\right\rangle_{2}\\
    			& = (S_Df)^{\top}G_{\lambda, N}(S_Dh),
    \end{aligned}
	$$
	where $\langle\cdot,\cdot \rangle_2$ denotes the standard inner product in $\mathbb R^N$. Hence,
	$$
	V_1 = \frac{1}{N}\sum_{i = 1}^n\sum_{j=1}^n K_{nN}^i G_{\lambda ,N}\left(K_{nN}^j\right)^{\top} =\frac{1}{N} \mathbf{1}_n^{\top} K_{nN}G_{\lambda ,N}K_{nN}^{\top}\mathbf{1}_n,
	$$
	where $K_{nN}^i$ denotes the $i$th row of $K_{nN}$.
	$$
	V_2 = \frac{1}{N}\sum_{i = 1}^n K_{nN}^i G_{\lambda ,N}\left(K_{nN}^i\right)^{\top} = \frac{1}{N}\mathrm{Tr}\left( K_{nN}G_{\lambda ,N}K_{nN}^{\top} \right).
	$$ 
	Similarly, we have
	$$
		V_3 = \frac{1}{N}\mathbf{1}_m^{\top} K_{mN}G_{\lambda ,N}K_{mN}^{\top}\mathbf{1}_m, \quad V_4 = \frac{1}{N} \mathrm{Tr}\left( K_{mN}G_{\lambda ,N}K_{mN}^{\top} \right),
	$$
	and 
	$$
	V_5 = \frac{1}{N}\sum_{i=1}^n\sum_{j=1}^m K_{nN}^i G_{\lambda ,N}\left(K_{mN}^j\right)^{\top} =\frac{1}{N} \mathbf{1}_n^{\top} K_{nN}G_{\lambda ,N}K_{mN}^{\top}\mathbf{1}_m,
	$$
	which completes the proof.
\end{proof}

\begin{proof}[Proof of Lemma~\ref{lemma: core of approximation error}]
We first show that for $s\geq 0$, 
\begin{equation}
    \label{equ: Guo2017}
     \left\| \left( g_{\lambda}(L_{K,D})L_{K,D}-I\right)(L_{K,D} + \lambda I)^{s} \right\|\leq2^{\min\{s,\nu_g\}}(\gamma_{\min\{s,\nu_g\}} +b + 1) (1 + \kappa^2)^{\max\{s-\nu_g,0\}}\lambda^{\min\{s,\nu_g\}}.
\end{equation}
Since $(a+b)^t\leq 2^t(a^t + b^t)$ for any $a,b,t >0$, we have
$$
\begin{aligned}
    &\quad \,\left\| \left( g_{\lambda}(L_{K,D})L_{K,D}-I\right)(L_{K,D} + \lambda I)^{\min\{s,\nu_g\}} \right\|\\
    & \leq \left\| \left( g_{\lambda}(L_{K,D})L_{K,D}-I\right)(L_{K,D})^{\min\{s,\nu_g\}} \right\|2^{\min\{s,\nu_g\}} +\left\| \left( g_{\lambda}(L_{K,D})L_{K,D}-I\right) \right\| \lambda^{\min\{s,\nu_g\}}2^{\min\{s,\nu_g\}}\\
    & \leq 2^{\min\{s,\nu_g\}}(\gamma_{\min\{s,\nu_g\}} +b +1)\lambda^{\min\{s,\nu_g\}},
\end{aligned}
$$
 where the last step follows using \eqref{equ: filter property 2}.  Hence, we have
 $$
 \begin{aligned}
      &\quad \,\left\| \left( g_{\lambda}(L_{K,D})L_{K,D}-I\right)(L_{K,D} + \lambda I)^{s} \right\|\\
      & \leq\left\| \left( g_{\lambda}(L_{K,D})L_{K,D}-I\right)(L_{K,D} + \lambda I)^{\min\{s,\nu_g\}}\right\|\cdot \left\|(L_{K,D} + \lambda I)^{\max\{s-\nu_g,0\}}\right\| \\
      & \leq 2^{\min\{s,\nu_g\}}(\gamma_{\min\{s,\nu_g\}} +b + 1) \lambda^{\min\{s,\nu_g\}}(1 + \kappa^2)^{\max\{s-\nu_g,0\}},
 \end{aligned}
 $$
where the last step follows due to $\|L_{K,D} + \lambda I\| \leq \lambda + \kappa^2 \leq 1 +\kappa^2$. The inequality in \eqref{equ: Guo2017} holds. Now we bound $\mathscr{T}_{D,\lambda,u,v}$ for $0\leq v \leq 1/2$, $1/2<v<1$ and $v\geq 1$, respectively.

For $0<v\leq 1/2$, we have
	$$
	\begin{aligned}
		&\,\quad \left\| (L_{K,D}+\lambda I)^{u}\left( g_{\lambda}(L_{K,D})L_{K,D}-I\right) (L_K + \lambda I)^{v}\right\|\\
		& \leq \left\| (L_{K,D}+\lambda I)^{u}\left( g_{\lambda}(L_{K,D})L_{K,D}-I\right)(L_{K,D} + \lambda I)^{v} \right\| \cdot \left\|(L_{K,D} + \lambda I)^{-v}(L_{K} + \lambda I)^{v} \right\|\\
		& \leq  \left\| \left( g_{\lambda}(L_{K,D})L_{K,D}-I\right)(L_{K,D} + \lambda I)^{u+v} \right\|\cdot  \left\|(L_{K,D} + \lambda I)^{-v}(L_{K} + \lambda I)^{v} \right\|\\
		&\leq 2^{\min\{u+v,\nu_g\}}(\gamma_{\min\{u+v,\nu_g\}} +b + 1) (1 + \kappa^2)^{\max\{u+v-\nu_g,0\}}\lambda^{\min\{u+v,\nu_g\}}\cdot \mathscr{P}_{D,\lambda}^{2v},
	\end{aligned}
	$$
	where the last step follows using \eqref{equ: Guo2017} and \eqref{equ: Cordes inequality} of Lemma~\ref{lemma: basic operator inequalities}. For $v>1/2$, by adding and subtracting the operator $(L_{K,D} + \lambda I)^{v}$, we have
	\begin{equation}
	    \label{equ: v>1/2}
	    \begin{aligned}
		&\quad \, \left\| (L_{K,D}+\lambda I)^{u}\left( g_{\lambda}(L_{K,D})L_{K,D}-I\right) (L_K+\lambda I)^{v}\right\|\\
		& \leq  \left\| (L_{K,D}+\lambda I)^{u}\left( g_{\lambda}(L_{K,D})L_{K,D}-I\right) \left[(L_K + \lambda I)^{v} - (L_{K,D} + \lambda I)^{v}\right]\right\|+ \\
		&\quad  \,\left\| (L_{K,D}+\lambda I)^{u}\left( g_{\lambda}(L_{K,D})L_{K,D}-I\right)  (L_{K,D} + \lambda I)^{v}\right\|\\
		& \leq \left\| \left( g_{\lambda}(L_{K,D})L_{K,D}-I\right) (L_{K,D}+\lambda I)^{u} \right\| \left \| (L_K + \lambda I)^{v} -(L_{K,D} + \lambda I)^{v} \right\| +\\
		& \quad \,\left\| \left( g_{\lambda}(L_{K,D})L_{K,D}-I\right)  (L_{K,D}+\lambda I)^{u+v}\right\|\\
		& \leq 2^{u}(\gamma_{u} +b + 1) \lambda^{u}\left\| (L_K + \lambda I)^{v} -(L_{K,D} + \lambda I)^{v} \right\| +\\
		&\quad \ 2^{\min\{u+v,\nu_g\}}(\gamma_{\min\{u+v,\nu_g\}} +b + 1) \lambda^{\min\{u+v,\nu_g\}}(1 + \kappa^2)^{\max\{u+v-\nu_g,0\}},
	\end{aligned}
	\end{equation}
	where the last step following by \eqref{equ: Guo2017}. For any $\varepsilon>0$, let $$
	A_{\varepsilon} = \frac{L_{K} + \lambda I}{(1+\varepsilon)(\lambda + \kappa^2)}\quad \mathrm{and}\quad  B_{\varepsilon} = \frac{L_{K,D} + \lambda I}{(1+\varepsilon)(\lambda + \kappa^2)}.
	$$
	 Then the spectrum of $A_{\varepsilon}$ and $B_{\varepsilon}$ is contained in $ (\lambda/\{(1+2\varepsilon)(\lambda + \kappa^2)\},1)$. If $1/2 <v< 1$, we have
	$$
	\|A_{\varepsilon}^v -B_{\varepsilon}^v\|\leq v \left\{\frac{\lambda}{(1+2\varepsilon)(\lambda + \kappa^2)}\right\}^{v -1}\|A_{\varepsilon}-B_{\varepsilon}\|
	$$
	by \eqref{equ: Dicker2017 Lemma 8} of Lemma~\ref{lemma: basic operator inequalities}. After a brief simplification and letting $\varepsilon \to 0^{+}$, we obtain 
	\begin{equation}
	    \label{equ: v:(1/2,1)}
	    \left\| (L_K + \lambda I)^{v} -(L_{K,D} + \lambda I)^{v} \right\| \leq v\lambda^{v-1}\mathscr{W}_D.
	\end{equation}
	If $v \geq 1$, we have 
	$$
	\|A_{\varepsilon}^v-B_{\varepsilon}^v\|\leq 2v\|A_{\varepsilon} - B_{\varepsilon}\|
	$$
		by \eqref{equ: Dicker2017 Lemma 7} of Lemma~\ref{lemma: basic operator inequalities}, and a similar argument gives
		\begin{equation}
		\label{equ: v:(1, infty)}
		    \left\| (L_K + \lambda I)^{v} -(L_{K,D} + \lambda I)^{v} \right\| \leq 2v(1 +\kappa^2)^{v-1}\mathscr{W}_D.
		\end{equation} 
		Plugging these two estimates, \eqref{equ: v:(1, infty)} and \eqref{equ: v:(1/2,1)}, into \eqref{equ: v>1/2}, we complete the proof.
\end{proof}

\begin{proof}[Proof of Lemma~\ref{lemma: norm of ThetaldQP0}]
By definition of the covariance operator, we have
$$
\Sigma_{P}:=\int_{\mathcal{X}}(K_x - \mu_P)\otimes (K_x - \mu_P)dP(x)= \int_{\mathcal{X}}K_x\otimes K_xdP(x) - \mu_P\otimes \mu_P.
$$
Since $f(x) = dP/dP_0(x)-1$, we further have
	$$ 
	\begin{aligned}
	    \Sigma_P &= \int_{\mathcal{X}} K_x \otimes K_x dP_0(x) + \int_{\mathcal{X}} K_x \otimes K_x \left(\frac{dP}{dP_0}(x)-1\right)dP_0(x)-\mu_P\otimes\mu_P\\
	    & = L_K + \int_{\mathcal{X}} K_x \otimes K_x f(x)dP_0(x) -\mu_P \otimes \mu_P.
	\end{aligned}
	$$
	By the positiveness of $\Sigma_P$ and $\mu_P \otimes \mu_P$ on $\mathcal{H}_K$, there holds
	$$
	0 \preceq \Sigma_P\preceq L_K + \int_{\mathcal{X}} K_x \otimes K_x f(x)dP_0(x)=:L_K +S_f,
	$$
	and by definition of $\Theta_{\lambda,P,P_0}$ in \eqref{equ: ThetaldQP0}, we have
	\begin{equation}
		\label{equ: order of operators}
	0\preceq \Theta_{\lambda,P,P_0}\preceq   (L_K + \lambda I)^{-1/2}L_K(L_K + \lambda I)^{-1/2}+(L_K + \lambda I)^{-1/2}S_f(L_K + \lambda I)^{-1/2}.
	\end{equation}
We first bound the Hilbert-Schmidt norm. On the one hand, we have
\begin{equation}
\label{equ: HS norm involving Sf}
    \begin{aligned}
   &\quad\,\, \left\|(L_K + \lambda I)^{-1/2}S_f(L_K + \lambda I)^{-1/2} \right\|_{HS}^2= \mathrm{Tr}\left[(L_K + \lambda I)^{-1}S_f (L_K + \lambda I)^{-1}S_f \right]\\& =\iint_{\mathcal{X\times X}} \mathrm{Tr} \left[\left(L_{K} +\lambda I\right)^{-1}  K_x \otimes K_x\left(L_{K} +\lambda I\right)^{-1}  K_y \otimes K_y  \right]f(x)f(y)dP_0(y)dP_0(x)\\
	& \leq \left\{\iint_{\mathcal{X\times X}} \mathrm{Tr}^2 \left[\left(L_{K} +\lambda I\right)^{-1}  K_x \otimes K_x\left(L_{K} +\lambda I\right)^{-1}  K_y \otimes K_y  \right]dP_0(y)dP_0(x)\right\}^{1/2}\\
 & 	\quad\,\cdot \left[\iint_{\mathcal{X\times X}} f^2(x)f^2(y)dP_0(x)dP_0(y) \right]^{1/2}\\
 & = \left\{\iint \mathrm{Tr}^2 \left[\left(L_{K} +\lambda I\right)^{-1}  K_x \otimes K_x\left(L_{K} +\lambda I\right)^{-1}  K_y \otimes K_y  \right]dP_0(y)dP_0(x)\right\}^{1/2}\|f\|_{L^2(P_0)}^2,
\end{aligned}
\end{equation}
where the second step follows using the Cauchy-Schwarz's inequality. Note that 
$$
\begin{aligned}
	& \quad \,\, \mathrm{Tr} \left[\left(L_{K} +\lambda I\right)^{-1}  K_x \otimes K_x\left(L_{K} +\lambda I\right)^{-1}  K_y \otimes K_y  \right]\\
	& =\left\langle K_y \otimes K_y ,(L_{K} +\lambda I)^{-1}  K_x \otimes K_x(L_{K} +\lambda I)^{-1}\right\rangle_{HS}  \\
	&= \left\langle K_y,\left(L_{K} +\lambda I\right)^{-1}K_x \right\rangle_{K}^2 \\
	& \leq \kappa^4 \lambda^{-2},
\end{aligned}
$$
where the second step follows using Lemma~\ref{lemma: Balazs2005} and the last step follows since $\sup_{x}K(x,x) \leq \kappa^2$, we have
 $$
\begin{aligned}
	&\quad \,\,\iint \mathrm{Tr}^2 \left[\left(L_{K} +\lambda I\right)^{-1}  K_x \otimes K_x\left(L_{K} +\lambda I\right)^{-1}  K_y \otimes K_y  \right]dP_0(y)dP_0(x) \\
	& \leq \kappa^4\lambda^{-2} \iint \mathrm{Tr} \left[\left(L_{K} +\lambda I\right)^{-1}  K_x \otimes K_x\left(L_{K} +\lambda I\right)^{-1}  K_y \otimes K_y  \right]dP_0(y)dP_0(x)\\
	& = \kappa^4\lambda^{-2} \mathrm{Tr}\left[L_K^2(L_K + \lambda I)^{-2}\right]\\
	&\leq \kappa^4 \lambda^{-2} \mathcal{N}(\lambda),
\end{aligned}
$$
where the last inequality follows from Lemma~\ref{lemma: trace inequality} and $\|(L_K + \lambda I)^{-1}L_K\| \leq 1$. Plugging this estimate into \eqref{equ: HS norm involving Sf} implies
\begin{equation}
    \label{equ: HS norm part 1}
    \left\|(L_K + \lambda I)^{-1/2}S_f(L_K + \lambda I)^{-1/2} \right\|_{HS} \leq \kappa\lambda^{-1/2}\left[\mathcal{N}(\lambda)\right]^{1/4}\|f\|_{L^2(P_0)}.
\end{equation}
On the other hand, we have
$$
\begin{aligned}
    \left\|(L_K + \lambda I)^{-1/2}L_K(L_K + \lambda I)^{-1/2} \right\|_{HS}^2= \mathrm{Tr}\left[L_K^2(L_K + \lambda I)^{-2}\right]\leq \mathcal{N}(\lambda),
\end{aligned}
$$
where the second step follows from Lemma~\ref{lemma: trace inequality} and $\|L_K(L_K + \lambda I)^{-1}\|\leq 1$. It implies
\begin{equation}
    \label{equ: HS norm part 2}
    \left\|(L_K + \lambda I)^{-1/2}L_K(L_K + \lambda I)^{-1/2} \right\|_{HS}\leq\left[ \mathcal{N}(\lambda)\right]^{1/2}.
\end{equation}
Hence, from \eqref{equ: HS norm part 1}, \eqref{equ: HS norm part 2} and \eqref{equ: order of operators}, we conclude that 
$$
\|\Theta_{\lambda,P,P_0}\|_{HS} \leq  \kappa \lambda^{-1/2}\left[\mathcal{N}(\lambda)\right]^{1/4} \|f\|_{L^2(P_0)} + \left[\mathcal{N}(\lambda)\right]^{1/2}.
$$
We next deal with the operator norm. From \eqref{equ: order of operators}, we  have 
 $$
\begin{aligned}
	\|\Theta_{\lambda,P,P_0}\|& \leq \|(L_K + \lambda I)^{-1/2}S_f(L_K + \lambda I)^{-1/2}\| + \|(L_K + \lambda I)^{-1/2}L_K(L_K + \lambda I)^{-1/2}\|\\
	& \leq 1 + \|(L_K + \lambda I)^{-1/2}S_f(L_K + \lambda I)^{-1/2}\| \\
	& \leq 1 + \|(L_K + \lambda I)^{-1/2}S_f(L_K + \lambda I)^{-1/2}\|_{HS}.
\end{aligned}
$$
Combining with \eqref{equ: HS norm part 1}, we obtain
$$
\|\Theta_{\lambda,P,P_0}\|\leq 1 + \kappa\lambda^{-1/2}\left[\mathcal{N}(\lambda)\right]^{1/4}\|f\|_{L^2(P_0)}.
$$
The proof is complete.
\end{proof}

\begin{proof}[Proof of Lemma~\ref{lemma: DKW for quantile}]
Recall $\hat{F}_{B,\lambda}$ in \eqref{empirical permutation distribution} and $F_{\lambda}$ in \eqref{equ: true permuation distribution}, as well as $\hat{q}_{1-\alpha}^{B,\lambda}$ in~\eqref{empirical quantile} and $q_{1-\alpha}^{\lambda}$ in~\eqref{equ: qld}. It follows that
\begin{equation}
    \label{equ: qbld11}
    \begin{aligned}
\hat{q}_{1-\alpha}^{B,\lambda}&:=\inf \left\{ t: \hat{F}_{B,\lambda}(t) \geq 1-\alpha\right\}\\
&=\inf \left\{ t: \frac{1}{B+1}\sum_{b = 0}^B \mathbf{1}\left\{\widehat{\xi}_{\lambda}^{b}(P,P_0) \leq t\right\} \geq 1-\alpha\right\}\\
&\leq \inf \left\{ t: \frac{1}{B}\sum_{b = 1}^B \mathbf{1}\left\{\widehat{\xi}_{\lambda}^{b}(P,P_0) \leq t\right\} \geq \frac{B+1}{B}(1-\alpha)\right\}\\
&= \inf \left\{ t: \hat{F}_{B-1,\lambda}(t) \geq \frac{B+1}{B}(1-\alpha)\right\}.
\end{aligned}
\end{equation}
Define the event
$$
\mathcal{A} := \left\{\sup_{t\in \mathbb{R}}:\left |\hat{F}_{B-1,\lambda}(t) - F_{\lambda} (t)\right| \leq\sqrt{\frac{1}{2B} \log \left(\frac{2}{\eta}\right)} \right\}.
$$
The Dvoretzky-Kiefer-Wolfowitz inequality~\citep{DKW1,DKW2} guarantees that $\mathbb{P}(\mathcal{A})\geq 1-\eta$.  Under the event $\mathcal{A}$, it follows from~\eqref{equ: qbld11} that
$$
\hat{q}_{1-\alpha}^{B,\lambda} \leq  \inf \left\{ t: {F}_{\lambda}(t) \geq \sqrt{\frac{1}{2B} \log \left(\frac{2}{\eta}\right)}+ \frac{B+1}{B}(1-\alpha)\right\}=:q^{\lambda}_{1-\alpha_1}.
$$
Thus, to ensure $\hat{q}_{1-\alpha}^{B,\lambda} \leq q_{1-\alpha/2}^{\lambda}$, a sufficient condition is 
$$
1-\alpha_1=\sqrt{\frac{1}{2B} \log \left(\frac{2}{\eta}\right)}+ \frac{B+1}{B}(1-\alpha) \leq 1-\alpha/2,
$$
which can be guaranteed by choosing $B \geq \frac{3}{\alpha^2} \left(\log(2\eta^{-1}) + \alpha (1-\alpha)\right)$. The proof is complete.
\end{proof}

\begin{lemma}
\label{lemma: (LKD+ldI)1/2f}
Let $L_K$ be the integral operator on $\mathcal{H}_K$ and $f\in \mathcal{H}_K$, then $\|(L_{K} + \lambda I)^{1/2}f\|_K\leq \|f\|_{L^2(P_0)} + \sqrt{\lambda}\|f\|_{K}$.
\end{lemma}
\begin{proof}
Since $L_K = \sum_{i\geq 1}\lambda_i\varphi_i\otimes_{L^2(P_0)}\varphi_i$ on $L^2(P_0)$, $\{\sqrt{\lambda_i}\varphi_i\}_{i\geq 1}$ forms an orthonormal basis of $\mathcal{H}_K$, and $L_K = \sum_{i\geq 1}\lambda_i (\sqrt{\lambda_i}\varphi_i)\otimes_K(\sqrt{\lambda_i}\varphi_i)$ on $\mathcal{H}_K$. Hence, there holds
$$
	\begin{aligned}
		\left\|(L_{K} + \lambda I)^{1/2}f\right\|_K^2 &= \left\|\sum_{i \geq 1} (\lambda_i + \lambda)^{1/2}\left\langle f,\sqrt{\lambda_i}\varphi_i \right\rangle_{K} \left(\sqrt{\lambda_i}\varphi_i\right)\right\|_K^2\\
        &= \sum_{i\geq1}\lambda_i\left\langle f,\sqrt{\lambda_i}\varphi_i\right \rangle_{K}^2 + \lambda \sum_{i\geq1}\left\langle f,\sqrt{\lambda_i}\varphi_i \right\rangle_{K}^2 \\
		& = \left\|L_K^{1/2}f\right\|_{K}^2  + \lambda \|f\|_{K}^2 \\
        &= \|f\|_{L^2(P_0)}^2  + \lambda \|f\|_{K}^2.
	\end{aligned}
	$$
	Combining with the basic inequality $\sqrt{a+ b} \leq \sqrt{a} + \sqrt{b}$ for $a,b \geq 0$ completes the proof.
\end{proof}

\section{Auxiliary Results on Linear Operators}
\label{section: appendix C}
  In this appendix, we recall basic definitions and properties for several important classes of linear operators on Hilbert spaces, and provide some frequently used lemmas throughout the proof.
  
  Let $(\mathcal{H},\langle \cdot, \cdot \rangle_{\mathcal{H}})$ be a real Hilbert space. 
A linear operator $T: \mathcal{H} \to \mathcal{H}$ is called \emph{bounded} if its   \emph{operator norm} $
\|T\| := \sup_{\|\psi\|_{\mathcal{H}} \leq 1}{\|T\psi\|_{\mathcal{H}}}$ is finite, where the norm $\|\cdot\|_{\mathcal{H}}$ is induced by the inner product $\langle \cdot,\cdot\rangle_{\mathcal{H}}$ on $\mathcal{H}$. The space of all bounded linear operators on $\mathcal{H}$ is denoted by $\mathcal{B}(\mathcal{H})$. For any $T \in \mathcal{B}(\mathcal{H})$, there exists a unique operator $T^* \in \mathcal{B}(\mathcal{H})$, called \emph{adjoint} of $T$, satisfying $\langle \psi, T\phi \rangle_{\mathcal{H}} = \langle T^*\psi, \phi \rangle_{\mathcal{H}}$ for any $\psi, \phi \in \mathcal{H}$. If $T \in \mathcal{B}(\mathcal{H})$, then $\|T\| =\|T^*\|$ and $\|T^*T\| = \|T\|^2$. An operator $T\in \mathcal{B}(\mathcal{H})$ is called \emph{self-adjoint} if $T =T^{\ast}$. A self-adjoint operator $T \in \mathcal{B}(\mathcal{H})$ is called \emph{positive}, denoted as $T \succeq 0$, if $\langle \psi, T\psi \rangle_{\mathcal{H}} \geq 0$ for any $\psi \in \mathcal{H}$, and \emph{strictly positive}, denoted as $T \succ 0$, if $\langle \psi, T\psi \rangle_{\mathcal{H}} > 0$ for any nonzero $\psi \in \mathcal{H}$. If $T$ is compact and positive on $\mathcal{H}$, the spectral theorem ensures that there exists a normalized eigenpairs of $T$, denoted as $\{(\lambda_n, e_n)\}_{n\geq 1}$, with eigenvalues $\lambda_1 \geq \lambda_2 \geq \cdots \geq 0$ and eigenfunctions $\{e_n\}_{n\geq 1}$ forming an orthonormal basis of $\mathcal{H}$. There also holds $\|T\| = \lambda_1$. An operator $T \in \mathcal{B}(\mathcal{H})$ is called \emph{trace-class}, if for some orthonormal basis $\{e_n\}_{n\ge 1}$ of $\mathcal{H}$, $
\|T\|_1 := \operatorname{Tr}(|T|) = \sum_{n=1}^{\infty} \langle e_n, |T| e_n \rangle_{\mathcal{H}} < \infty$, where $|T| := (T^* T)^{1/2}$ is the operator modulus. The space of all trace-class operators is denoted by $\mathcal{B}_1(\mathcal{H})$. The \emph{trace} of $T \in \mathcal{B}_1(\mathcal{H})$ is $
\operatorname{Tr}(T) := \sum_{n=1}^{\infty} \langle e_n, T e_n \rangle_{\mathcal{H}}$. An operator $T\in \mathcal{B}(\mathcal{H})$ is called \emph{Hilbert-Schmidt} if $
\|T\|_{HS}^2 := \sum_{n=1}^{\infty} \|T e_n\|_{\mathcal{H}}^2 < \infty$. The space of all Hilbert-Schmidt operators is denoted by $\mathcal{B}_{HS}(\mathcal{H})$.  The space $\mathcal{B}_{HS}(\mathcal{H})$ becomes a Hilbert space when equipped with the inner product $
\langle S, T \rangle_{HS} := \text{Tr}(T^*S) = \sum_{n=1}^\infty \langle Se_n, Te_n \rangle_{\mathcal{H}}$. The induced norm $\|\cdot\|_{HS}$ is called the \emph{Hilbert-Schmidt norm}. The spaces of bounded, Hilbert--Schmidt, and trace-class operators admit the following inclusion relationship: $
\mathcal{B}_1(\mathcal{H}) \subset \mathcal{B}_{HS}(\mathcal{H}) \subset \mathcal{B}(\mathcal{H})$
with the norm relations $
\|T\| \le \|T\|_{HS} \le \|T\|_1 $ and $
\|S T\|_{HS} \le \|S\|_{HS} \|T\|$.

We next collect some auxiliary lemmas that will be used in this paper. The first lemma characterizes the Hilbert-Schmidt structure of rank-one operators generated by outer products, which can be found in Lemma A.4.39 of~\citet{balazs2005regular}.
The second lemma is a trace inequality for powers of product of positive operators~\citep{shebrawi2013trace}.
The thrid lemma collects several inequalities for powers of positive operators under the spectral calculus.  The proofs of \eqref{equ: Cordes inequality}, \eqref{equ: Dicker2017 Lemma 7} and \eqref{equ: Dicker2017 Lemma 8} can be found in Theorem IX.2.1 of \citet{bhatia2013matrix}, Lemmas 7-8 of \citet{Dicker2017}, respectively.

 \begin{lemma}
 \label{lemma: Balazs2005}
 Let $(\mathcal{H},\langle \cdot,\cdot\rangle_{\mathcal{H}})$ be a real Hilbert space. For any $f, g, u,v \in \mathcal{H}$, there holds $ f\otimes_{\mathcal{H}} g, u \otimes_{\mathcal{H}} v \in \mathcal{B}_{HS}(\mathcal{H})$ and   $\langle f\otimes_{\mathcal{H}} g, u \otimes_{\mathcal{H}} v\rangle_{HS} = \langle f,u \rangle_{\mathcal{H}}\langle g,v \rangle_{\mathcal{H}}$, where $f\otimes_{\mathcal{H}} g$ denotes the outer product of $f$ and $g$ on $\mathcal{H}$, defined as $(f\otimes_{\mathcal{H}} g)h:=f \langle g,h\rangle_{\mathcal{H}}$ for $h \in \mathcal{H}$.
 \end{lemma}
\begin{lemma}
\label{lemma: trace inequality}
Let $A,B$ be positive matrices. There holds
$$
\mathrm{Tr}(AB)^k \leq \min\left\{ \|A\|^k\ \mathrm{Tr}(B)^k, \|B\|^k\ \mathrm{Tr}(A)^k \right\}
$$
for any positive integer $k$.
\end{lemma}

 \begin{lemma}
 \label{lemma: basic operator inequalities}
 Let $A,B$ be positive matrices. There holds
 \begin{equation}
     \label{equ: Cordes inequality}
     \|A^uB^u\| \leq \|AB\|^u,\quad 0\leq u\leq 1,
 \end{equation}
 where $\mu(A) $ for $\mu: [0,\infty) \to [0,\infty)$ is define by spectral calculus. If $\|A\|, \|B\| \leq 1$, then
 \begin{equation}
     \label{equ: Dicker2017 Lemma 7}
     \|A^u-B^u\| \leq 2u \|A-B\|,\quad u > 1.
 \end{equation}
  For $v_1 \in (0,1)$, if the spectrum of $A,B$ is contained in $(v_2,1)$ for some $v_2 \in (0,1)$, then
  \begin{equation}
     \label{equ: Dicker2017 Lemma 8}
     \|A^{v_1}-B^{v_1}\| \leq v_1v_2^{v_1 -1} \|A-B\|.
 \end{equation}
 \end{lemma}
\bibliography{sample}

@book{lehmann2008testing,
  title={Testing statistical hypotheses},
  author={Lehmann, Erich Leo and Romano, Joseph P and Casella, George},
  year={2008},
  publisher={Springer}
}

@inproceedings{smola2007hilbert,
  title={A {Hilbert} space embedding for distributions},
  author={Smola, Alex and Gretton, Arthur and Song, Le and Sch{\"o}lkopf, Bernhard},
  booktitle={International Conference on Algorithmic Learning Theory},
  pages={13--31},
  year={2007}
}

@article{gretton2012kernel,
  title={A kernel two-sample test},
  author={Gretton, Arthur and Borgwardt, Karsten M and Rasch, Malte J and Sch{\"o}lkopf, Bernhard and Smola, Alexander},
  journal={Journal of Machine Learning Research},
  volume={13},
  number={1},
  pages={723--773},
  year={2012},
  publisher={JMLR. org}
}

@inproceedings{gretton2007kernel,
  title={A kernel method for the two-sample-problem},
  author={Gretton, Arthur and Borgwardt, Karsten and Rasch, Malte and Sch{\"o}lkopf, Bernhard and Smola, Alex},
  booktitle={Advances in Neural Information Processing Systems},
  pages = {513--520},
  year={2007}
}

@article{hagrass-two-sample-test,
  title={Spectral regularized kernel two-sample tests},
  author={Hagrass, Omar and Sriperumbudur, Bharath and Li, Bing},
  journal={The Annals of Statistics},
  volume={52},
  number={3},
  pages={1076--1101},
  year={2024},
  publisher={Institute of Mathematical Statistics}
}

@article{hagrass-gof-test,
  title={Spectral regularized kernel goodness-of-fit tests},
  author={Hagrass, Omar and Sriperumbudur, Bharath K and Li, Bing},
  journal={Journal of Machine Learning Research},
  volume={25},
  number={309},
  pages={1--52},
  year={2024}
  }

@article{balasubramanian2021optimality,
  title={On the optimality of kernel-embedding based goodness-of-fit tests},
  author={Balasubramanian, Krishnakumar and Li, Tong and Yuan, Ming},
  journal={Journal of Machine Learning Research},
  volume={22},
  number={1},
  pages={1--45},
  year={2021}
}

@article{schrab2023mmd,
  title={{MMD} aggregated two-sample test},
  author={Schrab, Antonin and Kim, Ilmun and Albert, M{\'e}lisande and Laurent, B{\'e}atrice and Guedj, Benjamin and Gretton, Arthur},
  journal={Journal of Machine Learning Research},
  volume={24},
  number={194},
  pages={1--81},
  year={2023}
}

@article{li2019optimality,
 title={On the Optimality of Gaussian Kernel Based Nonparametric Tests against Smooth Alternatives},
 author={Li, Tong and Yuan, Ming},
 journal={Journal of Machine Learning Research},
 volume={25},
 number={334},
 pages={1--62},
 year={2024}
}

@article{bauer2007regularization,
  title={On regularization algorithms in learning theory},
  author={Bauer, Frank and Pereverzev, Sergei and Rosasco, Lorenzo},
  journal={Journal of Complexity},
  volume={23},
  number={1},
  pages={52--72},
  year={2007},
  publisher={Elsevier}
}

@book{engl1996regularization,
  title={Regularization of inverse problems},
  author={Engl, Heinz Werner and Hanke, Martin and Neubauer, Andreas},
  volume={375},
  year={1996},
  publisher={Springer Science \& Business Media}
}

@article{gerfo2008spectral,
  title={Spectral algorithms for supervised learning},
  author={Gerfo, L Lo and Rosasco, Lorenzo and Odone, Francesca and Vito, E De and Verri, Alessandro},
  journal={Neural Computation},
  volume={20},
  number={7},
  pages={1873--1897},
  year={2008},
  publisher={MIT Press One Rogers Street, Cambridge, MA 02142-1209, USA journals-info~…}
}

@article{sriperumbudur2022approximate,
  title={Approximate kernel {PCA}: Computational versus statistical trade-off},
  author={Sriperumbudur, Bharath K and Sterge, Nicholas},
  journal={The Annals of Statistics},
  volume={50},
  number={5},
  pages={2713--2736},
  year={2022},
  publisher={Institute of Mathematical Statistics}
}

@article{caponnetto2007optimal,
  title={Optimal rates for the regularized least-squares algorithm},
  author={Caponnetto, Andrea and De Vito, Ernesto},
  journal={Foundations of Computational Mathematics},
  volume={7},
  pages={331--368},
  year={2007},
  publisher={Springer}
}

@book{bhatia2013matrix,
  title={Matrix analysis},
  author={Bhatia, Rajendra},
  volume={169},
  year={2013},
  publisher={Springer Science \& Business Media}
}

@article{lin2020distributed,
  title={Distributed kernel ridge regression with communications},
  author={Lin, Shao-Bo and Wang, Di and Zhou, Ding-Xuan},
  journal={Journal of Machine Learning Research},
  volume={21},
  number={93},
  pages={1--38},
  year={2020}
}

@article{guo2017learning,
  title={Learning theory of distributed spectral algorithms},
  author={Guo, Zheng-Chu and Lin, Shao-Bo and Zhou, Ding-Xuan},
  journal={Inverse Problems},
  volume={33},
  number={7},
  pages={074009},
  year={2017},
  publisher={IOP Publishing}
}

@article{lin2017distributed,
	title={Distributed learning with regularized least squares},
	author={Lin, Shao-Bo and Guo, Xin and Zhou, Ding-Xuan},
	journal={Journal of Machine Learning Research},
	volume={18},
	number={92},
	pages={1--31},
	year={2017}
}

@article{szekely2004testing,
	title={Testing for equal distributions in high dimension},
	author={Sz{\'e}kely, G{\'a}bor J and Rizzo, Maria L and others},
	journal={InterStat},
	volume={5},
	number={16.10},
	pages={1249--1272},
	year={2004},
	publisher={Citeseer}
}

@article{gong2014assessing,
  title={Assessing the goodness of fit of personal risk models},
  author={Gong, Gail and Quante, Anne S and Terry, Mary Beth and Whittemore, Alice S},
  journal={Statistics in Medicine},
  volume={33},
  number={18},
  pages={3179--3190},
  year={2014},
  publisher={Wiley Online Library}
}

@article{ritchey1986application,
  title={An application of the chi-squared goodness-of-fit test to discrete common stock returns},
  author={Ritchey, Robert J},
  journal={Journal of Business \& Economic Statistics},
  volume={4},
  number={2},
  pages={243--254},
  year={1986},
  publisher={Taylor \& Francis}
}

@article{frezza2014goodness,
  title={Goodness of fit assessment for a fractal model of stock markets},
  author={Frezza, Massimiliano},
  journal={Chaos, Solitons \& Fractals},
  volume={66},
  pages={41--50},
  year={2014},
  publisher={Elsevier}
}

@article{schermelleh2003evaluating,
  title={Evaluating the fit of structural equation models: Tests of significance and descriptive goodness-of-fit measures},
  author={Schermelleh-Engel, Karin and Moosbrugger, Helfried and M{\"u}ller, Hans and others},
  journal={Methods of Psychological Research Online},
  volume={8},
  number={2},
  pages={23--74},
  year={2003}
}

@article{massey1951kolmogorov,
  title={The {Kolmogorov-Smirnov} test for goodness of fit},
  author={Massey Jr, Frank J},
  journal={Journal of the American statistical Association},
  volume={46},
  number={253},
  pages={68--78},
  year={1951},
  publisher={Taylor \& Francis}
}

@article{fromont2013two,
  author = {Magalie Fromont and B{\'e}atrice Laurent and Patricia Reynaud-Bouret},
  title = {{The two-sample problem for Poisson processes: Adaptive tests with a nonasymptotic wild bootstrap approach}},
  volume = {41},
  journal = {The Annals of Statistics},
  number = {3},
  publisher = {Institute of Mathematical Statistics},
  pages = {1431 -- 1461},
  year = {2013}
}

@inproceedings{NIPS2012dbe272ba,
title = {Optimal kernel choice for large-scale two-sample tests},
 author = {Gretton, Arthur and Sejdinovic, Dino and Strathmann, Heiko and Balakrishnan, Sivaraman and Pontil, Massimiliano and Fukumizu, Kenji and Sriperumbudur, Bharath K.},
 booktitle = {Advances in Neural Information Processing Systems},
 pages = {},
 volume = {25},
 year = {2012}
}

@inproceedings{liu2020learning,
  title={Learning deep kernels for non-parametric two-sample tests},
  author={Liu, Feng and Xu, Wenkai and Lu, Jie and Zhang, Guangquan and Gretton, Arthur and Sutherland, Danica J},
  booktitle={International Conference on Machine Learning},
  pages={6316--6326},
  year={2020},
  organization={PMLR}
}

@article{ingster1987minimax,
  title={Minimax testing of nonparametric hypotheses on a distribution density in the $L_p$ metrics},
  author={Ingster, Yu I},
  journal={Theory of Probability \& Its Applications},
  volume={31},
  number={2},
  pages={333--337},
  year={1987},
  publisher={SIAM}
}

@article{ingster1993asymptotically,
  title={Asymptotically minimax hypothesis testing for nonparametric alternatives. {I, II, III}},
  author={Ingster, Yuri I},
  journal={Mathematical Methods of Statistics},
  volume={2},
  number={2},
  pages={85--114},
  year={1993}
}

@article{blanchard2019lepskii,
  title={Lepskii principle in supervised learning},
  author={Blanchard, Gilles and Math{\'e}, Peter and M{\"u}cke, Nicole},
  journal={arXiv preprint arXiv:1905.10764},
  year={2019}
}

@article{cramer1928composition,
  title={On the composition of elementary errors: First paper: Mathematical deductions},
  author={Cram{\'e}r, Harald},
  journal={Scandinavian Actuarial Journal},
  volume={1928},
  number={1},
  pages={13--74},
  year={1928},
  publisher={Taylor \& Francis}
}

@article{dempster1958high,
  title={A high dimensional two sample significance test},
  author={Dempster, Arthur P},
  journal={The Annals of Mathematical Statistics},
  pages={995--1010},
  year={1958},
  publisher={JSTOR}
}

@article{wang2024learning,
	title={Learning with centered reproducing kernels},
	author={Wang, Chendi and Guo, Xin and Wu, Qiang},
	journal={Analysis and Applications},
	volume={22},
	number={03},
	pages={507--534},
	year={2024},
	publisher={World Scientific}
}

@article{kim2022minimax,
	title={Minimax optimality of permutation tests},
	author={Kim, Ilmun and Balakrishnan, Sivaraman and Wasserman, Larry},
	journal={The Annals of Statistics},
	volume={50},
	number={1},
	pages={225--251},
	year={2022},
	publisher={Institute of Mathematical Statistics}
}

@article{blanchard2016convergence,
  title={Convergence rates of kernel conjugate gradient for random design regression},
  author={Blanchard, Gilles and Kr{\"a}mer, Nicole},
  journal={Analysis and Applications},
  volume={14},
  number={06},
  pages={763--794},
  year={2016},
  publisher={World Scientific}
}

@article{Dicker2017,
author = {Lee H. Dicker and Dean P. Foster and Daniel Hsu},
title = {{Kernel ridge vs. principal component regression: Minimax bounds and the qualification of regularization operators}},
volume = {11},
journal = {Electronic Journal of Statistics},
number = {1},
publisher = {Institute of Mathematical Statistics and Bernoulli Society},
pages = {1022--1047},
year = {2017},
doi = {10.1214/17-EJS1258},
URL = {https://doi.org/10.1214/17-EJS1258}
}

@phdthesis{balazs2005regular,
  author = {Balazs, Peter},
  title = {Regular and irregular Gabor multipliers with application to psychoacoustic masking},
  year = {2005},
school = {University of Vienna},
  type = {{Ph.D.} thesis},
  url = {https://www.oeaw.ac.at/fileadmin/Institute/ISF/PDF/Team/2005_peter_balazs_dissertation.pdf}
}

@article{shebrawi2013trace,
  title={Trace inequalities for matrices},
  author={Shebrawi, Khalid and Albadawi, Hussien},
  journal={Bulletin of the Australian Mathematical Society},
  volume={87},
  number={1},
  pages={139--148},
  year={2013},
  publisher={Cambridge University Press}
}

@article{sriperumbudur2011universality,
  title={Universality, Characteristic Kernels and {RKHS} Embedding of Measures.},
  author={Sriperumbudur, Bharath K and Fukumizu, Kenji and Lanckriet, Gert RG},
  journal={Journal of Machine Learning Research},
  volume={12},
  number={7},
  year={2011},
  pages={2389--2410}
}

@article{simon2018kernel,
  title={Kernel distribution embeddings: Universal kernels, characteristic kernels and kernel metrics on distributions},
  author={Simon-Gabriel, Carl-Johann and Sch{\"o}lkopf, Bernhard},
  journal={Journal of Machine Learning Research},
  volume={19},
  number={44},
  pages={1--29},
  year={2018}
}

@article{raissi2018testing,
 title = {Testing normality for unconditionally heteroscedastic macroeconomic variables},
journal = {Economic Modelling},
volume = {70},
pages = {140--146},
year = {2018},
author = {Hamdi Ra{\"i}ssi},
}

@article{DKW2,
author = {P. Massart},
title = {The Tight Constant in the {Dvoretzky-Kiefer-Wolfowitz} Inequality},
volume = {18},
journal = {The Annals of Probability},
number = {3},
publisher = {Institute of Mathematical Statistics},
pages = {1269 -- 1283},
year = {1990}
}

@article{DKW1,
author = {A. Dvoretzky and J. Kiefer and J. Wolfowitz},
title = {Asymptotic Minimax Character of the Sample Distribution Function and of the Classical Multinomial Estimator},
volume = {27},
journal = {The Annals of Mathematical Statistics},
number = {3},
publisher = {Institute of Mathematical Statistics},
pages = {642 -- 669},
year = {1956}
}

\end{document}